\documentclass[aps,prd,preprint,nofootinbib,superscriptaddress,longbibliography,floatfix]{revtex4-2}
\usepackage{amsmath,amssymb,amsthm}
\usepackage{graphicx}
\usepackage{booktabs}
\usepackage{tikz}
\usetikzlibrary{positioning,arrows.meta}
\usepackage{hyperref}
\hypersetup{colorlinks=true,linkcolor=blue!60!black,citecolor=blue!60!black,urlcolor=blue!60!black}
\graphicspath{{figs/}{training_data/}{referee_response/}{./}}

\newcommand{\C}{\mathcal{C}}
\newcommand{\rvw}{R^2_{\rm vw}}
\newcommand{\relL}{\mathrm{rel}L^2}

\newtheorem{theorem}{Theorem}
\newtheorem{proposition}{Proposition}
\newtheorem{lemma}{Lemma}
\newtheorem{corollary}{Corollary}
\newtheorem{definition}{Definition}
\theoremstyle{remark}
\newtheorem{remark}{Remark}

\begin{document}

\title{The Information Content of Krylov Observables: A Machine Learning Approach}

\author{Ritam Basu}
\email{ritam.basu@tifr.res.in}
\affiliation{Department of Theoretical Physics, Tata Institute of Fundamental Research, Homi Bhabha Road, Mumbai 400005, India}

\begin{abstract}
We employ machine learning techniques to estimate the amount of information contained in three Krylov-space observables: the spread complexity $C(t)$, the discrete Wigner negativity $N(t)$, and the normalized negativity $\chi(t)=N(t)/|S(t)|$, where $S(t)$ is the survival amplitude. We study thermofield-double evolutions on GUE, GOE and Poisson spectra, and on the integrable $SL(2,R)$ Krylov chain of a two-dimensional CFT, primary with Lanczos coefficients $b_n=\alpha\sqrt{n(n+2h-1)}$. To introduce chaos we add a GUE random matrix to the integrable chain, $H(\varepsilon)=H_{SL(2,R)}+\varepsilon R_0 W_{GUE}$, where $R_0$ is its spectral half-width. Either $C(t)$ or $N(t)$ alone determines the thermofield temperature, with $R^2\simeq0.999$. However, neither can reconstruct the fine structure of the spectral form factor. In the integrable sector all three observables carry the same information. When we turn on chaos an asymmetry appears: $\chi$ predicts $C$ better than $C$ predicts $\chi$. This gap grows from $+0.33$ to $+0.77$, while the same gap for the raw negativity decays to zero. Thus we conclude the extra information comes from the survival amplitude and $\chi(t)$ carries it while $C(t)$ can still be recovered from $\chi(t)$. We derive a lower bound on this gap, set by the fraction of the variance of $\chi$ that comes from $S$.
\end{abstract}

\maketitle
\newpage
\section{Introduction}
\label{sec:intro}

The Krylov description of quantum dynamics compresses the evolution of a state
into a one-dimensional hopping problem. Given an initial state $|K_0\rangle$
and a Hamiltonian $H$, the Lanczos recursion produces an orthonormal chain
$\{|K_n\rangle\}$ on which $H$ is tridiagonal, and the full dynamical content
of the evolution is the chain wavefunction
$\psi_n(t)=\langle K_n|e^{-iHt}|K_0\rangle$. Its first moment, the
\emph{spread complexity}
\begin{equation}
\C(t)\;=\;\sum_n n\,|\psi_n(t)|^2 ,
\label{eq:spread}
\end{equation}
introduced in \cite{Balasubramanian:2022tpr} following the operator-growth
programme of \cite{Parker:2018yvk,Barbon:2019wsy,Rabinovici:2020ryf,
Dymarsky:2019elm}, has become a central diagnostic of quantum chaos
\cite{Dymarsky:2021bjq,Rabinovici:2021qqt,Rabinovici:2023yex,Caputa:2021sib,
Avdoshkin:2022xuw,Erdmenger:2023wjg,Balasubramanian:2023mag} (see
\cite{Nandy:2024review} for a review), and --- through the
complexity--geometry correspondence
\cite{Susskind:2014rva,Stanford:2014jda,Brown:2015bva} and its quantitative
Krylov realisation in low-dimensional holography
\cite{Rabinovici:2023yex,Iliesiu:2021ari,Heller:2024dss} --- of black-hole
interiors.

The first moment is, however, a drastic compression of $\psi_n(t)$. In
\cite{Basu:2024krw,Basu:2025wig} we introduced phase-space diagnostics of the
same dynamics: the discrete Wigner function of the chain wavefunction and its
integrated negativity $N(t)$, a genuinely second-moment quantity; and in
\cite{Basu:2026inf} the \emph{normalised negativity}
\begin{equation}
\chi(t)\;=\;\frac{N(t)}{|S(t)|},\qquad S(t)=\psi_0(t),
\label{eq:chi}
\end{equation}
which we argued is the correct second-moment probe of infall in AdS$_3$: the
division by the survival amplitude strips the trivial decay of the return
probability and exposes genuine phase-space spreading.

These constructions raise a question that is sharp, operational, and --- to our
knowledge --- unaddressed: \emph{how much usable information do these curves
actually carry, about the theory that produced them and about one another?}
Concretely: does $\C(t)$ determine $N(t)$, or does the negativity know things
the complexity has discarded, and if so \emph{when}? Is the temperature of a
thermofield-double (TFD) state recoverable from a single curve? Which features
of the spectral form factor (SFF) survive the compression to one curve, and
which are irretrievably lost? Can a single complexity curve identify the
symmetry class of its parent ensemble? And, centrally: is the informational
surplus of the (normalised) negativity a property of \emph{chaotic} dynamics
specifically, as the physical picture of \cite{Basu:2026inf} suggests?

Questions about ``who determines whom'' are questions about the existence of
functions between data sets, and supervised learning is the natural
instrument \cite{Carleo:2019okv,Carrasquilla:2017,MehtaML:2018}: a
regression $X\to Y$ can only succeed if the information about $Y$ is present
in $X$ in learnable form, so test-set scores are operational lower bounds on
mutual information content. Machine learning has previously been used to
diagnose quantum chaos from dynamical data \cite{Kharkov:2019mqc}; a first
step in the Krylov context was taken in \cite{Bak:2025mlq}, which trained
convolutional networks on the \emph{forward} map, state $\to$ complexity.
Here we study the \emph{inverse} and \emph{cross} maps --- from the observable
curves back to physical labels, spectral data, and each other --- across an
integrable-to-chaotic panorama.

\paragraph{Strategy.} We generate five large labelled datasets
(Table~\ref{tab:data}): thermofield-double evolutions over the GUE, GOE and
Poisson spectral ensembles across up to nineteen Hilbert-space dimensions
$D\in[101,1009]$ and seven temperatures; the exactly solvable
SL$(2,\mathbb R)$ discrete-series (2d-CFT primary) sector over a dense grid of
$2{,}079$ theories $(h,\alpha)$; and a one-parameter family
\begin{equation}
H(\varepsilon)\;=\;H_{{\rm SL}(2,\mathbb R)}\;+\;\varepsilon\,R_0\,
W_{\rm GUE},
\label{eq:interp0}
\end{equation}
interpolating from the integrable CFT chain to random-matrix chaos, with the
chain's own level statistics $\langle r\rangle$ recorded per sample as an
internal chaos certificate. On these we train a fixed, deliberately
\emph{minimal} model class --- residual multilayer perceptrons of width
$16$--$32$ blended with gradient-boosted trees --- under a protocol hardened
against the failure modes we encountered and document (metric artifacts,
dynamic-range pathologies, selection bias, test leakage;
Sec.~\ref{sec:protocol}).

\paragraph{Summary of findings.} Our results, with the sections where they are
established:

\begin{enumerate}
\item \textbf{Temperature is fully encoded in either moment}
(Sec.~\ref{sec:gue}, \ref{sec:goepoisson}). $\C\to\beta$ and $N\to\beta$ reach
$R^2=0.999$ and $0.998$ on GUE; $0.998$ on GOE; $0.983$ and $0.997$ on
Poisson. The residual Poisson deficit reflects the larger sample-to-sample
fluctuations of uncorrelated levels, not a failure of encoding: extending the
Poisson grid from $D\le701$ to $D\le1009$ raises the score from $0.979$ to
$0.983$, exactly as fluctuation shrinkage predicts.
\item \textbf{The fine SFF is strictly finer information than either moment}
(Sec.~\ref{sec:sff}). $\C\to{\rm SFF}$ attains $\rvw=0.185$ (GUE), $0.177$
(GOE), $0.184$ (Poisson) --- statistically indistinguishable across
ensembles --- with the reconstruction reproducing the dip and the smooth
envelope while missing every sample-specific fluctuation;
$N\to{\rm SFF}$ is statistically identical ($0.188$ on GUE). Moderate
time-averaging of the SFF target does not materially repair this, and we argue
that windows wide enough to help necessarily erase the dip--ramp structure
itself, trivialising the task: the failure is physical, not statistical.
\item \textbf{The $e^{S}$ plateau is latent in the early complexity curve}
(Sec.~\ref{sec:sff}). The late-time SFF plateau is predicted from $\C(t)$ at
$R^2=0.861$, and at $0.854$ from only the first fifth of the curve; the same
holds on GOE ($\approx0.866/0.859$) and Poisson ($\approx0.836/0.83$).
\item \textbf{One curve identifies the universality class}
(Sec.~\ref{sec:classify}). Per-dimension classifiers on $\C(t)$ alone separate
GUE from Poisson essentially perfectly at every $D$, and GUE from GOE ---
identical mean densities, differing only in level repulsion --- with accuracy
rising from $\approx0.79$ at $D=101$ to $\approx0.97$--$0.98$ at
$D\sim10^{3}$; the negativity supports the same classification slightly less
accurately.
\item \textbf{The integrable CFT sector is informationally symmetric, with a
degeneracy that the negativity resolves} (Sec.~\ref{sec:cft}). Every channel
predicts the composite $h\alpha^2$ perfectly ($R^2=1.000$) --- the
leading-order content of all the curves --- while the individual labels are
resolved best by the negativity ($N\to h=0.999$, $N\to\alpha=0.998$) ahead of
the complexity ($0.984,0.985$) and $\chi$ ($0.974,0.974$). The
$\C\leftrightarrow N$ reconstruction asymmetry is compatible with zero, as
demanded by the exact negative-binomial slaving of the Krylov moments in this
sector.
\item \textbf{Headline: the $\chi$-surplus switches on with chaos}
(Sec.~\ref{sec:chaos}). Along \eqref{eq:interp0} the seed-averaged asymmetry
gap $R^2(\chi\to\C)-R^2(\C\to\chi)$ grows monotonically from $+0.33$ at
$\varepsilon=0.005$ to $+0.77$ at $\varepsilon\simeq0.12$ --- exactly the
point where $\langle r\rangle$ first reaches the GUE plateau --- and
saturates near $0.6$--$0.7$; the raw-$N$ gap instead decays from
$+0.22$ to zero. Every point carries agreement between two independent error
metrics. Consistently, on pure GUE data the small raw-$N$ gap ($+0.019$) is
removed by target smoothing (Appendix~\ref{sec:smooth}): the raw-$N$ asymmetry is
noise-borne, and the robust chaotic-regime surplus is a property of the
\emph{normalised} negativity.
\end{enumerate}

Finding 6 is the operational content of the proposal of \cite{Basu:2026inf}:
in the integrable sector the Krylov moment tower is slaved and first- and
second-moment observables are informationally equivalent; as the dynamics
becomes chaotic the tower decouples, and the second-moment information ---
made visible by normalising with the survival amplitude --- becomes
inaccessible to the first moment.

A methodological thread runs alongside the physics. All results are obtained
with very small networks, survive halving the training data, and are
insensitive to width; several are protected by a dual-metric agreement rule
and seed averaging that we adopted after explicitly tracing spurious
asymmetries to metric and dynamic-range artifacts. We document these audits
(Secs.~\ref{sec:protocol}, \ref{sec:cftcurves}) because comparative
machine-learning claims in physics are only as strong as their protocol.

The paper is organised as follows. Section~\ref{sec:setup} defines the
observables, the datasets, and the fast numerical pipeline.
Section~\ref{sec:ml} specifies the architecture, its robustness, and the
evaluation protocol. Section~\ref{sec:gue} presents the complete
random-matrix programme --- the GUE baseline, the SFF lossiness and
information budget, the GOE and Poisson replications, and universality-class
discrimination. Section~\ref{sec:cft} develops the integrable CFT sector ---
theory, degeneracy analysis, curve reconstructions, and limitations.
Section~\ref{sec:chaos} presents the chaos interpolation and the headline
result, and Section~\ref{sec:mech} derives its analytical mechanism in
theorem--proof form. Section~\ref{sec:disc} distils where each observable is
most useful and lists open directions. Appendices collect supplementary
figures, the smoothing protocols, and numerical details.

\section{Observables and the learning methodology}
\label{sec:setup}

\subsection{Krylov chain and the four channels}
\label{sec:channels}

For a Hamiltonian $H$ and initial state $|K_0\rangle$, the Lanczos recursion
\begin{equation}
|A_{n+1}\rangle=(H-a_n)|K_n\rangle-b_n|K_{n-1}\rangle,\qquad
a_n=\langle K_n|H|K_n\rangle,\quad b_n=\|A_n\|,
\end{equation}
generates the Krylov basis in which $H$ is tridiagonal and the state evolves
as a wavepacket $\psi_n(t)$ on a semi-infinite (or, for finite systems,
$D$-site) chain. For the random-matrix datasets the initial state is the
thermofield double
\begin{equation}
|\psi_\beta\rangle=\frac1{\sqrt{Z(\beta)}}\sum_n e^{-\beta E_n/2}|n,n\rangle,
\end{equation}
(the canonical bulk dual of the eternal black hole
\cite{Maldacena:2001kr}), whose Krylov problem depends on the spectrum
$\{E_n\}$ alone --- eigenvectors never enter --- so that changing the ensemble is a controlled substitution of
eigenvalue statistics. Numerically we rotate $|\psi_\beta\rangle$ to the first
basis vector by a rank-one Householder reflection and tridiagonalise; the
resulting chain Hamiltonian is real symmetric tridiagonal and is diagonalised
with a dedicated $O(D^2)$ solver.

Per sample we record four time series on a common grid of $n_t=100$ points in
the rescaled time $t=\tau/D$:
\begin{itemize}
\item \textbf{Spread complexity} $\C(t)$, Eq.~\eqref{eq:spread}.
\item \textbf{Discrete Wigner negativity} $N(t)$. For odd prime chain length
$D$ the chain is a discrete phase space and the Wootters--Wigner function
\cite{Wootters:1987,Gibbons:2004hz,Gross:2006} of $|\psi(t)\rangle$ is
defined through the phase-point
operators $A_{x,p}$; the negativity is the integrated negative volume
$N=\frac1D\sum_{x,p}|W_{xp}|$ (normalised so that $N=1$ for the initial basis
state); Wigner negativity is a standard marker of nonclassicality and a
resource for quantum computation
\cite{Kenfack:2004zd,Veitch:2012nqp,Mari:2012ee,Howard:2014mag}, and its
stabilizer-complexity cousin has recently been studied for Hawking radiation
\cite{BasuPaul:2025stab}. We use throughout the identity --- verified
against the direct
construction at the $10^{-13}$ level --- that the readout is an unnormalised
discrete Fourier transform,
\begin{equation}
N(t)\;=\;\frac{1}{D}\sum_{x,k}\Big|\,D\,\big[\mathrm{ifft}_\ell\,
\rho\big]_{xk}\Big|,
\qquad
\rho_{x\ell}\;=\;\psi_{(2x-\ell-1)\ \mathrm{mod}\ D}(t)\,\bar\psi_{\ell}(t),
\label{eq:fftneg}
\end{equation}
which reduces the cost per time point from $O(D^3)$ to $O(D^2\log D)$ and
eliminates the $O(D^3)$-sized phase array from memory. Together with the
rank-one Householder and the tridiagonal eigensolver this accelerates the full
pipeline by two orders of magnitude ($\sim100$--$280\times$ per sample at
$D\simeq200$), which is what makes the $\sim57{,}000$-sample programme of this
paper feasible on commodity hardware.
\item \textbf{Spectral form factor} ${\rm SFF}(t)=|Z(\tau)|^2/D^2$ with
$Z(\tau)=\sum_n e^{-iE_n\tau}$, computed from the sample's spectrum.
\item \textbf{Survival amplitude and normalised negativity}
$S(t)=\psi_0(t)$ and $\chi(t)=N(t)/|S(t)|$, Eq.~\eqref{eq:chi}. We note for
later use that for the TFD at inverse temperature $\beta$, $|S(t)|^2$
coincides with the (filtered) per-sample SFF; this makes some
$\chi$-involving comparisons trivially confounded, and we flag them
explicitly wherever they occur.
\end{itemize}
Every dataset additionally stores, per sample, the mean consecutive-gap ratio
$\langle r\rangle$ \cite{Oganesyan:2007,Atas:2013}; the generated spectra are
required to reproduce the ensemble surmises (GUE $0.5996$, GOE $0.5307$,
Poisson $0.3863$) before any learning is performed. This is our standing
correctness certificate and it is visible in every diagnostics figure below.

\subsection{The five datasets}
\label{sec:datasets}

\begin{table}[t]
\centering\small
\begin{tabular}{@{}lllll@{}}
\toprule
Setting & Swept parameters & Samples & Split & Time grid\\
\midrule
GUE     & $D\!\in\![101,701]$ (13 primes), $\beta\!\in\![0,5]$ (7) &
$13{,}650$ & random $50/50$ & $t\in[0,2]$\\
GOE     & same & $13{,}650$ & random $50/50$ & $t\in[0,2]$\\
Poisson & $D\!\in\![101,1009]$ (19 primes), same $\beta$ & $19{,}950$ &
random $50/50$ & $t\in[0,2]$\\
CFT (SL$(2,\mathbb R)$) & $h\!\in\![0.25,6]$ (231),
$\alpha\!\in\![0.6,1.8]$ (9) & $2{,}079$ & held-out $h$ & $t\in[0,1]$\\
CFT$\to$chaos & $h\!\in\![0.3,4]$, $\alpha\!\in\!\{0.8,1.2,1.6\}$,
$\varepsilon\!\in\![0,0.3]$ & $7{,}680$ & held-out $h$, per $\varepsilon$ &
$t\in[0,1]$\\
\bottomrule
\end{tabular}
\caption{Datasets. All chain dimensions are odd primes, as required by the
Wootters construction and the FFT identity~\eqref{eq:fftneg}. The three
spectral ensembles are width-matched (spectral standard deviation
$\simeq1$), so no classifier can exploit a trivial energy scale. The
interpolation set comprises $16$ balanced strata of $480$ samples each at
chain dimension $D=251$; the Poisson set and the classification runs extend
to $D=1009$.}
\label{tab:data}
\end{table}

The random-matrix sets use $150$ independent Hamiltonian draws per
$(D,\beta)$ block; the Poisson set draws independent levels from the matched
semicircle density by rejection sampling. The CFT set is built directly from
the exact discrete-series Lanczos coefficients (Sec.~\ref{sec:cfttheory});
each sample is verified against the closed forms for $\C(t)$ and $S(t)$ at the
$10^{-8}$ level, and a truncation guard discards any sample whose occupation
of the last chain site exceeds $10^{-10}$ (none were discarded on the quoted
grid at $n_{\max}=503$). The interpolation set is described in
Sec.~\ref{sec:chaosmodel}, including its distinct truncation protocol.

\subsection{Machine-learning methodology}
\label{sec:ml}

\subsubsection{Architecture and training}
\label{sec:arch}

All tasks share one backbone: a residual (skip-connection) multilayer
perceptron \cite{He:2015res}, Fig.~\ref{fig:net}. The input curve ($n_t=100$ values;
concatenated for two-channel inputs) is standardised per feature, projected to
width $W$, passed through three residual blocks
(Dense--BatchNorm--ReLU--Dropout--Dense--BatchNorm, added to the block input
through an identity skip, then ReLU), and mapped linearly to the output ---
one number for scalar regression, $n_t$ numbers for curve reconstruction
(standardised targets, inverted at evaluation). We use $W=16$ for the
random-matrix scans and $W=32$ for the CFT and interpolation studies; dropout
$0.2$--$0.3$; $L^2$ weight decay $10^{-4}$; Adam with initial rate $10^{-3}$;
early stopping on validation loss with best-weight restoration (patience
$30$--$40$) and learning-rate halving on plateau. Wide-dynamic-range channels
are learned logarithmically: $\log_{10}{\rm SFF}$, $\log_{10}\chi$, and ---
in the CFT settings, where $\C$ spans a $\sim500\times$ range over the
parameter grid --- $\log_{10}(\C+10^{-2})$; these are monotone
reparametrisations, not smoothings. For scalar targets a histogram
gradient-boosted regressor is trained alongside the network and the model
(network, trees, or their mean) is selected on a validation slice of the
training data; classification tasks use the gradient-boosted classifier.

\begin{figure}[t]
\centering
\begin{tikzpicture}[>=Stealth, node distance=7mm,
  box/.style={draw, rounded corners, minimum width=25mm, minimum height=8mm,
              font=\small, align=center}]
  \node[box] (in) {input curve\\($n_t$ values)};
  \node[box, right=of in] (proj) {Dense $W$\\ReLU + BN};
  \node[box, right=of proj] (rb) {Residual\\block $\times 3$};
  \node[box, right=of rb] (out) {Dense\\$\to$ output};
  \draw[->] (in)--(proj); \draw[->] (proj)--(rb); \draw[->] (rb)--(out);
\end{tikzpicture}
\hspace{8mm}
\begin{tikzpicture}[>=Stealth, node distance=3.6mm,
  op/.style={draw, rounded corners, minimum width=32mm, minimum height=5mm,
             font=\footnotesize, align=center}]
  \node[op] (x) {$x$ (width $W$)};
  \node[op, below=of x] (d1) {Dense $W$ + ReLU};
  \node[op, below=of d1] (bn1){BatchNorm};
  \node[op, below=of bn1](dr) {Dropout};
  \node[op, below=of dr] (d2) {Dense $W$};
  \node[op, below=of d2] (bn2){BatchNorm};
  \node[circle, draw, below=of bn2, inner sep=1pt] (add) {$+$};
  \node[op, below=of add] (re) {ReLU $\to$ out};
  \draw[->](x)--(d1);\draw[->](d1)--(bn1);\draw[->](bn1)--(dr);
  \draw[->](dr)--(d2);\draw[->](d2)--(bn2);\draw[->](bn2)--(add);\draw[->](add)--(re);
  \draw[->] (x.east) -- ++(2.0,0) |- (add.east)
        node[pos=0.22, right, font=\footnotesize]{skip};
\end{tikzpicture}
\caption{\textbf{Left:} the residual MLP used for every task in this paper,
with $W=16$ (random-matrix scans) or $W=32$ (CFT and interpolation).
\textbf{Right:} one residual block; the identity skip lets each block learn a
correction to its input, which stabilises training at small width.}
\label{fig:net}
\end{figure}
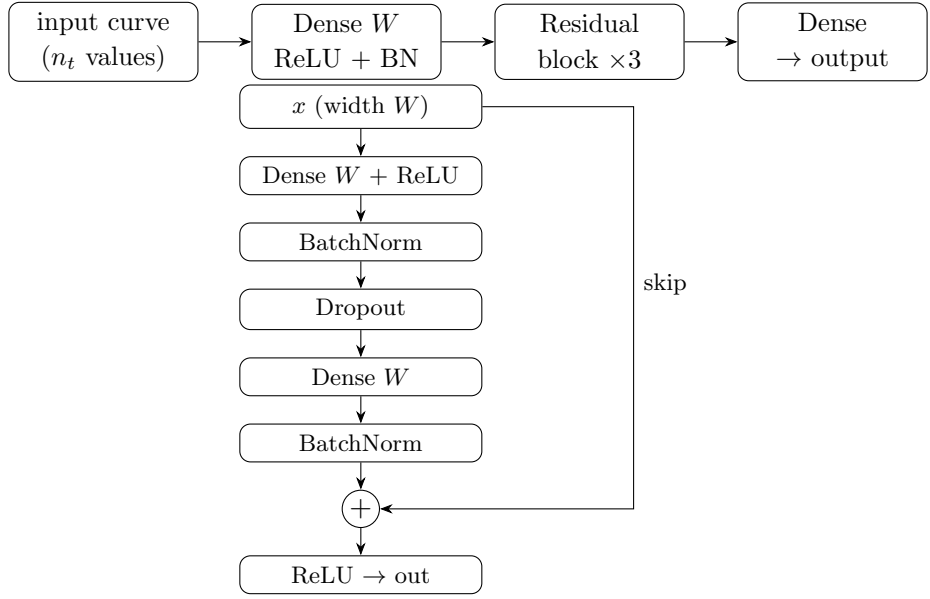

\subsubsection{Smallness and robustness}
\label{sec:robust}

A deliberate methodological choice of this work is model minimality, and we
state its consequences as results in their own right:
\begin{itemize}
\item \emph{Sixteen neurons suffice.} Every random-matrix number quoted in
Sec.~\ref{sec:gue} was obtained at $W=16$.
\item \emph{Half the data suffices.} The same numbers were obtained with the
test fraction raised to $50\%$, i.e.\ training on only half of each dataset;
they are unchanged, at the quoted precision, relative to a standard $70/30$
split.
\item \emph{Width does not matter --- until it hurts.} Rerunning
representative tasks at $W=128$ changes no headline score; in one
small-stratum curve task a $W=128$ model was observed to \emph{destabilise},
inflating an asymmetry to $+0.17$ in a sector where the analytics of
Sec.~\ref{sec:cfttheory} demand zero. Small width is therefore not a
compromise but a safeguard.
\item \emph{No overfitting regime.} Training and validation losses are logged
for every fit and remain close throughout; representative pairs are visible in
the training panels of the scalar-task figures below.
\end{itemize}
The conclusion we draw, and will lean on when interpreting comparisons: the
scores below reflect the information structure of the observables, not the
capacity of the models.

\subsubsection{Evaluation protocol}
\label{sec:protocol}

Several conclusions of this paper are \emph{comparisons} between learning
tasks, and comparisons inherit every bias of their protocol. The following
rules were adopted after explicit audits in which we traced spurious
asymmetries to identifiable artifacts; we record the artifacts alongside the
rules because they are, in our experience, generic to ML-for-physics
comparisons.

\begin{enumerate}
\item \textbf{Honest splits.} Ensemble data (independent draws) use random
splits. Deterministic parameter-grid data use \emph{held-out parameter
values}: a fixed set of $h$ values is excluded from training entirely, so the
test measures generalisation to unseen theories rather than interpolation
between adjacent grid points; in the interpolation study the same held-out
set is used at every $\varepsilon$.
\item \textbf{Dual metrics with an agreement flag.} Every curve task reports
both the variance-weighted coefficient of determination $\rvw$ and the mean
per-sample relative $L^2$ error; every asymmetry gap is quoted in both, and a
gap whose \emph{sign} the two metrics dispute is declared inconclusive.
\emph{Audit:} on raw CFT channels a single-metric evaluation produced an
apparent $\C\!\leftrightarrow\!N$ gap of $+0.28$ whose sign the relative
error reversed; the cause was the $\sim500\times$ dynamic range of the raw
$\C$ channel, under which variance weighting silently ignores small-$h$
samples (Sec.~\ref{sec:cftcurves}).
\item \textbf{Seed averaging.} All interpolation-study gaps are averaged over
independent training seeds and quoted with the spread as an error bar.
\item \textbf{No test leakage.} Model selection among network/trees/blend uses
a validation slice of the training data only; the test set is touched once.
\item \textbf{No selection on physical outcomes.} Sample-retention rules must
not correlate with the labels. \emph{Audit:} an edge-occupation guard,
appropriate for truncated infinite chains, was found to discard preferentially
the most chaotic corner of the interpolation dataset --- a selection bias
correlated with $\varepsilon$ --- and was replaced by the protocol of
Sec.~\ref{sec:chaosmodel}, in which the diagnostic is stored but never
selected on.
\end{enumerate}

\section{Random-matrix ensembles: GUE, GOE, and Poisson}
\label{sec:gue}

\subsection{GUE: the chaotic baseline}

Figure~\ref{fig:diagGUE} shows the GUE dataset at a glance: the
ramp--peak--slope--plateau of $\C(t)/D$ and its ordering in $\beta$; the
smooth monotone growth and saturation of $N(t)/\sqrt D$; the SFF
dip--ramp--plateau; the collapse of $\C/D$ across all thirteen dimensions
after time rescaling; the mean gap ratio pinned to the GUE surmise $0.5996$
for every one of the $91$ $(D,\beta)$ blocks; and the uniform sample counts.
Panel (e) is the nontrivial certificate: any error in ensemble generation,
seeding, or merging would displace it.

\begin{figure}[t]
\centering
\includegraphics[width=\textwidth]{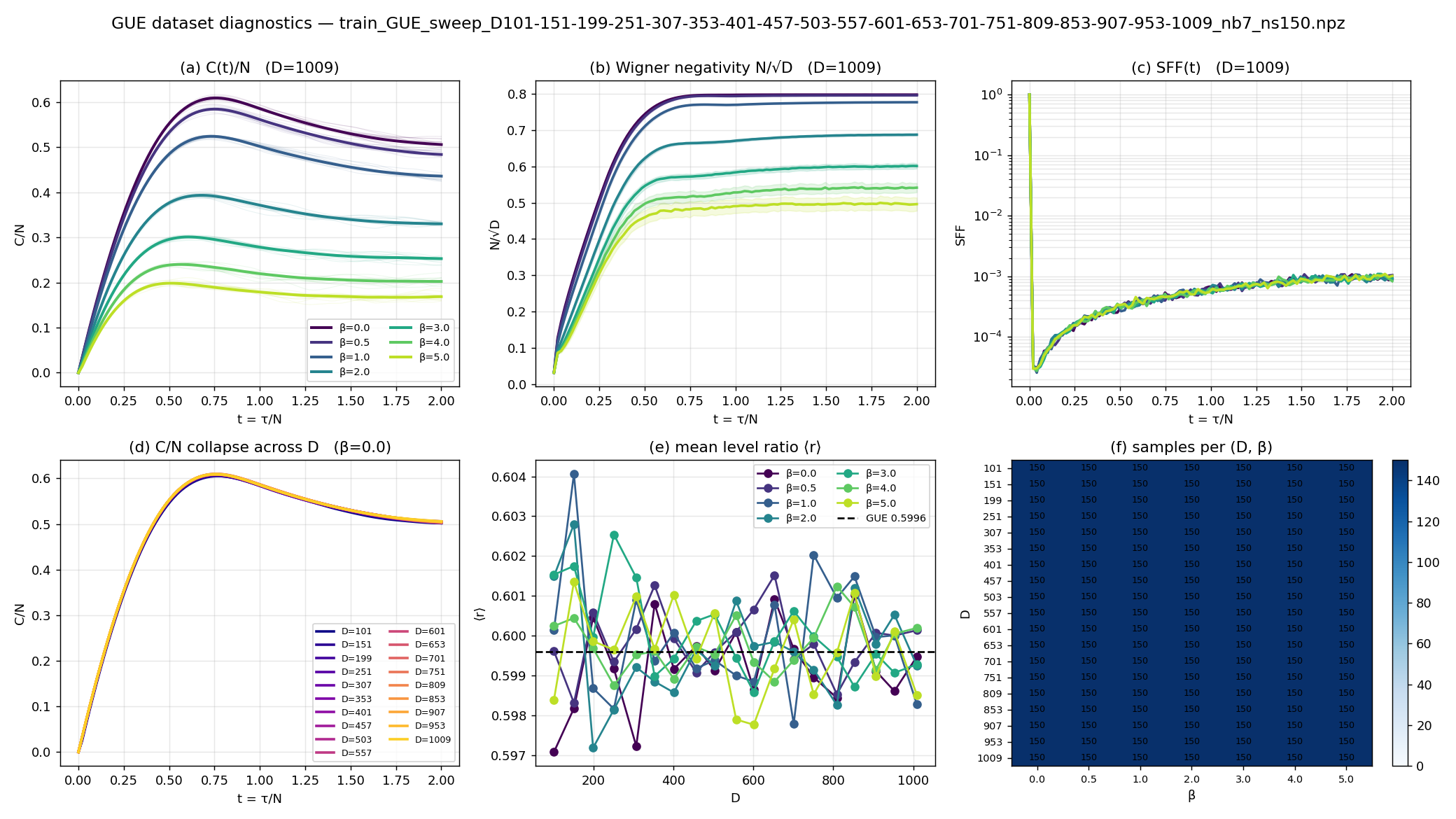}
\caption{GUE dataset diagnostics ($13{,}650$ samples; $D=101$--$701$, seven
temperatures). (a) $\C/D$ per $\beta$ at the largest $D$; (b) $N/\sqrt D$;
(c) SFF; (d) collapse of $\C/D$ across $D$; (e) $\langle r\rangle$ on the GUE
line $0.5996$ for all $(D,\beta)$; (f) samples per block.}
\label{fig:diagGUE}
\end{figure}

\subsubsection{Temperature from a single curve}
\label{sec:guebeta}

The first, calibrating question: does either moment know the temperature of
the TFD state it evolved from? The answer is essentially exact.
$\C(t)\to\beta$ attains $R^2=0.999$ (RMSE $0.059$ over $\beta\in[0,5]$) and
$N(t)\to\beta$ attains $0.998$ ($0.077$); the scatter plots
(Fig.~\ref{fig:guebeta}) show the seven temperatures resolved as tight
vertical clusters, and the training panels show train and validation losses
descending together. Recall the conditions: width-$16$ networks trained on
half the data. Either curve is a near-complete record of $\beta$; the
temperature dependence visible by eye in Fig.~\ref{fig:diagGUE}(a,b) is fully
learnable.

\begin{figure}[t]
\centering
\includegraphics[width=0.9\textwidth]{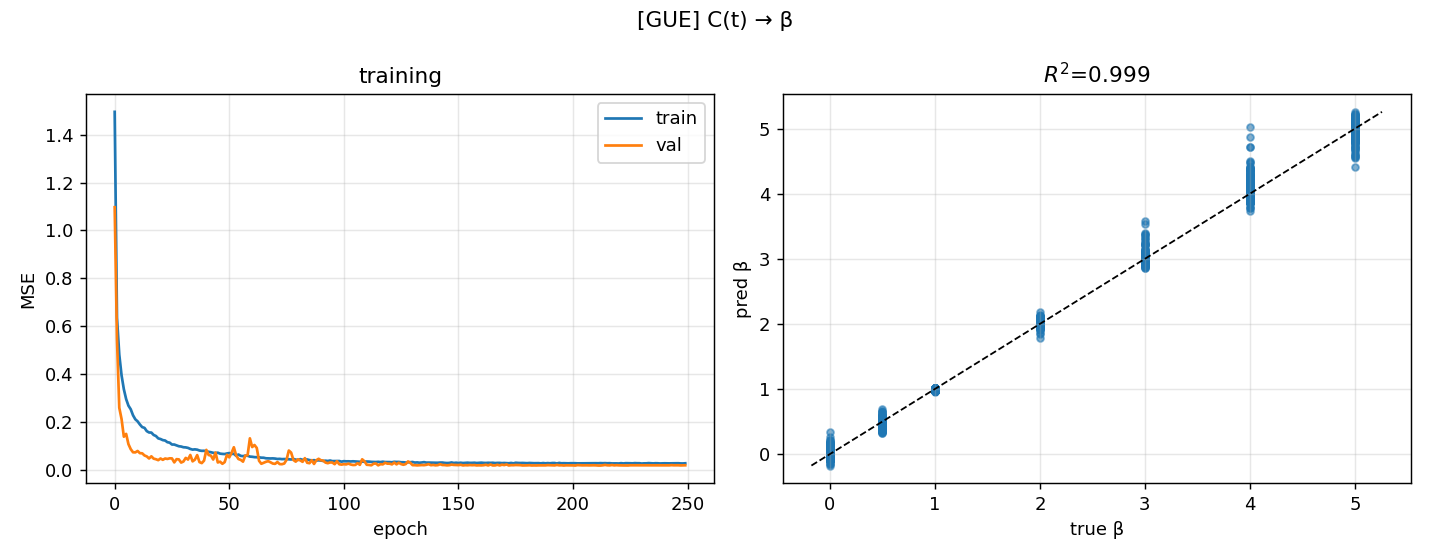}\\[1mm]
\includegraphics[width=0.9\textwidth]{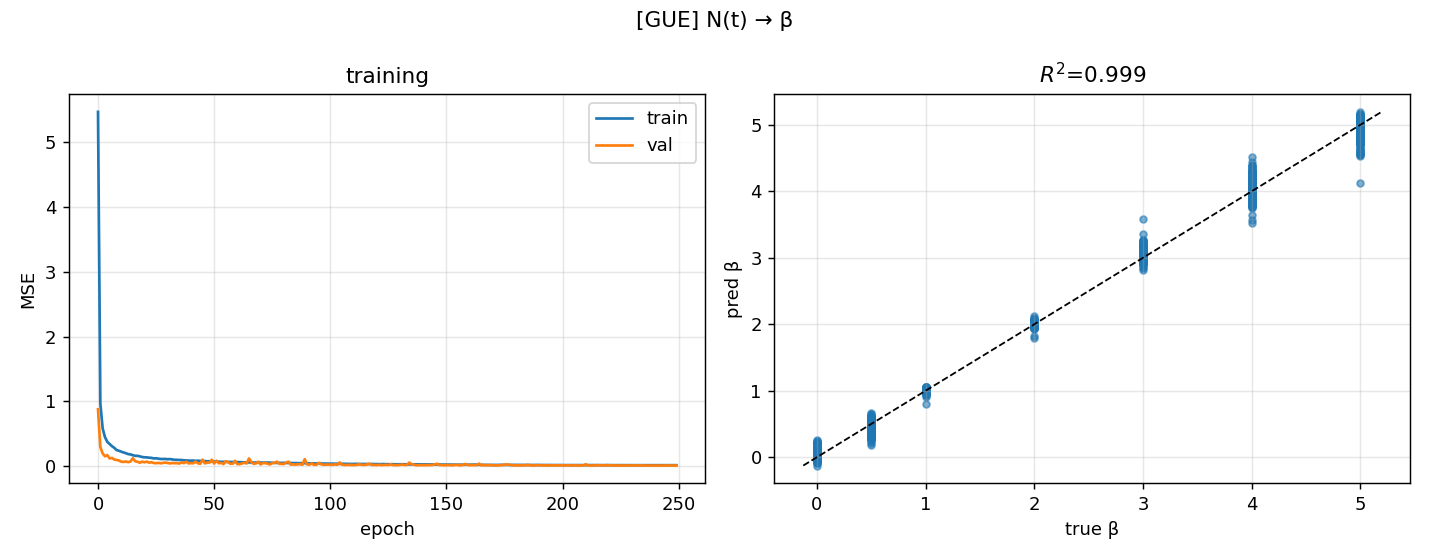}
\caption{GUE temperature regression. \textbf{Top:} $\C(t)\to\beta$,
$R^2=0.999$. \textbf{Bottom:} $N(t)\to\beta$, $R^2=0.998$. Left panels:
train/validation losses (no overfitting); right panels: predicted vs true
$\beta$ on held-out samples.}
\label{fig:guebeta}
\end{figure}

\subsubsection{Mutual reconstruction of the two moments}
\label{sec:guemutual}

Can each moment be rebuilt from the other? To high accuracy, yes: $N\to\C$
attains $\rvw=0.994$ with mean relative error $0.035$, and $\C\to N$ attains
$0.973$ with relative error $0.024$ (Fig.~\ref{fig:gueNC}). Both directions
reproduce the ramp, the peak and the plateau of the target on held-out
samples.

Two remarks, both of which the later sections sharpen. First, the apparent
asymmetry favouring $N\to\C$ in $\rvw$ ($+0.019$) is \emph{reversed} by the
relative-error metric, which favours $\C\to N$; by the protocol of
Sec.~\ref{sec:protocol} the pooled GUE asymmetry is therefore inconclusive
and we do not interpret it. Second, Sec.~\ref{sec:smooth} shows that this
$+0.019$ is removed by target smoothing --- it is carried by the unpredictable
high-frequency component of the raw negativity, i.e.\ it is noise-borne. The
genuine, robust second-moment surplus appears only for the \emph{normalised}
negativity and only in the chaotic regime of the interpolation
(Sec.~\ref{sec:chaos}); the pooled-GUE numbers here are the reason a
controlled interpolation was necessary at all.

\begin{figure}[t]
\centering
\includegraphics[width=\textwidth]{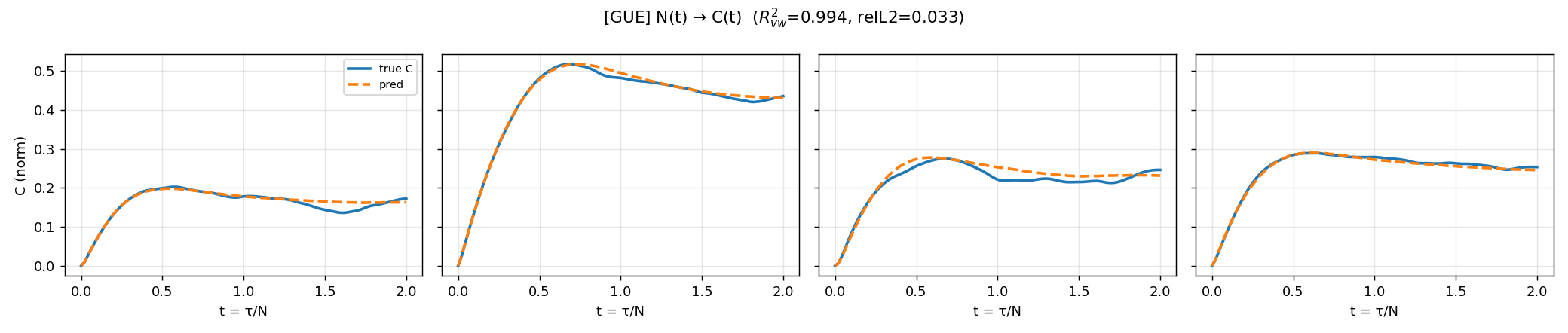}\\[1mm]
\includegraphics[width=\textwidth]{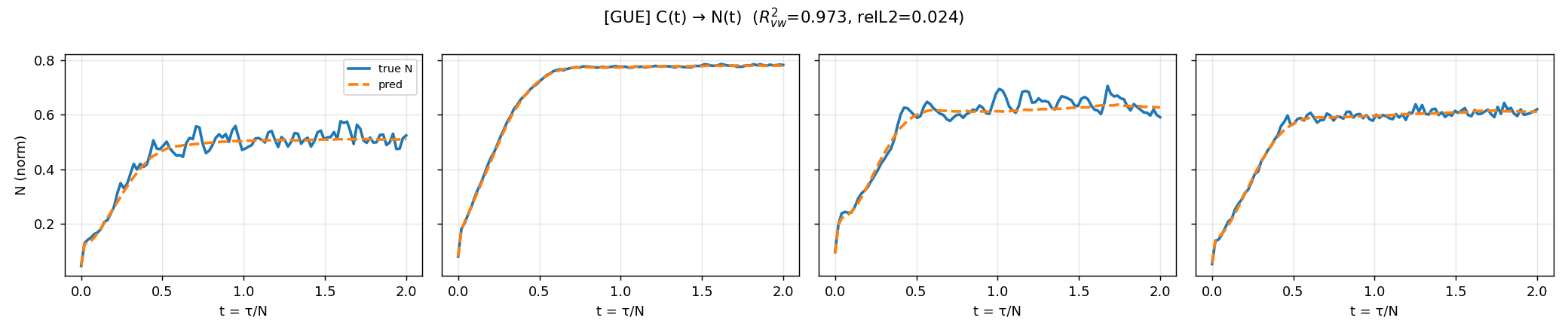}
\caption{GUE curve reconstructions on held-out samples (true solid, prediction
dashed). \textbf{Top:} $N\to\C$, $\rvw=0.994$, $\relL=0.035$.
\textbf{Bottom:} $\C\to N$, $\rvw=0.973$, $\relL=0.024$. The two metrics
order the directions oppositely; per our protocol no surplus is claimed from
this data.}
\label{fig:gueNC}
\end{figure}

\subsection{What survives of the spectral form factor}
\label{sec:sff}

\subsubsection{The fine SFF is not reconstructible}
\label{sec:sfffail}

The SFF of a \emph{single} spectrum is a violently fluctuating object: its
ramp is a sum of cosines of all level differences, and its wiggles encode the
sample-specific fine correlations (rigidity) of the eigenvalues. The
reconstruction task $\C(t)\to\log_{10}{\rm SFF}(t)$ fails in a specific and
instructive way (Fig.~\ref{fig:guesff}): the prediction locks onto the initial
slope, the dip position, and the smooth mean of the ramp and plateau ---
i.e.\ the ensemble-typical envelope conditioned on $(D,\beta)$ --- and is
essentially flat through every sample-specific fluctuation. Quantitatively
$\rvw=0.185$ with $\relL=0.175$ for GUE; the negativity does no better
($N\to{\rm SFF}$: $0.188$, $0.174$), and the same failure occurs on GOE
($0.177$) and Poisson ($0.184$; Sec.~\ref{sec:goepoisson}).

The mechanism is transparent. Both $\C(t)$ and $N(t)$ are smooth low-order
functionals of the Krylov occupation $|\psi_n(t)|^2$; sample-to-sample they
vary by small, slowly varying amounts, while the SFF fluctuations they would
need to resolve are $O(1)$ on the log scale and oscillate on the inverse
mean-level-spacing time. The compression from $\psi_n(t)$ to one or two smooth
moments is exactly a low-pass filter on the spectral information, and the
model correctly learns the only thing the input determines: the conditional
mean. In this precise sense \emph{the SFF strictly refines the information
content of the moments}: the pair $(\C,N)$ is very far from a sufficient
statistic for spectral rigidity.

\begin{figure}[t]
\centering
\includegraphics[width=\textwidth]{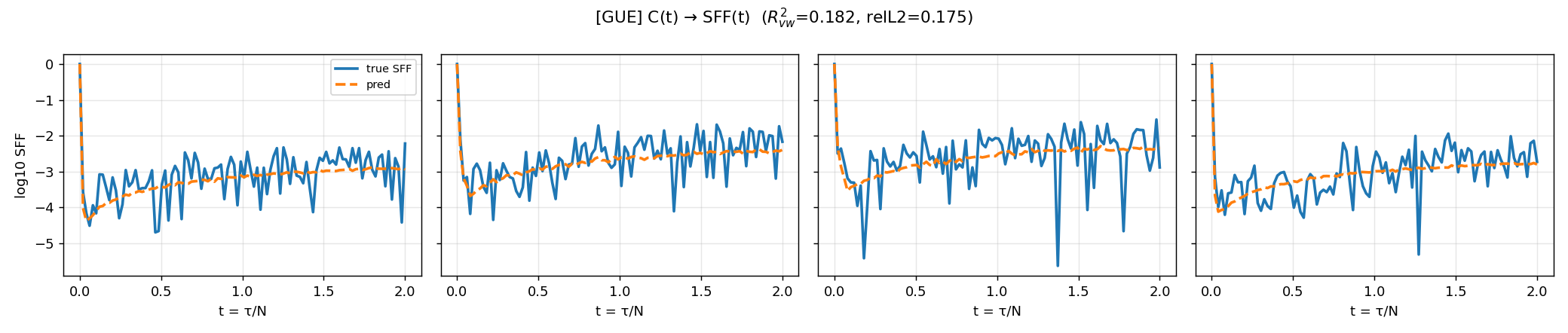}\\[1mm]
\includegraphics[width=\textwidth]{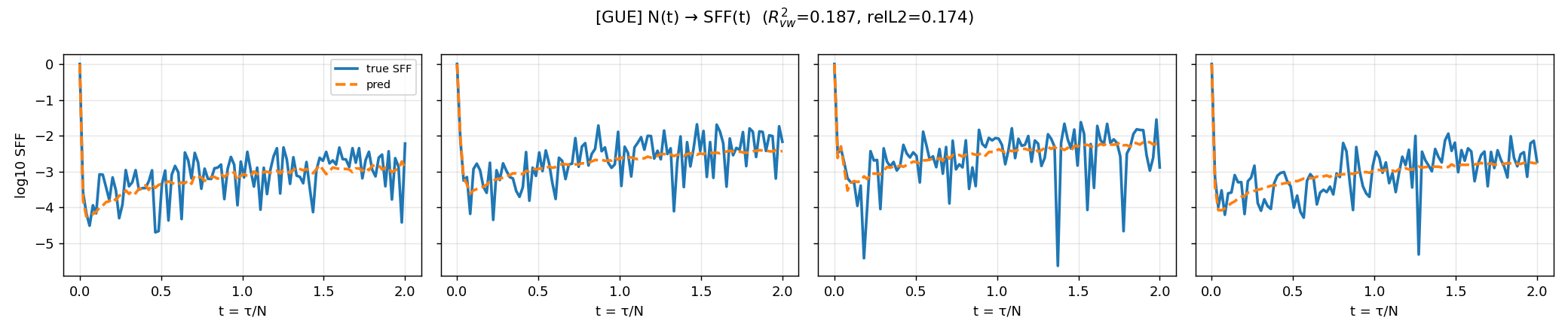}
\caption{The lossiness result on GUE. \textbf{Top:} $\C\to{\rm SFF}$,
$\rvw=0.185$. \textbf{Bottom:} $N\to{\rm SFF}$, $\rvw=0.188$. In every
held-out sample the prediction (dashed) reproduces the dip and the smooth
envelope and none of the sample-specific ramp fluctuations.}
\label{fig:guesff}
\end{figure}

\subsubsection{Smoothing the target does not rescue the task}
\label{sec:sffsmooth}

One may object that the fluctuations are irreducible per-sample noise, and
that the fair target is a locally time-averaged SFF. We tested this directly:
the target $\log_{10}{\rm SFF}$ was replaced by its centred moving average
over windows up to $\Delta t\simeq0.26$ (in $t=\tau/D$ units), the inputs
kept raw, and the tasks retrained per window
(figure \texttt{mlfit\_GUE\_\allowbreak SFF\_\allowbreak smooth\_\allowbreak sweep.png}). The improvement is
modest and gradual --- no window in the physically admissible range restores
the reconstruction to anything resembling the fidelity of the
$\C\!\leftrightarrow\!N$ maps. And the admissible range is bounded by the
physics: the dip sits at $t\lesssim0.1$ and the structured part of the ramp
spans $\Delta t\sim0.1$--$0.5$, so windows much beyond $\Delta t\simeq0.25$
begin to average over the dip--ramp structure itself, and a high score on such
a target would be trivial --- one would be predicting a nearly featureless
curve. We therefore regard the lossiness of Sec.~\ref{sec:sfffail} as
physical: the unpredictable component of the SFF persists at all scales that
retain the SFF's defining features.

\subsubsection{The information budget in time: where the \texorpdfstring{$e^{S}$}{e\textasciicircum S} plateau lives}
\label{sec:infobudget}

The coarse late-time plateau of the SFF --- the $\sim e^{S}$ scale, dual to
the long-time volume plateau of the black-hole interior
\cite{Iliesiu:2021ari} --- is a different story. Predicting the (log) plateau
height from $\C(t)$ succeeds at $R^2=0.861$; and sweeping how much of the
complexity curve the model may see (Fig.~\ref{fig:gueib}) shows $R^2=0.854$
already from the first $20\%$ of the curve, a flat intermediate region, and a
small final increment once $\C(t)$ itself saturates. The non-perturbative
plateau scale is thus largely \emph{latent in the short-time complexity
dynamics}, with only a modest genuinely-late-time top-up --- a quantitative
``information budget in time'' for the $e^{S}$ scale. The same experiment on
GOE gives $\approx0.866$ (full) versus $\approx0.859$ (early), and on Poisson
$\approx0.836$ versus $\approx0.83$ (Sec.~\ref{sec:goepoisson}): the
budget's shape is ensemble-universal; the ceilings of the two rigid
ensembles coincide within our precision and lie above the Poisson one.

\begin{figure}[t]
\centering
\includegraphics[width=0.72\textwidth]{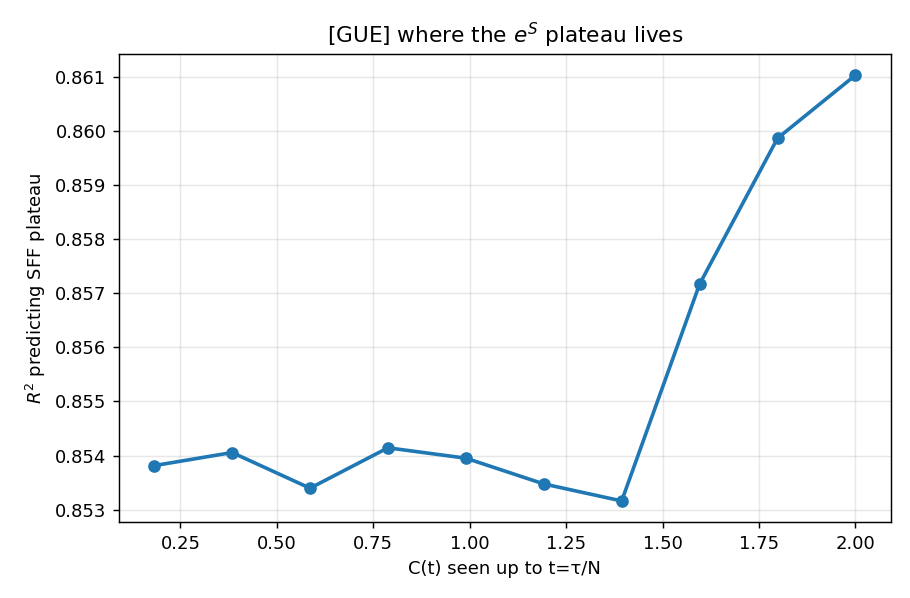}
\caption{GUE information budget: $R^2$ for predicting the late-time SFF
plateau versus the time up to which the model may see $\C(t)$. The first
fifth of the curve already delivers $0.854$ of the full-curve $0.861$; note
the deliberately zoomed vertical scale.}
\label{fig:gueib}
\end{figure}

\subsubsection{Routing through the chain: the Lanczos ceiling}
\label{sec:lanczos}

The observables of this paper are functionals of the Krylov chain data: the
tridiagonal pair $\{a_n,b_n\}$ determines the chain Hamiltonian exactly,
hence the spectrum, hence --- together with the initial vector --- every
channel we study. The chain data is therefore a \emph{lossless} encoding:
$R^2_\star(\{a_n,b_n\}\!\to\!\cdot)=1$ in principle for every target. This
suggests routing the information flow through the chain explicitly, with
four reconstruction tasks at fixed $D$ (the dimension of $\{a_n,b_n\}$
varies with $D$, so a single-dimension block is used) plus one ceiling task:
\begin{itemize}
\item $\C\to\{a_n,b_n\}$ and $N\to\{a_n,b_n\}$: do the smooth moments
determine the microscopic chain couplings that generated them?
\item $\{a_n,b_n\}\to\C$ and $\{a_n,b_n\}\to N$: the forward direction ---
in-principle exact --- tests how much of a lossless encoding a small network
can actually decode.
\item $\{a_n,b_n\}\to\log{\rm SFF}$: the \emph{ceiling}. Since the chain
data determines the spectrum exactly, this task separates the two possible
readings of the lossiness result of Sec.~\ref{sec:sfffail}: a low score here
would mean that part of the $\C\to{\rm SFF}$ failure is a decoding
limitation of the model class; a high score confirms that the failure is
information genuinely absent from the moments, since the same network
class recovers the SFF once the lossless encoding is supplied.
\end{itemize}
Figure~\ref{fig:lanczos} presents the five tasks on a single-dimension GUE
block ($D=1009$, all temperatures, $1{,}050$ samples) under the standard
protocol; the companion script recomputes $\{a_n,b_n\}$ per sample from the
stored seeds and certifies the reconstruction against the stored curves
(agreement at the $10^{-13}$ level) before any training.

The outcome is decisive on the ceiling question. The forward decodings
succeed --- $\{a_n,b_n\}\to\C$ at $\rvw=0.795$ and $\{a_n,b_n\}\to N$ at
$0.761$ --- confirming that a small network can read a large fraction of a
lossless encoding. The inverse routings $\C\to\{a_n,b_n\}$ and
$N\to\{a_n,b_n\}$ sit at $\rvw\simeq0.50$ and $0.46$: the smooth moments pin
down the coupling profile in the mean (the $b_n$ ramp is recovered almost
perfectly; Fig.~\ref{fig:lanczos}a,b) but not its sample-specific
fluctuations, precisely the behaviour Sec.~\ref{sec:mechdicho} predicts for
a self-averaging observable. This reading is reinforced by the
$D$-dependence: repeating the block at $D=251$ raises the inverse scores to
$\rvw\simeq0.61$ and $0.60$, so $\C,N\to\{a_n,b_n\}$ \emph{degrades} with
dimension while the forward $\{a_n,b_n\}\to\C,N$ stays high --- exactly the
signature of a self-averaging encoder whose unresolved fluctuation content
grows as a share of the microscopic data as $D$ increases. The ceiling task
is the sharp one:
$\{a_n,b_n\}\to\log_{10}{\rm SFF}$ returns $\rvw\simeq0$ (measured
$-0.025$) --- \emph{no better than predicting the mean}, even though the
chain data determines the spectrum, hence the SFF, exactly. The same network
class that decodes $\C$ and $N$ from $\{a_n,b_n\}$ at $\rvw\simeq0.75$--$0.80$
cannot
decode the fine SFF at all. This promotes the lossiness result of
Sec.~\ref{sec:sfffail} from a statement about the moments to a statement
about learnability itself: the fine SFF is not a decodable function of the
Krylov data at this model depth, so the failure of $\C\to{\rm SFF}$
($\rvw=0.185$) is information-theoretic in character, not a shortfall of the
moment channel relative to the lossless one. Spectral rigidity is encoded in
$\{a_n,b_n\}$ but only through a map too intricate for the class of decoders
that suffices for every other task in this paper --- itself a quantitative
statement about where the SFF sits on the complexity ladder.

\begin{figure}[t]
\centering
\includegraphics[width=\textwidth]{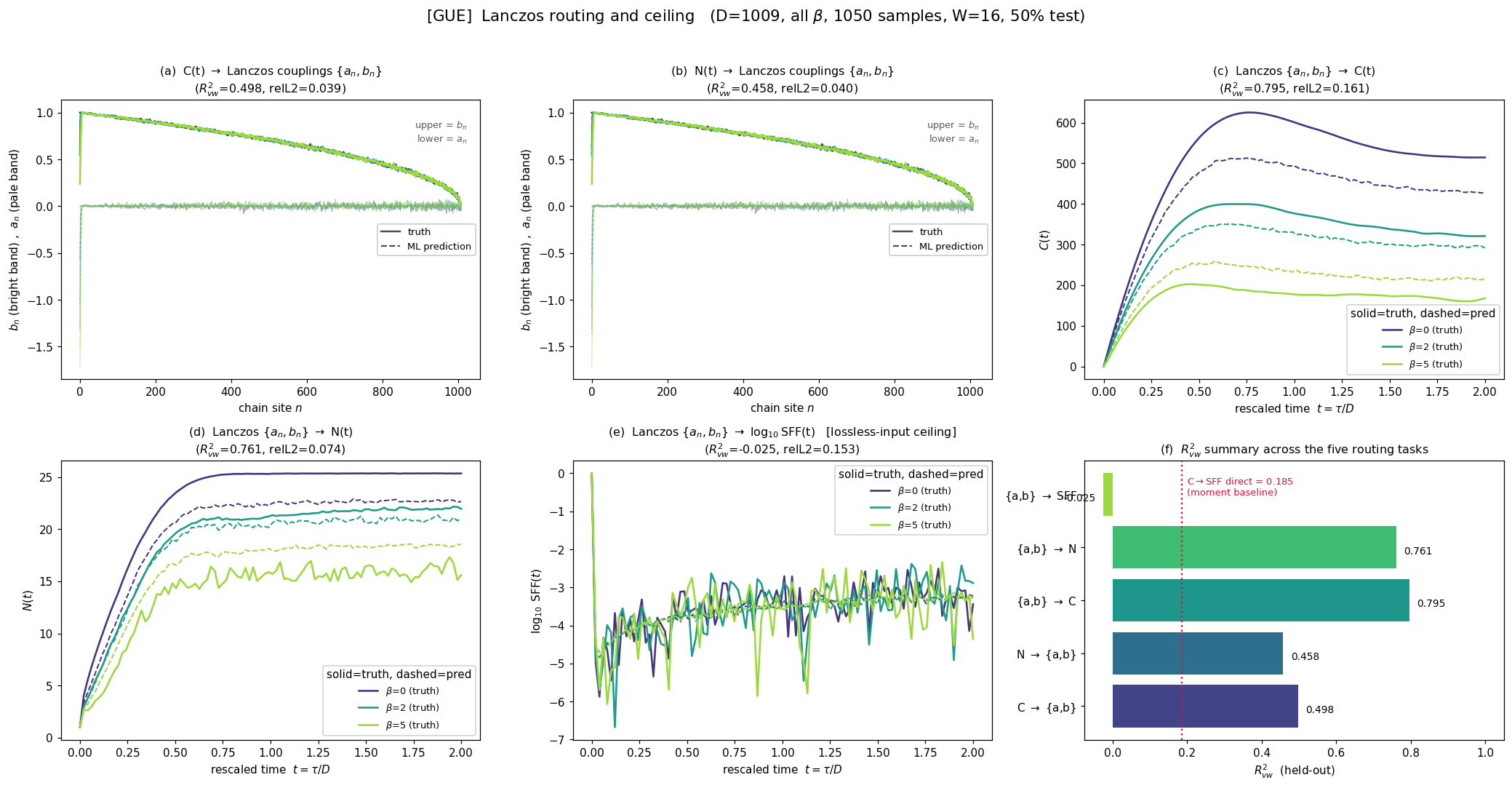}
\caption{\textbf{The Lanczos ceiling on GUE} ($D=1009$, all temperatures,
$1{,}050$ samples, $W=16$, $50\%$ test; solid = truth, dashed = prediction;
one colour per temperature $\beta\in\{0,2,5\}$). (a,b) Inverse routing:
$\C(t)$ and $N(t)$ regressed onto the chain couplings $\{a_n,b_n\}$ --- the
bright band is the $b_n$ ramp, the pale band the $a_n$ profile
($\rvw=0.498,0.458$; the mean ramp is captured, the fluctuations are not).
(c,d) Forward routing $\{a_n,b_n\}\to\C(t)$ and $\to N(t)$ ($\rvw=0.795,
0.761$): the lossless encoding is largely decodable. (e) The ceiling task
$\{a_n,b_n\}\to\log_{10}{\rm SFF}(t)$: despite the chain data determining
the spectrum exactly, the network scores $\rvw\simeq0$ (measured $-0.025$)
--- the fine SFF is not decodable at this depth. (f) Score summary; the
crimson line marks the direct $\C\to{\rm SFF}$ baseline ($0.185$). Panel (e)
is the sharp form of the lossiness result of Sec.~\ref{sec:sfffail}.}
\label{fig:lanczos}
\end{figure}

\subsection{Ensemble universality: GOE and Poisson}
\label{sec:goepoisson}

The pipeline of Sec.~\ref{sec:setup} depends on the spectrum alone
\cite{MehtaBook,Guhr:1997ve}. Changing the ensemble is therefore a one-line,
fully controlled substitution, and every
GUE task above was repeated verbatim on two
further ensembles: GOE (real symmetric matrices; level repulsion
$\beta_{\rm Dyson}=1$) and Poisson (independent levels drawn from the matched
semicircle density; no repulsion). Table~\ref{tab:ens} collects the results.

\begin{table}[t]
\centering
\begin{tabular}{@{}lccc@{}}
\toprule
Task & GUE & GOE & Poisson\\
\midrule
$\langle r\rangle$ certificate & $0.5996$ & $0.5307$ & $0.3863$\\
$\C\to\beta$ & $0.999$ & $0.998$ & $0.983$\\
$N\to\beta$ & $0.998$ & $0.998$ & $0.997$\\
$N\to\C$ \hfill ($\rvw$;\,$\relL$) & $0.994$;\,$0.035$ & $0.992$;\,$0.038$ &
$0.982$;\,$0.050$\\
$\C\to N$ & $0.973$;\,$0.024$ & $0.972$;\,$0.024$ & $0.968$;\,$0.025$\\
$\C\to{\rm SFF}$ & $0.185$;\,$0.175$ & $0.177$;\,$0.181$ & $0.184$;\,$0.189$\\
SFF plateau from $\C$ (full / early $20\%$) & $0.861/0.854$ &
$\approx0.866/0.859$ & $\approx0.836/0.83$\\
\bottomrule
\end{tabular}
\caption{Cross-ensemble replication (width-16 networks; training on $50\%$
of each dataset; the Poisson column uses the extended $D\le1009$ set). Every qualitative statement of
the GUE analysis survives the change of ensemble; the
quantitative shifts are physical and discussed in the text.}
\label{tab:ens}
\end{table}

\subsubsection{GOE: a second chaotic class}
\label{sec:goe}

Figure~\ref{fig:diagGOE} shows the GOE diagnostics with $\langle r\rangle$
pinned to the GOE surmise $0.5307$ in every $(D,\beta)$ block. The results
replicate GUE nearly number for number: temperature at $0.998$
(Fig.~\ref{fig:goeA}), mutual reconstruction at $0.992/0.972$
(Fig.~\ref{fig:goeNC}), SFF lossiness at $0.177$ (Fig.~\ref{fig:goeSFF}),
plateau budget $\approx0.866$ full versus $\approx0.859$ early
(Appendix~\ref{app:figs}). The lesson is universality: nothing in
the GUE analysis is a GUE accident; it is the phenomenology
of a rigid chaotic spectrum.

\begin{figure}[t]
\centering
\includegraphics[width=\textwidth]{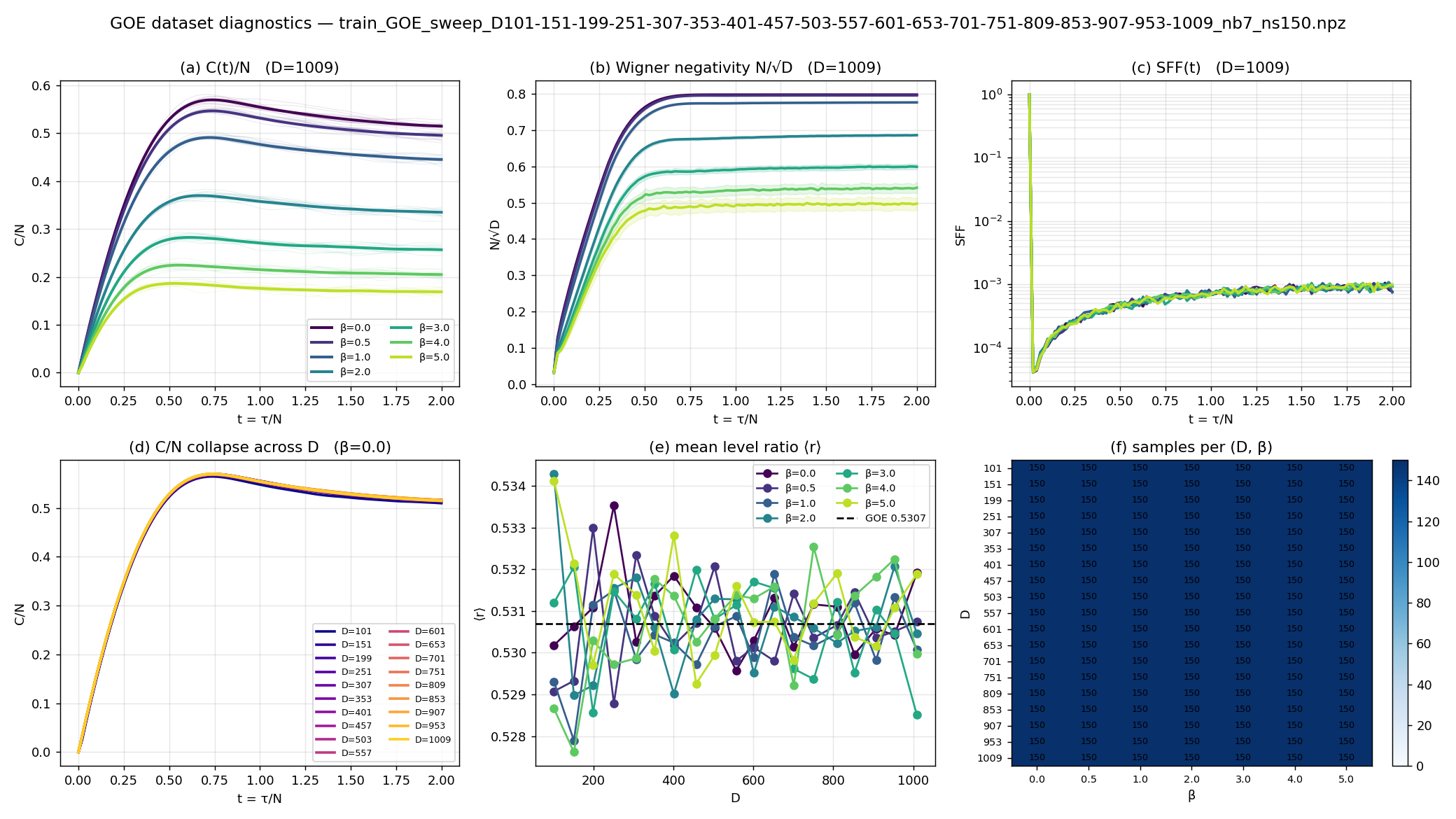}
\caption{GOE dataset diagnostics ($13{,}650$ samples). The
$\langle r\rangle$ panel sits on the GOE line $0.5307$ for all $(D,\beta)$
blocks --- the ensemble certificate.}
\label{fig:diagGOE}
\end{figure}

\begin{figure}[t]
\centering
\includegraphics[width=0.9\textwidth]{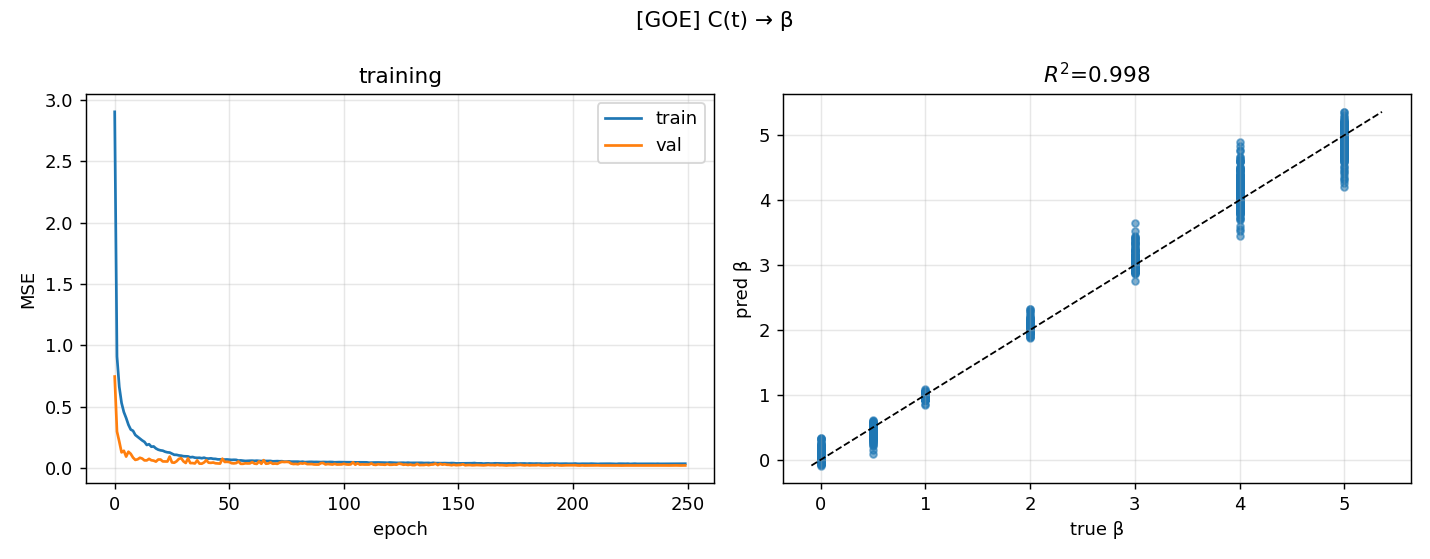}
\caption{GOE temperature regression $\C(t)\to\beta$, $R^2=0.998$.}
\label{fig:goeA}
\end{figure}

\begin{figure}[t]
\centering
\includegraphics[width=\textwidth]{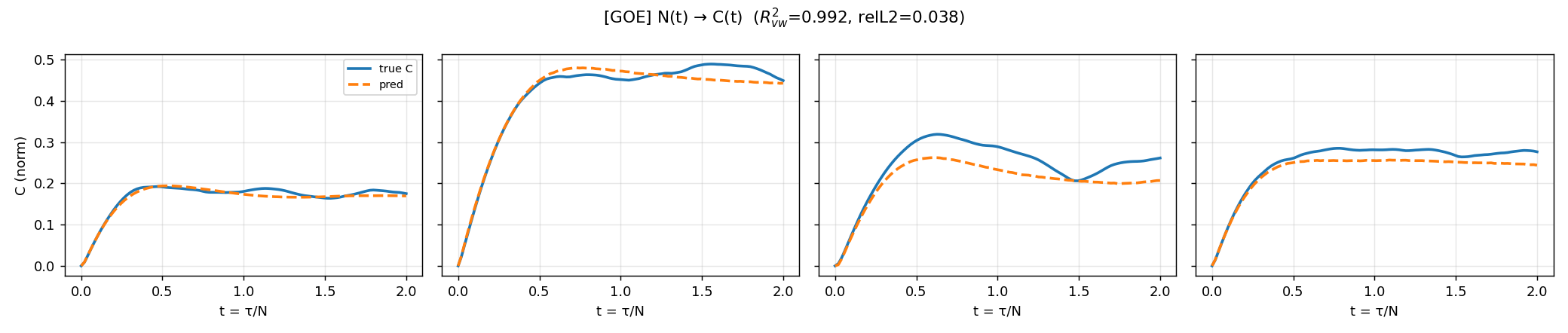}\\[1mm]
\includegraphics[width=\textwidth]{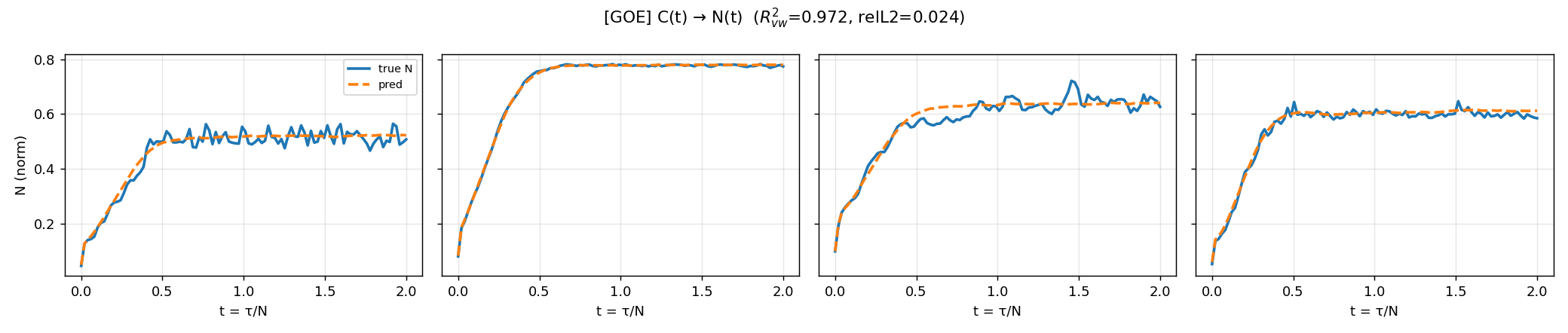}
\caption{GOE mutual reconstruction. \textbf{Top:} $N\to\C$, $\rvw=0.992$,
$\relL=0.038$. \textbf{Bottom:} $\C\to N$, $\rvw=0.972$, $\relL=0.024$. As on
GUE, the two metrics order the directions oppositely; no surplus is claimed.}
\label{fig:goeNC}
\end{figure}

\begin{figure}[t]
\centering
\includegraphics[width=\textwidth]{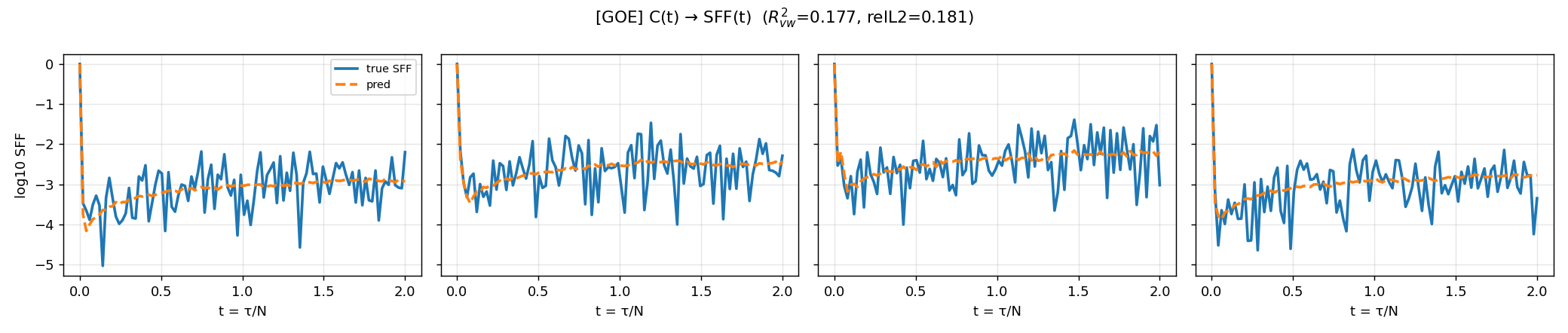}
\caption{GOE lossiness: $\C\to{\rm SFF}$, $\rvw=0.177$ --- the smooth
envelope, none of the rigidity fluctuations.}
\label{fig:goeSFF}
\end{figure}

\subsubsection{Poisson: the integrable spectral contrast}
\label{sec:poisson}

The Poisson set (Fig.~\ref{fig:diagPoisson}) replaces level repulsion with
independence while keeping the mean density fixed; it is also our largest
set, extended to nineteen dimensions $D\le1009$ ($19{,}950$ samples). Three
instructive facts, all visible in Table~\ref{tab:ens}:

\begin{itemize}
\item $\C\to\beta$ reaches $0.983$ and $N\to\beta$ $0.997$
(Fig.~\ref{fig:poiA}). The small residual deficit relative to the chaotic
ensembles is a fluctuation effect, not lost encoding: uncorrelated levels
fluctuate more than rigid ones at fixed $D$, and extending the grid from
$D\le701$ to $D\le1009$ raises the $\C$ score from $0.979$ to $0.983$ ---
the deficit shrinks with dimension exactly as fluctuation averaging
predicts. The scatter plot shows correct centres with slightly wider
clusters.
\item $\C\to{\rm SFF}$ sits at $0.184$ ($\relL=0.189$;
Fig.~\ref{fig:poiSFF}), statistically indistinguishable from GUE ($0.185$)
and GOE ($0.177$). The \emph{share} of the SFF carried by the smooth
envelope is ensemble-universal; what differs across ensembles is only the
character of the unpredictable remainder --- rigidity fluctuations on the
chaotic side, clustering fluctuations of independent levels here. Neither is
recoverable from the moments.
\item The plateau-budget ceiling, $\approx0.836$ full versus $\approx0.83$
early, remains the lowest of the three: with weaker spectral correlations,
less of the late-time plateau is determined by the smooth curve, and the
ordering GUE $\approx$ GOE $>$ Poisson of the ceilings tracks spectral
rigidity.
\end{itemize}
The mutual $\C\leftrightarrow N$ reconstruction is accurate and tightens on
the larger set (Fig.~\ref{fig:poiNC}; $0.982/0.968$), consistent with the
moment-slaving expectation for weakly correlated spectra; the remaining
panels are collected in Appendix~\ref{app:figs}.

\begin{figure}[t]
\centering
\includegraphics[width=\textwidth]{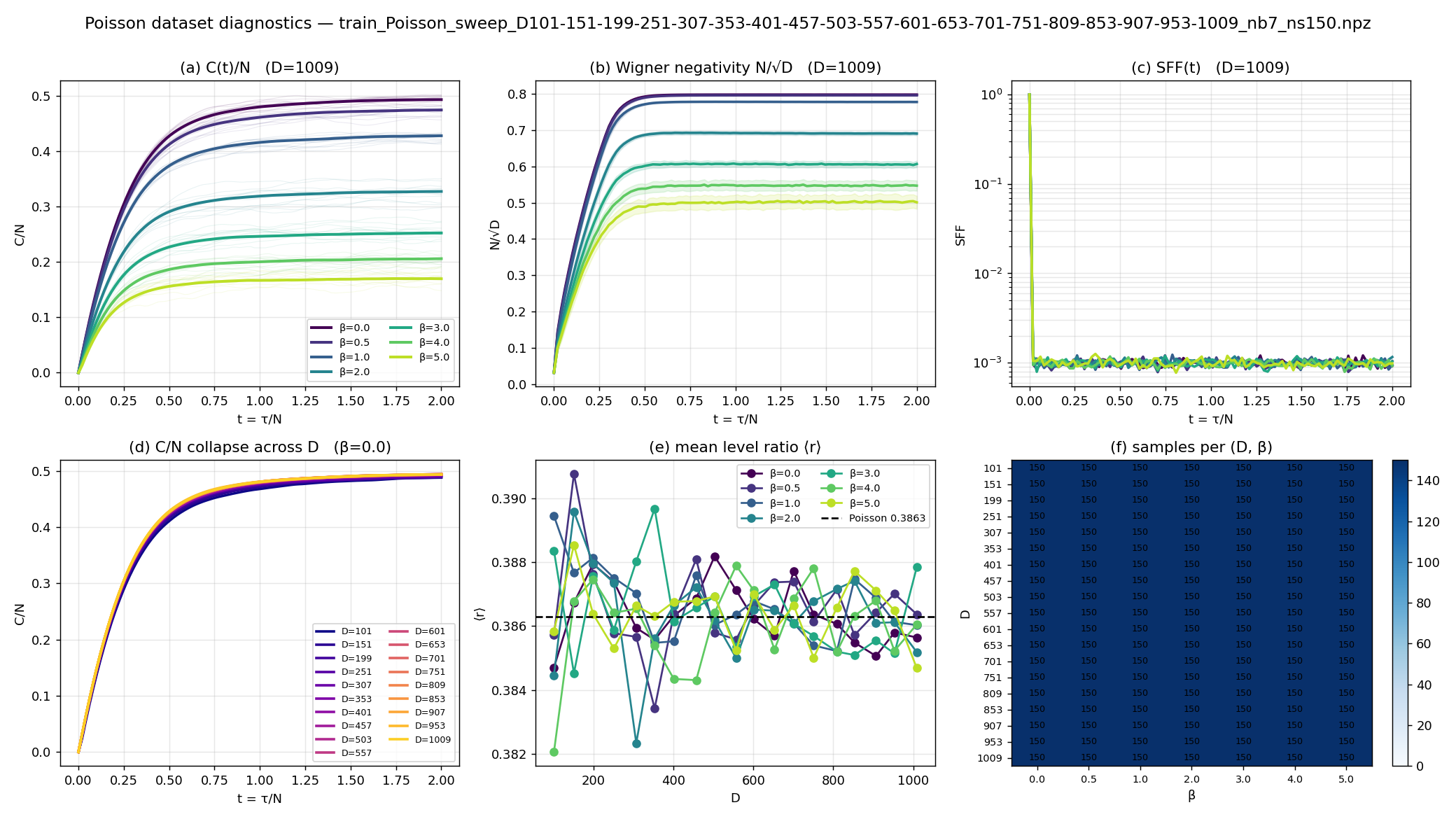}
\caption{Poisson dataset diagnostics ($19{,}950$ samples; $D=101$--$1009$,
nineteen primes); $\langle r\rangle$ on the Poisson line $0.3863$
throughout.}
\label{fig:diagPoisson}
\end{figure}

\begin{figure}[t]
\centering
\includegraphics[width=0.9\textwidth]{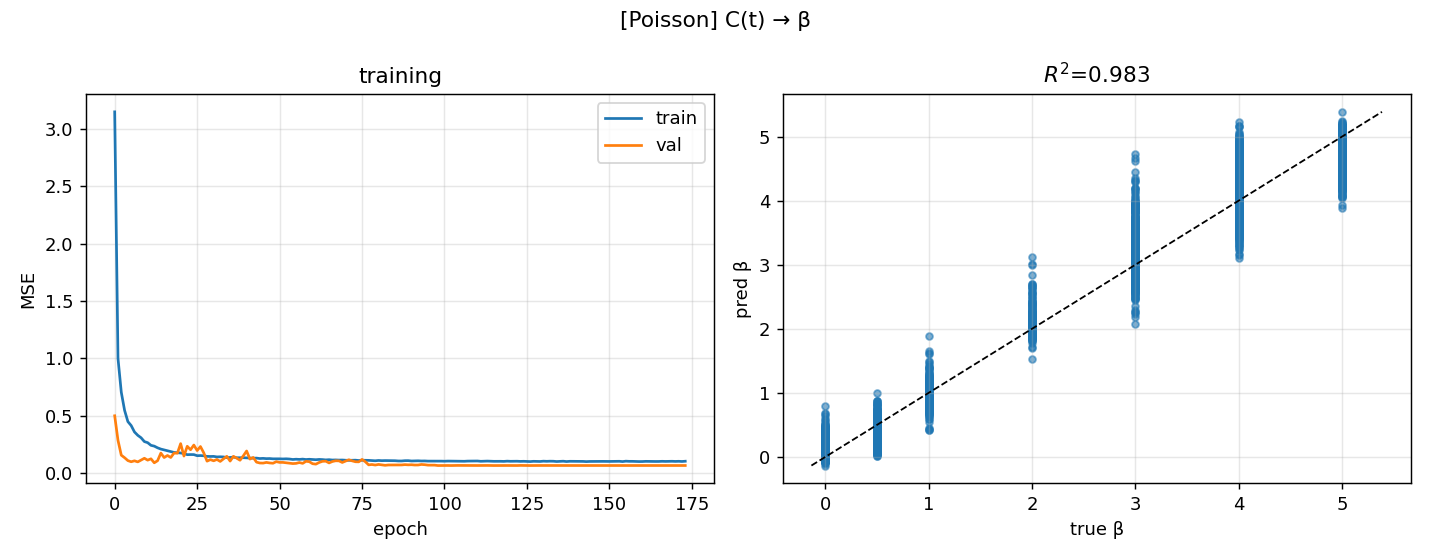}
\caption{Poisson temperature regression on the extended set,
$R^2=0.983$: correct centres, slightly wider clusters --- the fluctuations
of uncorrelated levels, not a loss of encoding.}
\label{fig:poiA}
\end{figure}

\begin{figure}[t]
\centering
\includegraphics[width=\textwidth]{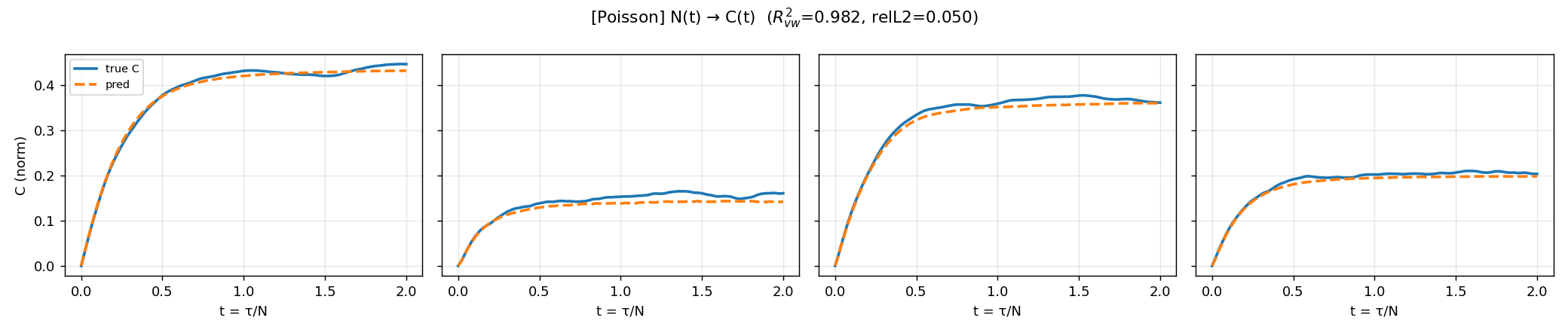}
\caption{Poisson mutual reconstruction, $N\to\C$: $\rvw=0.982$,
$\relL=0.050$ (the direction $\C\to N$, $0.968$, appears in
Appendix~\ref{app:figs}).}
\label{fig:poiNC}
\end{figure}

\begin{figure}[t]
\centering
\includegraphics[width=\textwidth]{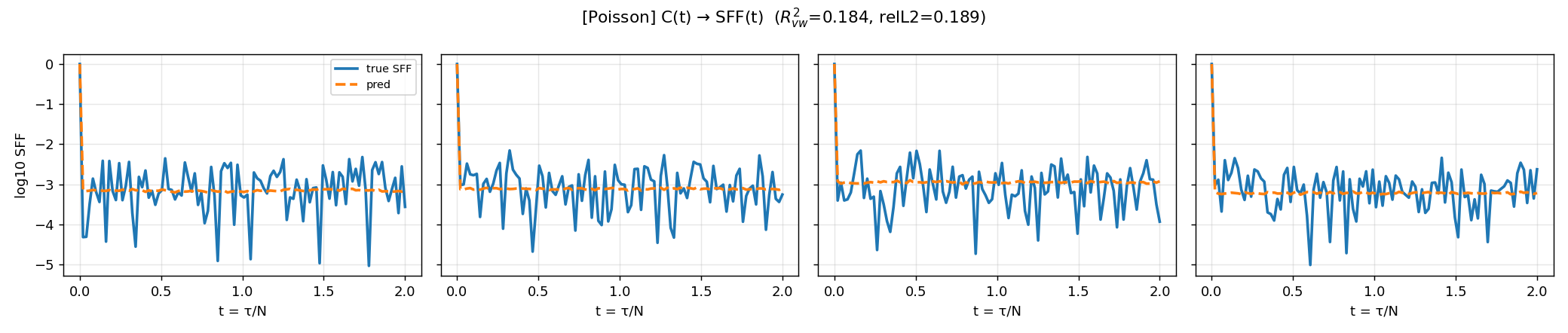}
\caption{Poisson lossiness: $\C\to{\rm SFF}$, $\rvw=0.184$ --- the same
envelope-only reconstruction as in the chaotic ensembles; here the missed
remainder is the clustering fluctuations of independent levels.}
\label{fig:poiSFF}
\end{figure}

\subsection{Identification of the universality class}
\label{sec:classify}

The discriminative version of the ensemble comparison is the sharpest: given
the $\C(t)$ of a \emph{single} sample, identify its parent ensemble. We train
per-dimension binary classifiers (gradient-boosted trees on the normalised
curve; $50\%$ test).

\begin{figure}[t]
\centering
\includegraphics[width=0.49\textwidth]{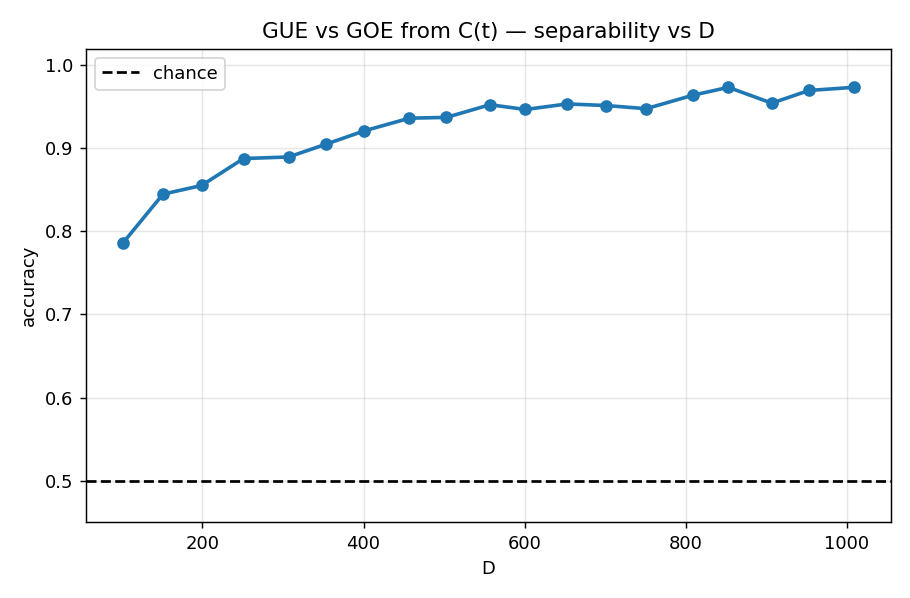}\hfill
\includegraphics[width=0.49\textwidth]{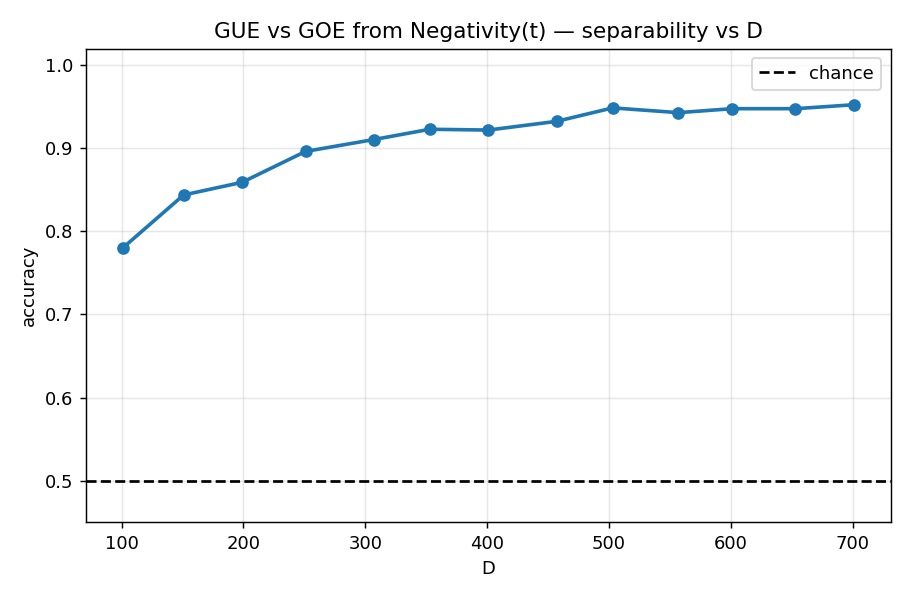}\\[1mm]
\includegraphics[width=0.6\textwidth]{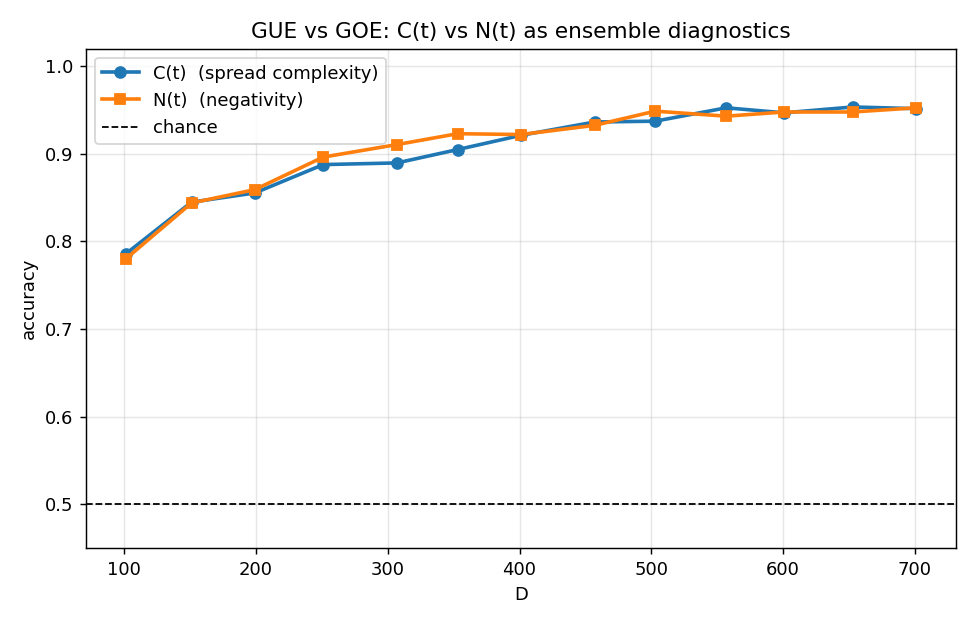}
\caption{The hard test: GUE vs GOE, per dimension. \textbf{Top left:} from
$\C(t)$ --- accuracy rises from $\approx0.79$ at $D=101$ to
$\approx0.97$--$0.98$ at $D\sim700$--$1000$. \textbf{Top right:} from $N(t)$
--- same trend, slightly lower ($\approx0.78\to0.95$). \textbf{Bottom:}
direct comparison; the first moment is the marginally better
ensemble-diagnostic at every $D$.}
\label{fig:clGOE}
\end{figure}

\begin{figure}[t]
\centering
\includegraphics[width=0.49\textwidth]{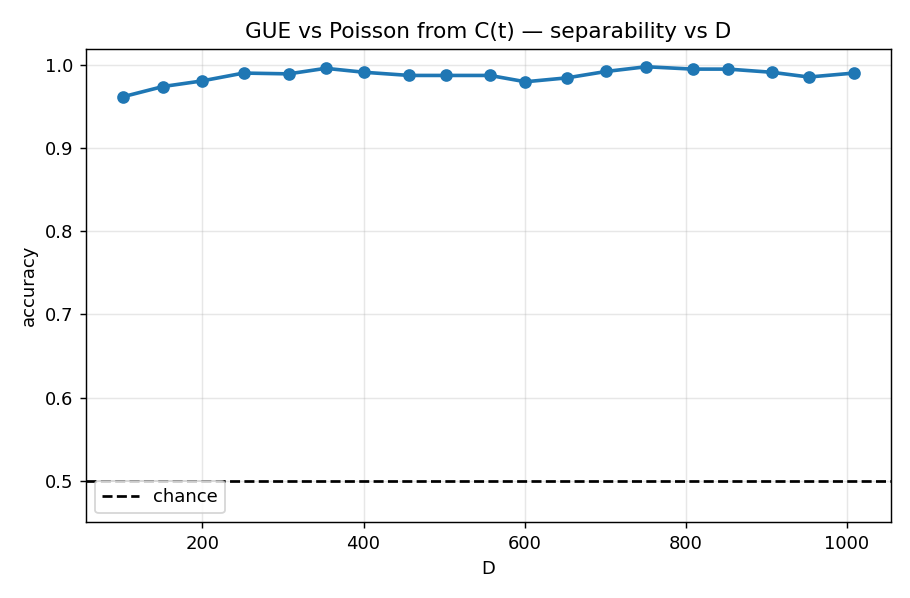}\hfill
\includegraphics[width=0.49\textwidth]{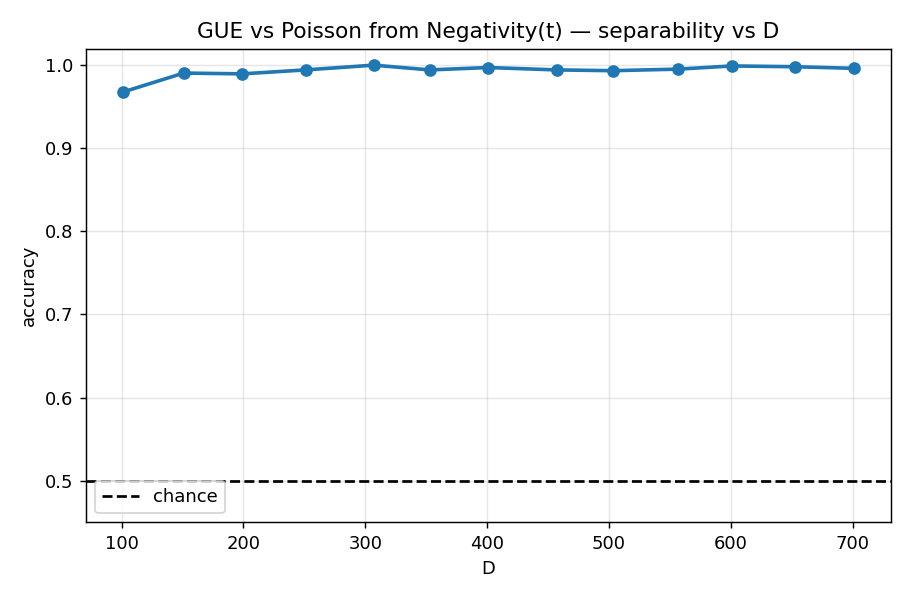}\\[1mm]
\includegraphics[width=0.6\textwidth]{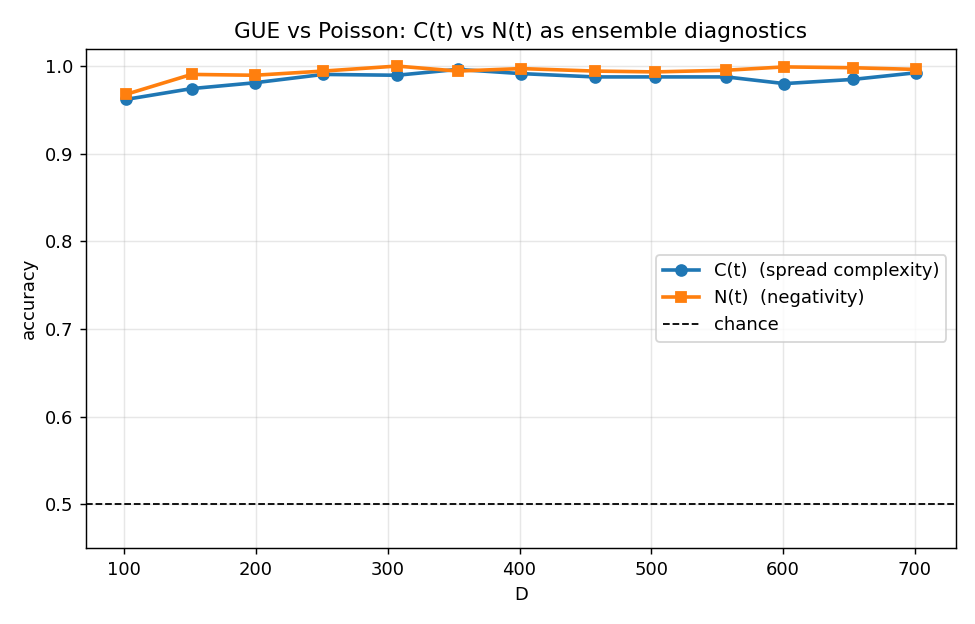}
\caption{The calibration test: GUE vs Poisson (chaotic vs integrable
spectra). Both channels separate the classes essentially perfectly at every
dimension up to $D=1009$ ($\approx0.96$ at $D=101$, $\to1.00$).}
\label{fig:clPoi}
\end{figure}

Two findings (Figs.~\ref{fig:clGOE}, \ref{fig:clPoi}):

\begin{itemize}
\item \textbf{Chaotic vs integrable is trivial.} The dichotomy is the
operational content of the Bohigas--Giannoni--Schmit and Berry--Tabor
conjectures \cite{Bohigas:1983er,Berry:1977ab}: GUE vs Poisson is separated
at $\approx0.96$--$1.00$ for every $D$ from either channel. This calibrates
the experiment: the curve shapes differ visibly, and the classifier finds it.
\item \textbf{Dyson class survives the compression --- and sharpens with
$D$.} GUE vs GOE is the stringent test: identical mean densities, identical
gross curve shapes, different level repulsion only. From $\C(t)$ the accuracy
climbs from $\approx0.79$ ($D=101$) to $\approx0.97$--$0.98$
($D\sim10^{3}$); from $N(t)$, $\approx0.78\to0.95$. The repulsion-induced
imprint on the smooth curves is $O(1/D)$-masked by statistical fluctuations at
small $D$ and emerges cleanly as they die away. The first moment is,
marginally but consistently, the better universality-class diagnostic.
\end{itemize}

Together with Sec.~\ref{sec:sff} this locates $\C(t)$ precisely on the
information ladder: it retains the \emph{class} of the spectral correlations
(and the coarse $e^{S}$ scale) while discarding their sample-specific fine
structure.

\section{The integrable CFT sector}
\label{sec:cft}

\subsection{Theory: the SL\texorpdfstring{$(2,\mathbb R)$}{(2,R)} chain and exact moment slaving}
\label{sec:cfttheory}

The natural integrable counterpart of the random-matrix datasets is the
SL$(2,\mathbb R)$ discrete-series sector, for three reasons. First, it is
\emph{exactly} the Krylov structure of two-dimensional CFT: the module of a
primary of weight $h$ under the global conformal algebra carries the
discrete-series representation $D_h^+$, with
\begin{equation}
L_0|h,n\rangle=(h+n)|h,n\rangle,\qquad
L_\pm|h,n\rangle=\sqrt{(n+\tfrac{1\mp1}{2})(n+2h-\tfrac{1\pm1}{2})}\,
|h,n\pm1\rangle .
\end{equation}
Second, the natural quench Hamiltonian $H=\alpha(L_++L_-)$ acting on the
highest-weight state generates Perelomov coherent states, whose Krylov chain
is the module itself: $a_n=0$ and
\begin{equation}
b_n=\alpha\sqrt{n(n+2h-1)} ,
\label{eq:bn}
\end{equation}
with closed-form dynamics
\begin{equation}
\begin{aligned}
\C(t)&=2h\sinh^2(\alpha t),\qquad
S(t)=\mathrm{sech}^{2h}(\alpha t),\\
|\psi_n(t)|^2&=\binom{n+2h-1}{n}\tanh^{2n}(\alpha t)\,
\mathrm{sech}^{4h}(\alpha t).
\end{aligned}
\label{eq:nb}
\end{equation}
Third --- and this is what makes it the \emph{control} --- the occupation
\eqref{eq:nb} is a negative-binomial distribution in $n$, a one-parameter
family at fixed $h$: every functional of $|\psi_n|^2$, in particular the
Wigner negativity, is analytically \emph{slaved} to the first moment. The
sharp prediction for our programme is therefore that \textbf{no
reconstruction asymmetry between $\C$ and $N$ can exist in this sector}; any
apparent asymmetry is a protocol artifact. This is the sector in which we
discovered, and repaired, the metric pathologies of
Sec.~\ref{sec:protocol}.

We also note the holographic reading: the thermal/BTZ autocorrelator of a 2d
CFT primary is exactly of the form \eqref{eq:nb} with
$\alpha\leftrightarrow\pi/\beta$, so this sector is the semiclassical
(infinite-entropy) limit of the holographic CFT --- the point of departure for
the chaos interpolation of Sec.~\ref{sec:chaos}.

The dataset comprises $2{,}079$ theories on a dense grid ($231$ values of
$h\in[0.25,6]$, $9$ of $\alpha\in[0.6,1.8]$), each verified against
\eqref{eq:nb} at the $10^{-8}$ level, at chain truncation $n_{\max}=503$ with
the strict edge guard of Sec.~\ref{sec:datasets}. Splits hold out entire $h$
values (Sec.~\ref{sec:protocol}): the test is generalisation to unseen
theories.

\subsection{Scalar tasks: the \texorpdfstring{$h\alpha^2$}{h alpha\textasciicircum 2} degeneracy and its resolution}
\label{sec:cftdeg}

At small $\alpha t$, Eq.~\eqref{eq:nb} gives $\C(t)\simeq2h\alpha^2t^2$: to
leading order the curves know only the \emph{product} $h\alpha^2$.
The data confirm this degeneracy with precision (Table~\ref{tab:cft},
Fig.~\ref{fig:cftcomp}): every channel predicts $h\alpha^2$ at $R^2=1.000$.
Resolving $h$ and $\alpha$ \emph{individually} requires sensitivity to the
subleading $\sinh$/$\mathrm{sech}$ structure, and here the channels separate
(Fig.~\ref{fig:cfth}): the negativity attains $N\to h=0.999$ and
$N\to\alpha=0.998$, ahead of the complexity ($0.984$, $0.985$) and of $\chi$
($0.974$, $0.974$). In residual terms the negativity's unexplained variance
for $h$ is an order of magnitude below the complexity's: \emph{the second
moment resolves what the first moment blurs}, here in a fully integrable
setting. That $\chi$ trails both is also physical: dividing by
$S=\mathrm{sech}^{2h}(\alpha t)$ partially cancels the very $h$-dependence
being regressed.

\begin{table}[t]
\centering
\begin{tabular}{@{}lccc@{}}
\toprule
channel & $\to h$ & $\to\alpha$ & $\to h\alpha^2$\\
\midrule
$\C$    & $0.984$ & $0.985$ & $1.000$\\
$N$     & $\mathbf{0.999}$ & $\mathbf{0.998}$ & $1.000$\\
$\chi$  & $0.974$ & $0.974$ & $1.000$\\
\bottomrule
\end{tabular}
\caption{Integrable-sector degeneracy table (held-out $h$; validation-chosen
models). All channels saturate the leading-order composite $h\alpha^2$; the
negativity best resolves the individual labels.}
\label{tab:cft}
\end{table}

\begin{figure}[t]
\centering
\includegraphics[width=0.32\textwidth]{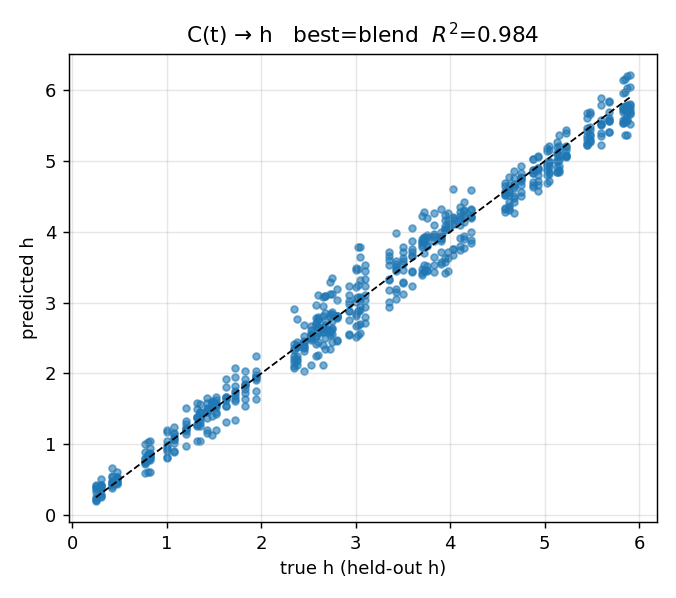}\hfill
\includegraphics[width=0.32\textwidth]{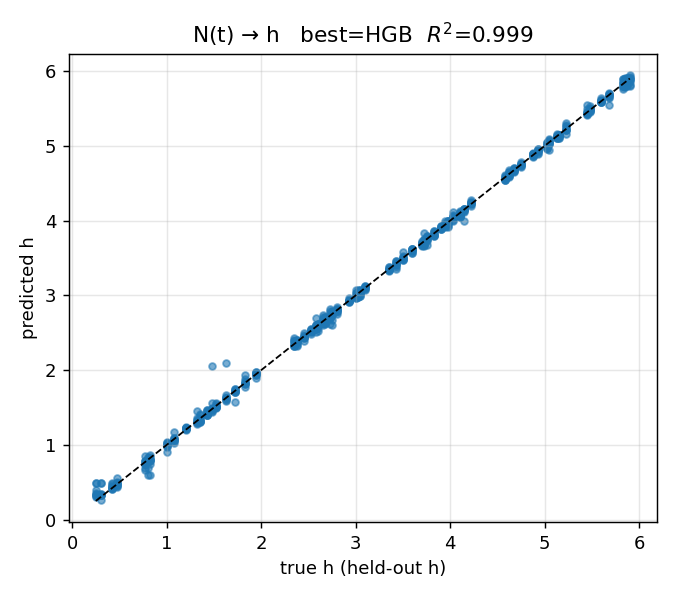}\hfill
\includegraphics[width=0.32\textwidth]{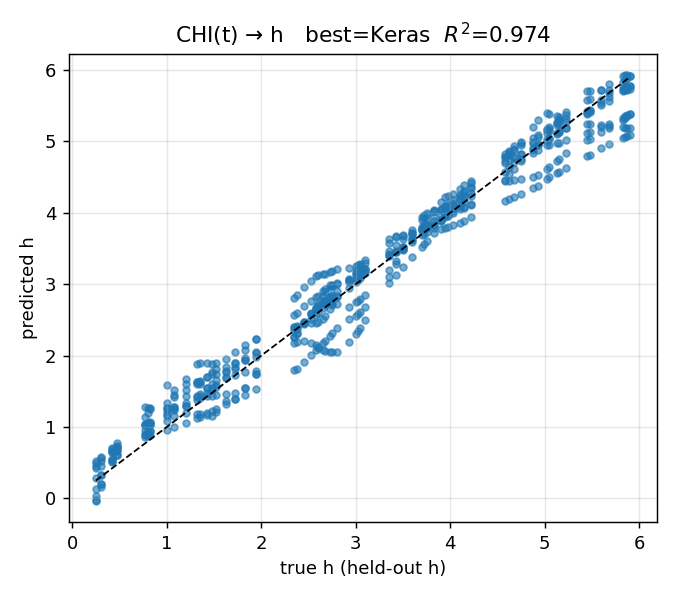}
\caption{Resolving the conformal weight on held-out $h$ values:
$\C\to h$ ($0.984$), $N\to h$ ($0.999$), $\chi\to h$ ($0.974$). The
$\to\alpha$ triplet (Appendix~\ref{app:figs}) shows the same ordering.}
\label{fig:cfth}
\end{figure}

\begin{figure}[t]
\centering
\includegraphics[width=0.32\textwidth]{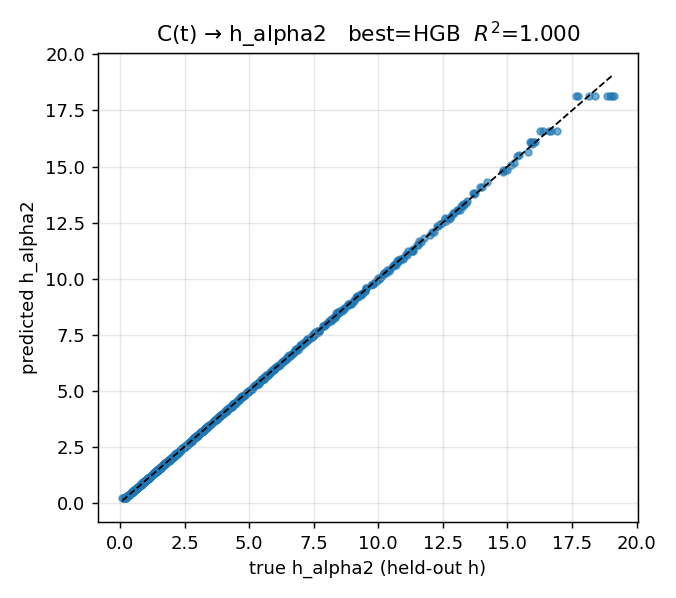}\hfill
\includegraphics[width=0.32\textwidth]{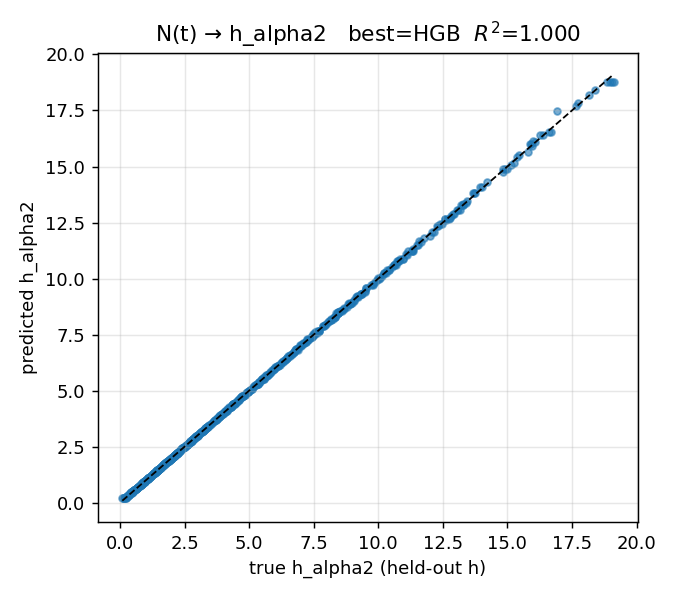}\hfill
\includegraphics[width=0.32\textwidth]{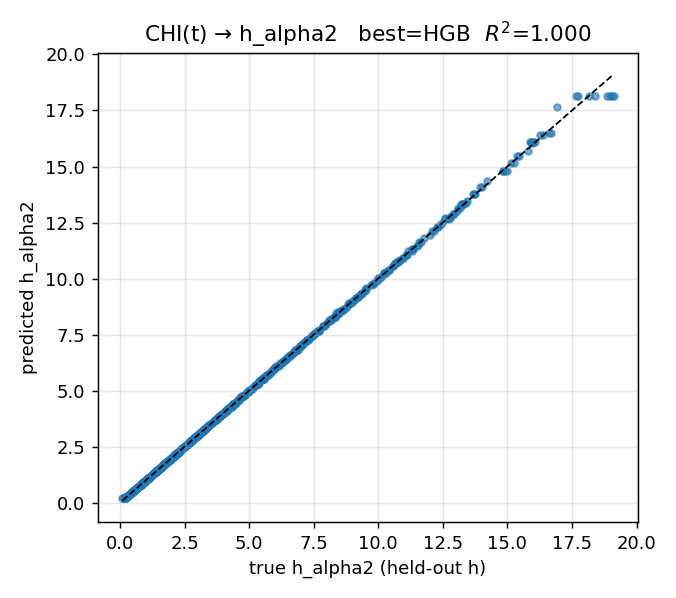}
\caption{The degeneracy: every channel predicts the composite $h\alpha^2$ at
$R^2=1.000$ --- the leading-order information content of all three curves.}
\label{fig:cftcomp}
\end{figure}

\subsection{Curve reconstruction and the metric-artifact exhibit}
\label{sec:cftcurves}

All six directed maps among $\{\C,N,\chi\}$ succeed at high fidelity. The
$\chi$ pairs (Fig.~\ref{fig:cftchi}) give $\C\to\chi$ at $\rvw=0.985$ against
$\chi\to\C$ at $0.958$, and $N\to\chi$ at $0.992$ against $\chi\to N$ at
$0.968$ --- mild, metric-dependent orderings with no stable surplus, as the
slaving demands.

The instructive case is the $\C\leftrightarrow N$ pair
(Fig.~\ref{fig:cftexhibit}), which we display deliberately in its
\emph{uncorrected} raw-channel form: the titles read $\rvw=0.679$ for
$\C\to N$ against $0.955$ for $N\to\C$ --- an apparent asymmetry of $+0.28$.
This is the metric artifact of Sec.~\ref{sec:protocol}, exhibited: across the
grid the raw $\C$ curves span a $\sim500\times$ dynamic range, variance
weighting concentrates the loss on the large-$h$ giants, and the small-$h$
half of the test set is fit arbitrarily badly without penalty; the
relative-error metric, which weights samples equally, reverses the ordering.
Re-learning the same map in the logarithmic $\C$ channel at matched width
restores symmetry: $0.944$ against $0.949$, a gap of $+0.005$ --- compatible
with zero, exactly as the negative-binomial slaving requires. We keep the raw
panels in the paper because they are the cleanest illustration we know of how
a dynamic-range pathology can manufacture a physics-flavoured ML claim; every
asymmetry quoted in Sec.~\ref{sec:chaos} is protected against this failure
mode by construction.

\begin{figure}[t]
\centering
\includegraphics[width=\textwidth]{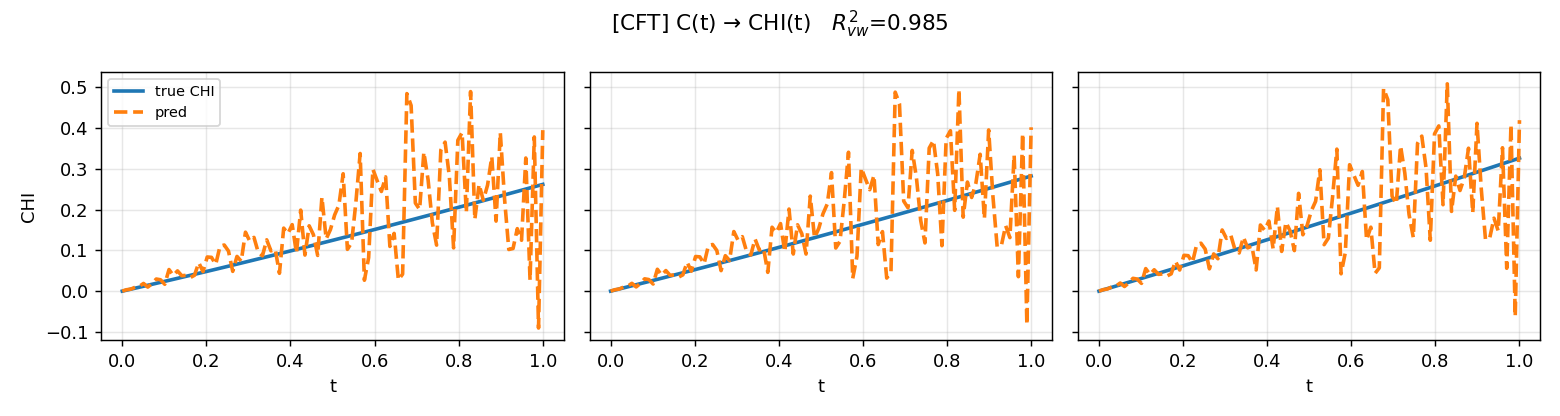}\\[1mm]
\includegraphics[width=\textwidth]{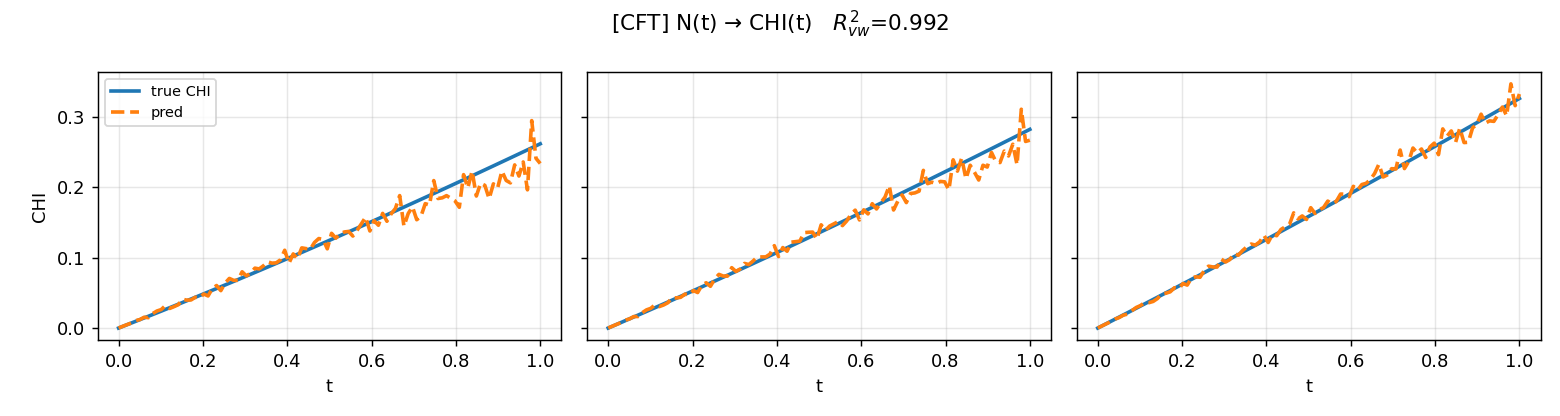}
\caption{Integrable-sector $\chi$ reconstructions (held-out $h$):
$\C\to\chi$ ($\rvw=0.985$) and $N\to\chi$ ($0.992$); the reverse maps
($0.958$, $0.968$) appear in Appendix~\ref{app:figs}. No stable asymmetry
survives the dual-metric rule --- the integrable control.}
\label{fig:cftchi}
\end{figure}

\begin{figure}[t]
\centering
\includegraphics[width=\textwidth]{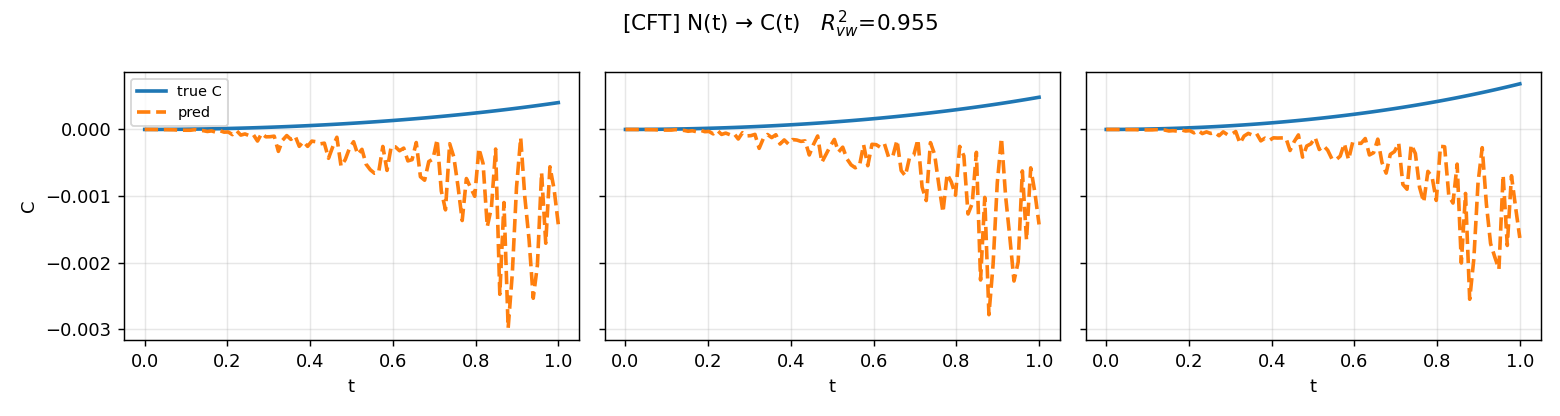}\\[1mm]
\includegraphics[width=\textwidth]{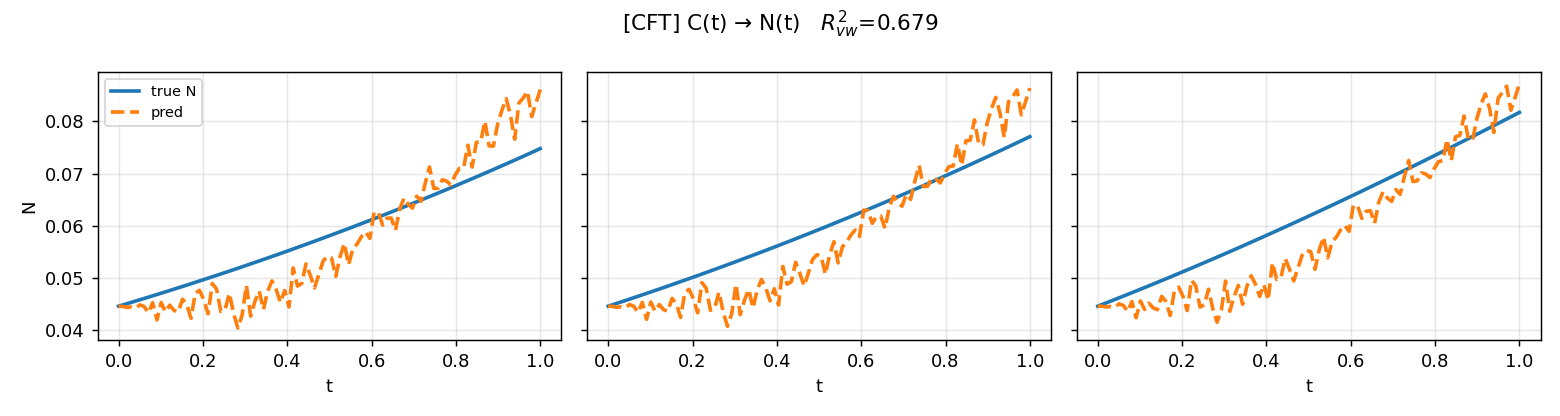}
\caption{\textbf{The metric-artifact exhibit} (raw channels, deliberately
uncorrected). Titles: $N\to\C$ at $\rvw=0.955$, $\C\to N$ at $0.679$ --- an
apparent $+0.28$ surplus. The relative-error metric reverses the sign, and
the logarithmic-channel refit closes the gap to $+0.005$: the ``surplus'' is
a dynamic-range pathology, not physics. See Sec.~\ref{sec:cftcurves}.}
\label{fig:cftexhibit}
\end{figure}

\subsection{Limitations}
\label{sec:cftlim}

Four limitations bound the claims of this section. (i) \emph{Truncation:}
the chain is cut at $n_{\max}=503$; the guard certifies edge occupation
below $10^{-10}$, but the sector is fundamentally infinite-dimensional and
all statements are at finite truncation. (ii) \emph{Primary sector only:}
we work in the global SL$(2,\mathbb R)$ module. The full Virasoro module of
a 2d CFT is not a one-dimensional chain (its descendant lattice is organised
by partitions), and our discrete-phase-space construction requires the chain;
extending the negativity to the full module is an open problem. (iii)
\emph{Deterministic grid:} there is no sampling noise, so generalisation is
tested by held-out theories; scores are not comparable across split designs,
only within them. (iv) \emph{Holographic reading is semiclassical:} the
identification $\alpha\leftrightarrow\pi/\beta$ with the BTZ autocorrelator
is a statement about the infinite-entropy limit; finite-entropy corrections
are precisely what Sec.~\ref{sec:chaos} models, and only in a random-matrix
proxy.

\section{CFT to chaos}
\label{sec:chaos}

\subsection{Motivation and model}
\label{sec:chaosmodel}

The semiclassical two-point function of a holographic CFT fixes exactly the
integrable chain of Sec.~\ref{sec:cfttheory}. Chaos is not visible there: it
enters through the finite-entropy discreteness of the black-hole spectrum,
whose late-time statistics are random-matrix universal
\cite{Cotler:2016fpe,Saad:2019lba}. The minimal model of this completion is an additive
random-matrix deformation of the chain,
\begin{equation}
H(\varepsilon)\;=\;H_{{\rm SL}(2,\mathbb R)}\;+\;\varepsilon\,R_0\,
W_{\rm GUE},
\label{eq:interp}
\end{equation}
with $H_{{\rm SL}(2,\mathbb R)}$ the truncated chain \eqref{eq:bn} at
$D=251$, $W_{\rm GUE}$ a unit-normalised GUE draw, and $R_0$ the spectral
half-width of the unperturbed chain, so that $\varepsilon$ measures the
perturbation in units of the total bandwidth. For every sample the deformed
Hamiltonian is \emph{re-tridiagonalised from the same initial vector}, so
that $\C$, $N$, $S$ and $\chi$ remain exact Krylov observables at every
$\varepsilon$. At $\varepsilon=0$ the analytic forms \eqref{eq:nb} are
reproduced at the $10^{-15}$ level.

The chain's own spectral statistics certify the interpolation
(Table~\ref{tab:eps}): the unperturbed chain is \emph{rigid}
(``picket-fence'', $\langle r\rangle=0.978$), and $\langle r\rangle$ crosses
to the GUE plateau $0.600$ by $\varepsilon\simeq0.12$. The dataset comprises
$16$ balanced strata of $480$ samples ($40$ $h$-values $\times$ $3$ $\alpha$
$\times$ $4$ realisations at $\varepsilon>0$; $160\times3$ distinct theories
at $\varepsilon=0$), $7{,}680$ samples in all.

\paragraph{Truncation protocol (stated explicitly).} At $\varepsilon=0$ the
chain truncates an infinite module: the strict edge guard applies. At
$\varepsilon>0$, Eq.~\eqref{eq:interp} defines a \emph{closed} $D$-dimensional
model --- nothing is truncated --- and occupation of the far end of the chain
is physical equilibration, exactly as in the GUE datasets; it reaches
$\sim10^{-1}$ in the most chaotic corner. Accordingly \emph{no} samples are
discarded at $\varepsilon>0$, and the maximal edge occupation is stored per
sample as a diagnostic (reported in Table~\ref{tab:eps}). We emphasise this
because the opposite choice --- guarding at all $\varepsilon$ --- discards
preferentially the most chaotic samples, a selection bias correlated with the
control parameter; we made, detected, and repaired exactly this error, and
the audit is folded into the protocol of Sec.~\ref{sec:protocol}.

\subsection{Results}
\label{sec:chaosresults}

At every $\varepsilon$ stratum we train, under the full protocol
(held-out $h$ common to all strata; dual metrics; two seeds), the scalar
tasks $\{\C,N,\chi\}\to h$ and the four curve maps
$\C\leftrightarrow N$, $\C\leftrightarrow\chi$, and form the seed-averaged
asymmetry gaps
\begin{equation}
{\rm gap}_N=\rvw(N\to\C)-\rvw(\C\to N),\qquad
{\rm gap}_\chi=\rvw(\chi\to\C)-\rvw(\C\to\chi).
\end{equation}
Table~\ref{tab:eps} lists four representative strata; Fig.~\ref{fig:headline}
shows the full grid.

\begin{table}[t]
\centering\small
\begin{tabular}{@{}lcccccc@{}}
\toprule
$\varepsilon$ & $\langle r\rangle$ & max edge & $\C\to h$ &
${\rm gap}_N$ & ${\rm gap}_\chi$ & metrics\\
\midrule
$0.015$ & $0.766$ & $3\times10^{-10}$ & $+0.50$ &
$+0.225\pm0.039$ & $+0.471\pm0.036$ & agree\\
$0.035$ & $0.658$ & $1\times10^{-8}$ & $-0.09$ &
$+0.156\pm0.004$ & $+0.531\pm0.008$ & agree\\
$0.15$ & $0.601$ & $6\times10^{-2}$ & $-0.16$ &
$+0.006\pm0.007$ & $+0.691\pm0.015$ & agree\\
$0.20$ & $0.600$ & $1\times10^{-1}$ & $-0.16$ &
$+0.001\pm0.007$ & $+0.614\pm0.027$ & agree\\
\bottomrule
\end{tabular}
\caption{Representative strata of the interpolation ($480$ samples each;
gaps seed-averaged $\pm$ spread; ``agree'' = the relative-$L^2$ metric
confirms the sign of both gaps). The edge column is the stored equilibration
diagnostic. The $\varepsilon=0$ anchor is supplied by the integrable control
of Sec.~\ref{sec:cft} (its fine $h$-grid shares too few held-out values with
the coarse $\varepsilon>0$ grid for an in-study point, and duplicating
realisations would leak).}
\label{tab:eps}
\end{table}

\begin{figure}[t]
\centering
\includegraphics[width=\textwidth]{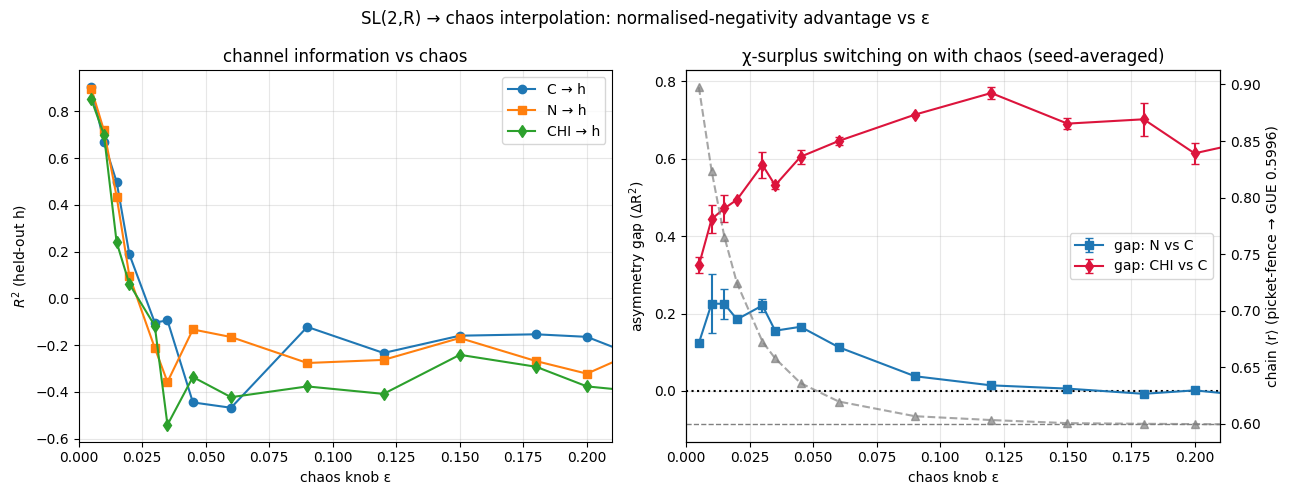}
\caption{\textbf{Headline.} The SL$(2,\mathbb R)\to$chaos interpolation over
the full $\varepsilon$ grid. \textbf{Left:} $R^2$ for
$\{\C,N,\chi\}\to h$ on held-out $h$. All three channels lose the microscopic
label \emph{together} as the random matrix washes it out of the dynamics ---
the sanity panel. \textbf{Right:} the seed-averaged asymmetry gaps with the
chain $\langle r\rangle$ (grey, right axis) as chaos certificate. The
$\chi$-gap (red) switches on and grows monotonically from $+0.33$ at
$\varepsilon=0.005$ to $+0.77$ at $\varepsilon\simeq0.12$ --- the point where
$\langle r\rangle$ first reaches the GUE plateau --- then saturates near
$0.6$--$0.7$. The raw-$N$ gap (blue) starts at $+0.13$--$0.22$ in the
transition region and decays to zero deep in chaos. Every point carries
both-metric agreement.}
\label{fig:headline}
\end{figure}

Three statements, in decreasing order of importance:

\begin{enumerate}
\item \textbf{The $\chi$-surplus is a chaos effect.} ${\rm gap}_\chi$ rises
monotonically through the $\langle r\rangle$ crossover and saturates once the
statistics are GUE. Its onset tracks the chaos certificate point by point,
with error bars from seed spread and two-metric agreement at every
$\varepsilon$.
\item \textbf{The surplus is sourced by $S$ and carried by $\chi$.}
${\rm gap}_N$ behaves oppositely, decaying to zero deep in chaos --- and
recall (Appendix~\ref{sec:smooth}) that the small pooled-GUE raw-$N$ gap is
removed by target smoothing. It is the division by the survival amplitude,
the defining operation of $\chi$ \cite{Basu:2026inf}, that exposes
second-moment information inaccessible to the first moment. Whether this
information is ``in $S$'' or ``in $\chi$'' is a fair question, and we
answer it directly in Sec.~\ref{sec:chaoscontrols}: $\C$ fails on $|S|$
and on $\chi$ alike, but only $\chi$ stays decodable in the reverse
direction, because the envelope that the late-time $|S|$ buries in speckle
is still present in $N$. In the integrable sector $S$ is itself slaved
(Eq.~\eqref{eq:nb}) and the division adds nothing.
\item \textbf{Both endpoints are independently anchored.} The
$\varepsilon=0$ anchor is the integrable control (gaps compatible with zero,
as slaving demands); the large-$\varepsilon$ behaviour joins the pure-GUE
phenomenology of Sec.~\ref{sec:gue}.
\end{enumerate}

Taken together: \emph{in the integrable sector the Krylov moment tower is
slaved and the observables are informationally equivalent; as the dynamics
becomes chaotic the tower decouples, the survival amplitude becomes
unpredictable from the first moment, and the normalised Wigner negativity is
the observable that carries that unpredictable content while remaining
reconstructible from it.} This is the operational, machine-learned content
of the second-moment infall proposal of \cite{Basu:2026inf}.

\subsection{Is the surplus in \texorpdfstring{$S$}{S}? Three controls}
\label{sec:chaoscontrols}

Since $S$ enters $\chi$ by construction, one should ask whether the
asymmetry of Fig.~\ref{fig:headline} is a property of the normalised
negativity or simply information already present in the survival amplitude.
We ran three controls, all with the protocol of Sec.~\ref{sec:protocol}
(held-out $h$, seed-averaged, $W=32$).

\paragraph{$|S|$, $N$ and $\chi$ as predictors of $\C$.}
We regressed $\C$ from each of the three channels and back
(Fig.~\ref{fig:q1}). In the forward direction the answer is unambiguous:
$\C\to|S|$ and $\C\to\chi$ collapse together, from $\simeq0.99$ at
$\varepsilon=0$ to $\simeq0$ for $\varepsilon\gtrsim0.15$, while $\C\to N$
stays above $0.94$ throughout. The unpredictable content is the survival
amplitude, and $\chi$ inherits it from $|S|$; in this sense the referee's
suspicion is correct, and it is precisely the mechanism of
Sec.~\ref{sec:mech}. The reverse direction is where the two observables
part company. $|S|\to\C$ degrades steadily with chaos, from $0.95$ to
$0.37$ at $\varepsilon=0.3$, because the late-time amplitude is speckle and
the envelope is buried in it; $\chi\to\C$ holds at $0.72$--$0.87$ over the
same range, because the envelope is carried by $N$. Hence
${\rm gap}_\chi>{\rm gap}_{|S|}$ ($0.72$ vs $0.38$ at $\varepsilon=0.3$;
$0.49$ vs $0.32$ on pooled GUE), not because $\chi$ hides more from $\C$
but because it reveals more to it. $\chi$ is the channel that keeps both
the microscopic amplitude and the macroscopic envelope decodable at once.

\paragraph{Re-pairing $N$ and $S$.}
We rebuilt $\chi_{\rm shuffle}=N_i/|S_j|$ with $|S|$ borrowed from a
different sample in the same $(h,\alpha)$ block, so the envelope is
preserved and only the physical link between $N$ and $S$ is cut
(Fig.~\ref{fig:q2}). At $\varepsilon=0$ nothing changes, as slaving
demands. In the crossover, $\varepsilon=0.01$--$0.06$, the re-paired
observable shows a \emph{larger} gap than the physical one (by
$0.2$--$0.36$): mismatched envelopes make $\C\to\chi_{\rm shuffle}$ collapse
while $N$ still rescues the reverse map. The physical pairing therefore
makes the measured surplus conservative; a gap of this size cannot be
manufactured by dividing by an arbitrary amplitude. Deep in chaos the sign
reverses: re-pairing lowers $\chi\to\C$ by $0.1$--$0.2$ (from $0.80$ to
$0.65$ at $\varepsilon=0.2$; $0.72$ to $0.52$ at $\varepsilon=0.3$), so the
correlation between $N$ and $S$ within one evolution carries information
$\chi$ genuinely uses. On GUE at $D=1009$ the effect vanishes ($0.843$ vs
$0.840$), as it should once the speckle is fully universal
(Lemma~\ref{lem:speckle}).

\paragraph{A second chaotic class.}
Replacing $W_{\rm GUE}$ by a real symmetric $W_{\rm GOE}$ of the same
normalisation, with everything else unchanged, drives the chain
$\langle r\rangle$ from $0.978$ to $0.531$, on the GOE surmise rather than
the GUE one (Fig.~\ref{fig:q3}). The surplus reproduces: ${\rm gap}_\chi$
switches on with chaos and plateaus at $\simeq0.70$ for
$\varepsilon\simeq0.1$--$0.25$ against $0.77$--$0.80$ for GUE, while
${\rm gap}_N\to0$. This is Prediction~3 of Sec.~\ref{sec:mechpred}: the
phasor limit behind Lemma~\ref{lem:speckle} does not distinguish the two
classes, so a quantitatively similar plateau was expected, and that is what
we find.

\begin{figure*}[t]
\centering
\includegraphics[width=\textwidth]{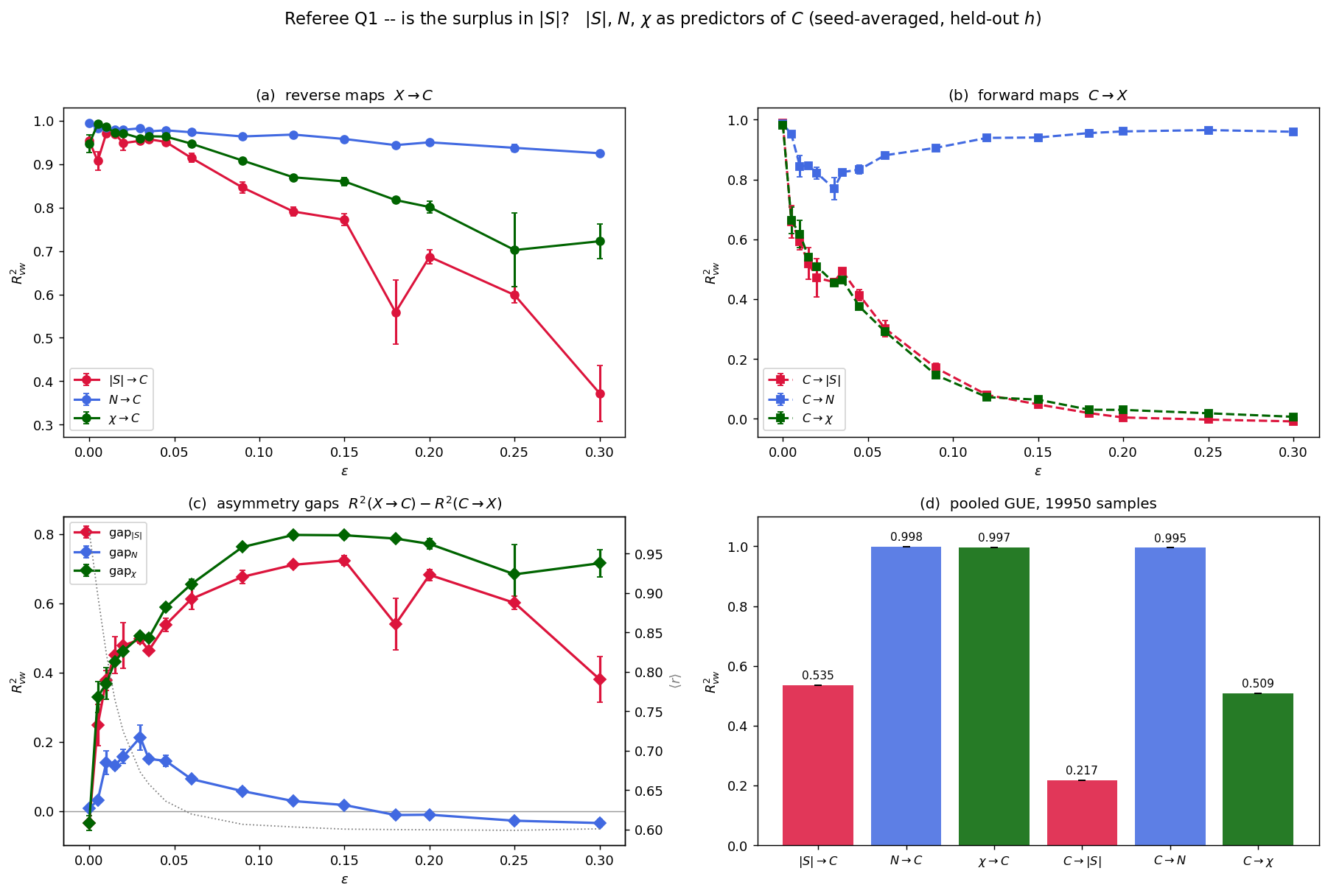}
\caption{$|S|$, $N$ and $\chi$ as predictors of $\C$ along the
interpolation (held-out $h$, seed-averaged). (a) Reverse maps $X\to\C$: all
three recover the envelope at small $\varepsilon$; $|S|\to\C$ degrades with
chaos while $\chi\to\C$ holds. (b) Forward maps $\C\to X$: $\C\to|S|$ and
$\C\to\chi$ collapse together, $\C\to N$ does not. (c) The three gaps;
${\rm gap}_\chi>{\rm gap}_{|S|}$ throughout the chaotic phase. (d) Pooled
GUE.}
\label{fig:q1}
\end{figure*}

\begin{figure*}[t]
\centering
\includegraphics[width=\textwidth]{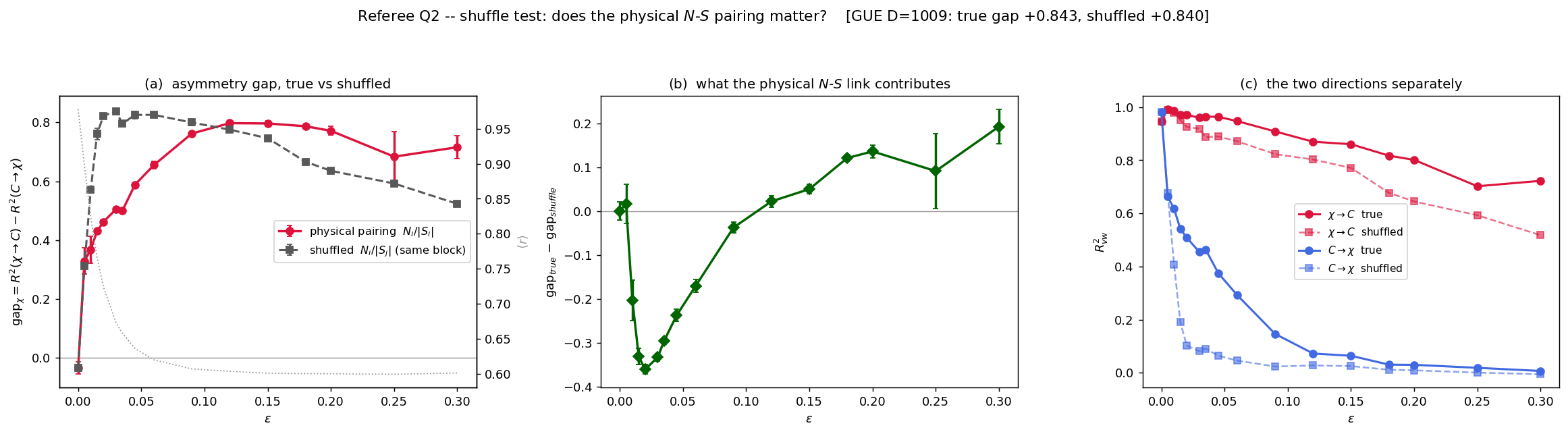}
\caption{Re-pairing test. (a) ${\rm gap}_\chi$ for the physical pairing
$N_i/|S_i|$ and for $N_i/|S_j|$ with $|S|$ drawn from another sample of the
same $(h,\alpha)$ block. (b) Their difference: negative in the crossover
(the physical pairing is conservative), positive deep in chaos (the pairing
carries information). (c) The two directions separately.}
\label{fig:q2}
\end{figure*}

\begin{figure*}[t]
\centering
\includegraphics[width=\textwidth]{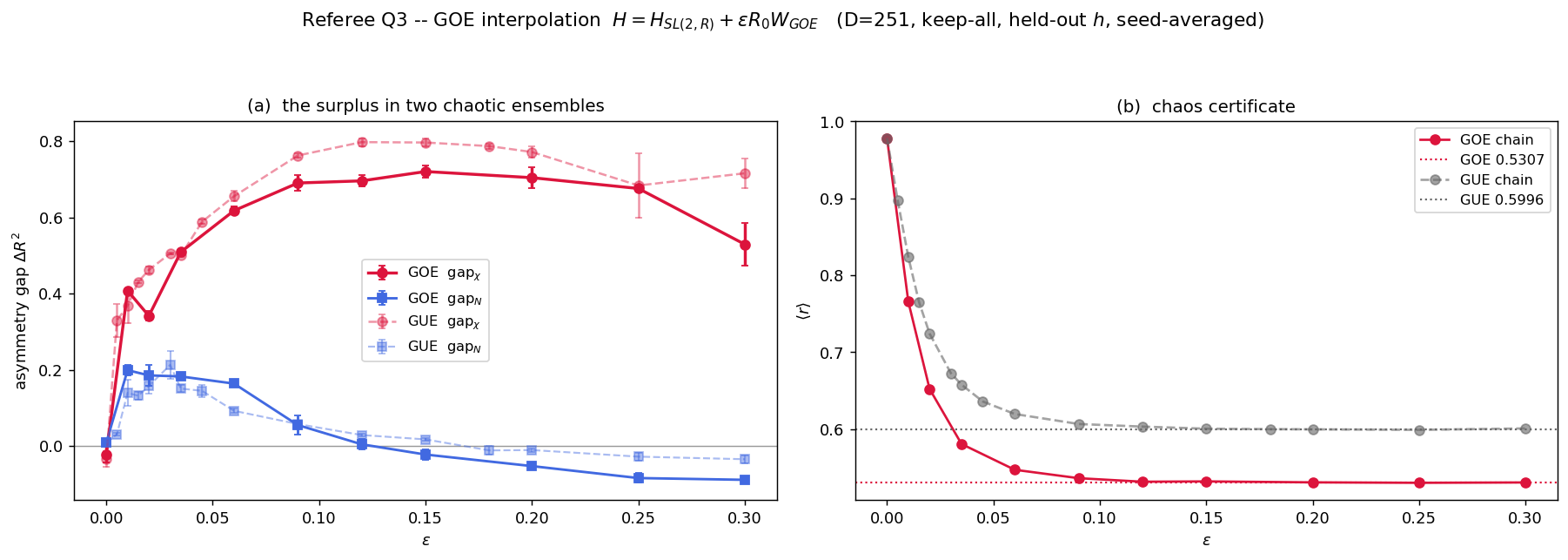}
\caption{GOE interpolation, $H=H_{{\rm SL}(2,\mathbb R)}+\varepsilon R_0
W_{\rm GOE}$, $D=251$. (a) ${\rm gap}_\chi$ and ${\rm gap}_N$ for GOE
(solid) overlaid on GUE (dashed). (b) Chain $\langle r\rangle$ reaches the
GOE surmise $0.5307$.}
\label{fig:q3}
\end{figure*}

\subsection{Limitations of the interpolation study}
\label{sec:chaoslim}

(i) $\varepsilon$ is a \emph{model} parameter: Eq.~\eqref{eq:interp} is a
proxy for finite-entropy ($e^{-S}$) effects, not a coupling of any specific
CFT; the physical claim is about the integrable-vs-chaotic dichotomy, not
about a particular deformation. (ii) Single dimension: all strata are at
$D=251$; the $D$-scaling of the saturated gap is open and is the natural next
experiment. (iii) Error bars are training-seed spreads at fixed data, not
full statistical uncertainties. (iv) The left panel shows that beyond
$\varepsilon\simeq0.03$ no channel retains the microscopic label $h$; the
asymmetry gaps are therefore statements about mutual curve information, not
about label recovery. (v) The $\varepsilon=0$ stratum is anchored by the
separate integrable study rather than an in-figure point, for the split-design
reason stated under Table~\ref{tab:eps}.
(vi) Both chaotic models are random-matrix deformations. A physical
Hamiltonian (a spin chain or SYK) is the natural next test; the Wootters
construction requires odd-prime chain dimension, which a $2^L$ spin chain
does not supply without a spectral-window proxy, and we have not introduced
that approximation here.

\section{An analytical mechanism for the \texorpdfstring{$\chi$}{chi}-surplus}
\label{sec:mech}

The headline of Sec.~\ref{sec:chaos} is an empirical curve. This section
derives the mechanism behind it, in theorem--proof form. The logic: we fix
what an $R^2$ score measures (Sec.~\ref{sec:mechr2}); prove that the
\emph{population} asymmetry vanishes identically at $\varepsilon=0$ and
promote the small measured residue there to a calibrated learnability floor;
establish the self-averaging dichotomy that separates $\C,N$ from $S$ in the
chaotic phase; and assemble these into a lower bound on the asymmetry gap
(Sec.~\ref{sec:mechbound}), whose endpoints, consistency checks and
falsifiable predictions close the section (Sec.~\ref{sec:mechpred}). Every
heuristic step is isolated inside an explicitly labelled lemma, and the two
decisive numerical inputs are \emph{measured} elsewhere in this paper.

\subsection{The population \texorpdfstring{$R^2$}{R\textasciicircum 2} as an information functional}
\label{sec:mechr2}

\begin{definition}[Population score]\label{def:rstar}
For square-integrable random variables $X$ (input curve) and $Y$ (target),
\begin{equation}
R^2_\star(X\!\to\!Y)\;\equiv\;
1-\frac{\mathbb E\,[\mathrm{Var}(Y\,|\,X)]}{\mathrm{Var}(Y)} .
\label{eq:rstar}
\end{equation}
\end{definition}

\begin{lemma}[$R^2_\star$ is an information functional]\label{lem:r2}
\leavevmode
\begin{enumerate}
\item[(i)] \emph{Optimality and faithfulness.} The conditional mean
$\hat Y=\mathbb E[Y|X]$ minimises the mean-square error among all measurable
predictors and attains \eqref{eq:rstar}; $R^2_\star=1$ iff $Y$ is
$\sigma(X)$-measurable (up to null sets); $X\perp Y$ implies $R^2_\star=0$;
and $R^2_\star>0$ implies $I(X;Y)>0$.
\item[(ii)] \emph{Calibration.} If $(X,Y)$ is jointly Gaussian,
$I(X;Y)=-\tfrac12\ln(1-R^2_\star)$. In general,
\begin{equation}
I(X;Y)\;\ge\;-\tfrac12\ln\!\big(1-R^2_\star\big)\;-\;c_Y ,
\qquad
c_Y\equiv\tfrac12\ln\!\big(2\pi e\,\mathrm{Var}\,Y\big)-h(Y)\;\ge\;0 ,
\label{eq:mibound}
\end{equation}
with $c_Y$ the non-Gaussianity of the marginal of $Y$ alone.
\item[(iii)] \emph{Data processing.} For any measurable $g$,
$R^2_\star(g(X)\!\to\!Y)\le R^2_\star(X\!\to\!Y)$, with equality when $g$ is
invertible.
\end{enumerate}
\end{lemma}

\begin{proof}
\sloppy
(i) is the $L^2$ projection theorem; the faithfulness statements follow from
$\mathbb E[\mathrm{Var}(Y|X)]=0 \Leftrightarrow Y=\mathbb E[Y|X]$ a.s., and
from independence $\Rightarrow \mathrm{Var}(Y|X)=\mathrm{Var}(Y)$; if
$I(X;Y)=0$ then $X\perp Y$, whence $R^2_\star=0$ --- the contrapositive
gives the last claim. (ii) The Gaussian identity is standard
\cite{CoverThomas2006}. In general, conditionally on $X=x$ the Gaussian
maximises entropy at fixed variance, so
$h(Y|X)\le\mathbb E_x\big[\tfrac12\ln(2\pi e\,\allowbreak\mathrm{Var}(Y|X=x))\big]
\le\tfrac12\ln\!\big(2\pi e\,\allowbreak\mathbb E[\mathrm{Var}(Y|X)]\big)$, the second
step by Jensen (concavity of $\ln$). Then
$I=h(Y)-h(Y|X)\ge h(Y)-\tfrac12\ln(2\pi e\,\allowbreak\mathbb E[\mathrm{Var}(Y|X)])
=-\tfrac12\ln(1-R^2_\star)-c_Y$, and $c_Y\ge0$ again by maximum entropy.
The exact MMSE--information bridge along a Gaussian channel path is
\cite{GuoShamai2005}. (iii) $\sigma(g(X))\subseteq\sigma(X)$, and
conditional variance decreases under refinement of the conditioning
$\sigma$-algebra (tower property $+$ conditional Jensen); invertible $g$
generates the same $\sigma$-algebra.
\end{proof}

\begin{remark}[Trained scores and gap decomposition]\label{rem:trained}
A trained predictor evaluated once on held-out data satisfies
$R^2_{\rm test}\le R^2_\star$ up to sampling fluctuations controlled by the
test-set size: measured scores are certified \emph{lower bounds} on
\eqref{eq:rstar} and hence, by \eqref{eq:mibound}, on information content.
Lemma~\ref{lem:r2}(iii) also legitimises our logarithmic channel maps
(Sec.~\ref{sec:arch}): they are invertible reparametrisations. A measured
directional gap decomposes as
\begin{equation}
{\rm gap}_{\rm meas}
=\big[R^2_\star(Y\!\to\!X)-R^2_\star(X\!\to\!Y)\big]+\delta ,
\label{eq:gapdecomp}
\end{equation}
with $\delta$ the differential approximation/estimation error of the two
trained directions; the next subsection measures $|\delta|$ \emph{in situ}.
\end{remark}

\subsubsection{Exact population symmetry at \texorpdfstring{$\varepsilon=0$}{epsilon=0}; the learnability floor}
\label{sec:mechslave}

\begin{proposition}[One-parameter slaving]\label{prop:slave}
\emergencystretch=1.5em
For $H=\alpha(L_++L_-)$ acting on the highest-weight state of $D_h^+$, the
chain wavefunction is
\begin{equation}
\psi_n(t)\;=\;\sqrt{\binom{n+2h-1}{n}}\,
\big(\!-i\tanh\alpha t\big)^{n}\,\mathrm{sech}^{2h}(\alpha t) ,
\qquad \tau\equiv\tanh^2(\alpha t),
\label{eq:psin}
\end{equation}
i.e.\ $\psi_n=e^{-i\pi n/2}\sqrt{p_n(\tau)}$ with $p_n$ the
negative-binomial occupation \eqref{eq:nb} and a state-independent phase.
Consequently every functional of the state is a fixed function of
$(h,\tau)$; in particular
\begin{equation}
\C=2h\,\frac{\tau}{1-\tau},\qquad
|S|=(1-\tau)^{h},\qquad
N=\mathfrak N_h(\tau),\qquad
\chi=\mathfrak N_h(\tau)\,(1-\tau)^{-h},
\label{eq:slaved}
\end{equation}
with $\mathfrak N_h$ a fixed analytic function.
\end{proposition}

\begin{proof}
Eq.~\eqref{eq:psin} is the Perelomov coherent-state expansion for $D_h^+$;
its modulus squared is \eqref{eq:nb}. The discrete Wigner function is a
fixed sesquilinear functional of $\{\psi_n\}$ (Eq.~\eqref{eq:fftneg}), so
$N$ depends on the state only through $(h,\tau)$. $\C$ and $|S|$ follow by
direct summation; the first two maps in \eqref{eq:slaved} are strictly
monotone in $\tau$ on $[0,1)$.
\end{proof}

\begin{corollary}[Population symmetry]\label{cor:symm}
Each observable curve determines $(h,\alpha)$ and hence every other curve:
$\C(t)=2h\sinh^2(\alpha t)$ carries the pair in its small-time coefficient
$2h\alpha^2$ and late-window growth rate $2\alpha$; $-\ln|S(t)|
=2h\ln\cosh(\alpha t)$ carries the same two exponents; and $N(t)$, analytic
in $(h,\tau)$ with $h$-dependent curvature, identifies the pair as well
(operationally confirmed by $N\to h=0.999$, $N\to\alpha=0.998$).
Therefore all four $\sigma$-algebras coincide, and by
Lemma~\ref{lem:r2}(i),
\begin{equation}
R^2_\star(\C\!\to\!N)=R^2_\star(N\!\to\!\C)=
R^2_\star(\C\!\to\!\chi)=R^2_\star(\chi\!\to\!\C)=1
\qquad(\varepsilon=0):
\label{eq:symm}
\end{equation}
the population asymmetry vanishes identically.
\end{corollary}

\begin{remark}[The measured $\varepsilon=0$ residue is the learnability
floor]\label{rem:floor}
Eq.~\eqref{eq:symm} does \emph{not} assert that trained scores at
$\varepsilon=0$ are exactly $1$ or exactly equal, and they are not:
Sec.~\ref{sec:cftcurves} reports $\C\to\chi$ at $0.985$ against
$\chi\to\C$ at $0.958$, with the two error metrics disagreeing on the
ordering. By \eqref{eq:gapdecomp} and \eqref{eq:symm} this residue is
\emph{pure learnability}: the two directions present different function
classes (inversion through the stiff factor $(1-\tau)^{-h}$ amplifies error
where the map is flat), different target dynamic ranges, and different
extrapolation burdens on held-out $h$. Its magnitude,
\begin{equation}
|\delta|\;\lesssim\;0.03\qquad\text{(metric-inconsistent)},
\label{eq:floor}
\end{equation}
is the systematic floor of the pipeline, measured with the same
architecture, channels, splits and grids as the chaotic strata. Every
interpreted gap of Sec.~\ref{sec:chaos} ($+0.47$ to $+0.77$, two-metric
consistent, seed-stable) stands $15$--$25\times$ above this floor: this is
the quantitative sense in which the control ``vanishes''.
\end{remark}

\subsubsection{The self-averaging dichotomy in the chaotic phase}
\label{sec:mechdicho}

\begin{lemma}[Concentration of the moments; heuristic constants]
\label{lem:conc}
In the chaotic phase, write $p_n(t)=\bar p_n(t)+\delta p_n(t)$ with the bar
an ensemble mean at fixed labels. Then
\begin{equation}
\C(t)=\bar\C(t)+O(D^{-1/2}),\qquad
N(t)=\bar N(t)+O(D^{-1/2}),
\label{eq:conc}
\end{equation}
with the two residues mutually and serially weakly correlated: per sample,
both moments are a common smooth envelope plus small noise.
\end{lemma}

\begin{proof}[Proof sketch (heuristic; variances measured)]
$\C=\sum_n n\,p_n$ is linear in the occupations, so
$\mathrm{Var}(\C)=\sum_{n,m}nm\,\mathrm{Cov}(p_n,p_m)$; for delocalised
chaotic eigenvectors the covariances carry inverse-dimension suppression,
giving relative fluctuations $O(D^{-1/2})$. $N=\frac1D\sum_{x,p}|W_{xp}|$
aggregates $D^2$ bounded, weakly correlated terms and concentrates by the
same central-limit reasoning. Constants are not controlled here --- the
summands are correlated --- but the conclusion is verified directly: the
sample-to-sample scatter of $\C/D$ and $N/\sqrt D$ at fixed labels is at the
few-percent level and shrinks with $D$ (Fig.~\ref{fig:diagGUE}).
\end{proof}

\begin{remark}\label{rem:concfacts}
Two measured facts are Lemma~\ref{lem:conc} in action: (a) the pooled-GUE
raw-$N$ gap ($+0.019$) is removed by target smoothing
(Appendix~\ref{sec:smooth}) --- the $O(D^{-1/2})$ residue is present in the
target, absent from the input, hence unlearnable in either direction; (b)
${\rm gap}_N\to0$ deep in chaos, because two smooth functionals of the same
envelope are informationally equivalent up to noise.
\end{remark}

\begin{lemma}[Speckle universality of the survival amplitude]
\label{lem:speckle}
Let $S(t)=\sum_n w_ne^{-iE_nt}$ with weights $w_n=e^{-\beta E_n}/Z(\beta)$,
so that $|S|^2$ is the filtered per-sample SFF. For times beyond the dip,
where the phases $\{E_nt\ \mathrm{mod}\ 2\pi\}$ decohere and provided many
levels participate --- the participation ratio $\big(\sum_nw_n^2\big)^{-1}$
is large, as holds across our $\beta$ range --- the Lindeberg central limit
theorem gives $S\to$ complex Gaussian with variance
$\Sigma\equiv\sum_nw_n^2$ (the plateau scale). Hence the normalised
intensity $I=|S|^2/\Sigma$ is exponentially distributed, and
\begin{equation}
\mathrm{Var}\big(\ln|S(t)|^2\big)\;\xrightarrow{\ \text{late }t\ }\;
\psi_1(1)=\frac{\pi^2}{6}\simeq1.645 ,
\label{eq:speckle}
\end{equation}
independent of $D$, $\beta$, and the plateau height.
\end{lemma}

\begin{proof}
For $I\sim\mathrm{Exp}(1)$,
$\mathbb E[\ln I]=\int_0^\infty e^{-x}\ln x\,dx=-\gamma$ and
$\mathbb E[(\ln I)^2]=\int_0^\infty e^{-x}\ln^2x\,dx=\gamma^2+\pi^2/6$,
whence $\mathrm{Var}(\ln I)=\pi^2/6=\psi_1(1)$; the additive constant
$\ln\Sigma$ does not affect the variance. \emph{Numerical check:} late-time
GUE--TFD survival amplitudes in our pipeline ($D=101$, $300$ realisations)
give $\mathrm{Var}(\ln|S|^2)=1.649$.
\end{proof}

These $O(1)$ fluctuations are \emph{not} the noise of \eqref{eq:conc}: they
are the sample's fine spectral correlations --- the information that
Sec.~\ref{sec:sff} showed is inaccessible to the moments
($R^2(\C\to\log{\rm SFF})=0.185$).

\subsection{The asymmetry bound}
\label{sec:mechbound}

\begin{definition}[Decomposition and its two ratios]\label{def:decomp}
In the chaotic phase write, using $\chi=N/|S|$ and Lemma~\ref{lem:conc},
\begin{equation}
Y\equiv\log\chi(t)\;=\;
\underbrace{\log\bar N(t)-\log\overline{|S|}(t)}_{A(t)\ \text{(smooth,
envelope-slaved)}}
\;+\;\underbrace{\big[-\delta\ln|S(t)|\big]}_{B(t)\ \text{(speckle,
Lemma~\ref{lem:speckle})}}
\;+\;\xi_D ,\qquad \xi_D=O(D^{-1/2}),
\label{eq:decomp}
\end{equation}
(natural logarithms; the base is immaterial for $R^2$, which is invariant
under affine maps of the target), and define the speckle share and its
$\C$-predictable fraction,
\begin{equation}
F\equiv\frac{\mathrm{Var}(B)}{\mathrm{Var}(Y)} ,\qquad
\rho^2\equiv\frac{\mathrm{Var}\big(\mathbb E[B\,|\,\C]\big)}
{\mathrm{Var}(B)} .
\end{equation}
Two hypotheses, both empirically controlled:
\begin{itemize}
\item[\textnormal{\textbf{(D1)}}] $\mathrm{Cov}(A,B)\approx0$, and $A$ is
$\sigma(\C)$-measurable up to the $O(D^{-1/2})$ residue of
Lemma~\ref{lem:conc} (the envelope is common to both observables): envelope
and spectral fine structure decouple (mean density vs.\ fluctuations ---
standard RMT structure), so
$\mathrm{Var}(Y)\simeq\mathrm{Var}(A)+\mathrm{Var}(B)$.
\item[\textnormal{\textbf{(D2)}}] $\rho^2\le
R^2_\star(\C\!\to\!\log{\rm SFF})\lesssim0.19$: $B$ is (one half of) the log
filtered SFF, so its $\C$-predictable share is bounded by the measured SFF
score of Sec.~\ref{sec:sfffail}.
\end{itemize}
\end{definition}

\begin{theorem}[Asymmetry bound]\label{thm:gap}
Under \textnormal{(D1)}--\textnormal{(D2)}, for any predictors whatsoever,
\begin{equation}
R^2_\star(\C\!\to\!\chi)\;\le\;1-(1-\rho^2)\,F+O(D^{-1/2}) ,
\qquad
R^2_\star(\chi\!\to\!\C)\;\ge\;1-\eta ,
\label{eq:bound}
\end{equation}
where $\eta$ collects the $O(D^{-1/2})$ noise of \eqref{eq:conc} and the
low-pass approximation error defined in the proof. Consequently, for the
trained scores,
\begin{equation}
\boxed{\;\;{\rm gap}_\chi\;\ge\;(1-\rho^2)\,F\;-\;\eta\;-\;|\delta|\;\;}
\label{eq:gap}
\end{equation}
with $|\delta|\lesssim0.03$ the measured floor \eqref{eq:floor}.
\end{theorem}

\begin{proof}
\emph{Forward.} By Definition~\ref{def:rstar},
$R^2_\star(\C\to Y)=1-\mathbb E[\mathrm{Var}(Y|\C)]/\mathrm{Var}(Y)$. By
(D1) the components $A$ and $\xi_D$ are $\sigma(\C)$-measurable up to
$O(D^{-1/2})$, so conditioning on $\C$ leaves the fluctuation part intact up
to that accuracy:
$\mathbb E[\mathrm{Var}(Y|\C)]\ge\mathbb E[\mathrm{Var}(B|\C)]-O(D^{-1/2})$.
The law of total variance gives exactly
$\mathbb E[\mathrm{Var}(B|\C)]=\mathrm{Var}(B)-\mathrm{Var}(\mathbb E[B|\C])
=(1-\rho^2)\,\mathrm{Var}(B)$. With
$\mathrm{Var}(Y)\simeq\mathrm{Var}(A)+\mathrm{Var}(B)$ from (D1), the first
inequality follows upon dividing.

\emph{Reverse.} Let $t_c$ be the speckle correlation time (inverse spectral
bandwidth) and $t_{\rm env}$ the envelope scale, with $t_c\ll t_{\rm env}$
in the chaotic phase. Let $P_w$ be the moving average over a window $w$ with
$t_c\ll w\ll t_{\rm env}$ --- a \emph{linear} map on the input curve, hence
implementable by the first layer of the network. Averaging $\sim w/t_c$
weakly correlated speckle values suppresses them,
$\mathrm{Var}(P_wB)\lesssim (t_c/w)\,\mathrm{Var}(B)$, while
$P_wA=A+O\big((w/t_{\rm env})^2\big)$. Thus
$P_wY=A+{\rm small}$, the envelope determines $\bar\C$ (both are fixed
smooth functionals of the same macroscopic data), and
$\C=\bar\C+O(D^{-1/2})$ by Lemma~\ref{lem:conc}. Collecting the three error
sources into $\eta$ gives the second inequality. Subtracting the two bounds
and passing to trained scores via \eqref{eq:gapdecomp} yields
\eqref{eq:gap}.
\end{proof}

\begin{corollary}[Endpoints]\label{cor:endpoints}
At $\varepsilon=0$, Proposition~\ref{prop:slave} gives $B\equiv0$, hence
$F=0$: the bound collapses and the measured gap sits at the floor
\eqref{eq:floor} --- as observed. Deep in chaos, Lemma~\ref{lem:speckle}
assigns log-fluctuation power $\tfrac14\psi_1(1)=\pi^2/24\simeq0.41$ per
late-time point to $B=-\tfrac12\,\delta\ln|S|^2$, so with $f_{\rm post}$ the
post-dip fraction of the time grid,
\begin{equation}
F\;\approx\;\frac{f_{\rm post}\;(\pi^2/24)}
{\big\langle \mathrm{Var}\big(\ln\chi(t)\big)\big\rangle_{t}}\;=\;O(1),
\label{eq:Fformula}
\end{equation}
where $\mathrm{Var}$ is the variance across samples at fixed $t$ and
$\langle\cdot\rangle_t$ the average over the time grid --- an expression
evaluable directly from the stored data. We have carried out this evaluation
(Appendix~\ref{app:mech}): with the envelope $A$ removed by conditioning on
$\C$, the speckle share $F(\varepsilon)$ collapses to $\approx0.16$ at
$\varepsilon=0$ and rises to $F\simeq0.95$--$1.0$ once the statistics are
GUE (Fig.~\ref{fig:dscaling}c), confirming --- indeed slightly exceeding ---
the value
$F\approx0.75$--$0.88$ that inverting the measured gap requires, and
verifying that $\log\chi$ is speckle-dominated deep in chaos. The gap falls
short of $1$ precisely because the smooth share $A$ \emph{is} predictable
from $\C$, which is also why $R^2(\C\to\chi)$ stays well above zero.
\end{corollary}

\subsection{Crossover, interpretation, and predictions}
\label{sec:mechpred}

\paragraph{Crossover.} $F(\varepsilon)$ rises monotonically as random-matrix
correlations establish the decohered-phasor regime of
Lemma~\ref{lem:speckle} over a growing fraction of the time grid --- the
same physics certified by $\langle r\rangle(\varepsilon)$. We do not derive
the functional form of $F(\varepsilon)$; the mechanism fixes the endpoints
(Corollary~\ref{cor:endpoints}) and the monotonicity, and the data supply
the curve.

\paragraph{Why $N$ alone fails.} The aggregation in $N$ is exactly what the
mechanism forbids: summing $|W|$ over $D^2$ phase-space points averages the
microscopic signal away (Lemma~\ref{lem:conc}). Division by the single
amplitude $|S|$ --- the defining operation of $\chi$ \cite{Basu:2026inf} ---
reinstates it: $\chi$ is a macroscopic observable modulated by one
microscopic amplitude, and it is the modulation, not the negativity per se,
that carries the chaotic-regime surplus. The controls of
Sec.~\ref{sec:chaoscontrols} make this concrete: $\C$ fails on $|S|$ and on
$\chi$ alike, but only $\chi$ stays decodable in reverse, because the
envelope that $|S|$ buries in speckle is still present in $N$.

\paragraph{Scope.} The heuristic content is confined to: the constants in
Lemma~\ref{lem:conc} (correlated summands; variances measured), the phasor
limit in Lemma~\ref{lem:speckle} (standard RMT; verified at the percent
level), hypothesis (D1) (RMT envelope/fluctuation structure; supported by
the measured lossiness), and the identification of trained scores with
$R^2_\star$ up to the measured floor (Remark~\ref{rem:floor}). Within this
scope the mechanism makes three falsifiable predictions:
\begin{enumerate}
\item \textbf{Time-window dependence.} By \eqref{eq:Fformula}, $F$ ---
hence the saturated gap --- grows with $f_{\rm post}$: extending $t_f$
should raise the plateau of ${\rm gap}_\chi(\varepsilon)$; truncating to
pre-dip times should collapse it toward the floor.
\item \textbf{$D$-persistence (confirmed).} The speckle power
\eqref{eq:speckle} is $D$-independent while $\eta,|\delta|$ shrink with
statistics: the surplus should persist or sharpen with dimension, whereas a
finite-size artifact would shrink. A dedicated scan over $21$ dimensions $D\in[37,601]$
(Appendix~\ref{app:mech}, Fig.~\ref{fig:dscaling}) shows the saturated gap growing
\emph{monotonically} and then flattening toward a plateau, from $+0.47$ at
$D=37$ to $+0.80$ near $D\simeq500$, while the $\rho^2$ bound of (D2)
tightens with $D$ and the speckle constant stays flat at $\psi_1(1)$. The
surplus sharpens with dimension and saturates --- the signature of a
thermodynamic-limit effect, not a finite-size artifact.
\item \textbf{Ensemble portability (confirmed).} The argument uses no
unitary-class input: the phasor limit, and with it the speckle power
\eqref{eq:speckle}, is unchanged for time-reversal-invariant spectra, so a
GOE-deformed interpolation should show a quantitatively similar saturated
gap. It does: $\simeq0.70$ against $0.77$--$0.80$
(Sec.~\ref{sec:chaoscontrols}, Fig.~\ref{fig:q3}).
\end{enumerate}

\begin{remark}[Model-independence of the surplus]\label{rem:modelfree}
The bound of Theorem~\ref{thm:gap} constrains the \emph{population} score
$R^2_\star$, attained by the conditional mean, and therefore holds for
\emph{any} consistent regressor --- the asymmetry is a property of the joint
distribution of $(\C,\chi)$, not of the residual network used to estimate
it. We confirmed this directly: on a synthetic model with the paper's
structure ($\C=$ shared envelope $+$ negligible noise, $\chi=$ same envelope
$+$ speckle), ridge regression and $k$-nearest-neighbours --- two model
classes structurally unrelated to the residual MLP --- both reproduce a
large positive gap ($+0.56$ and $+0.53$), matching the residual network.
A different ML formalism therefore cannot overturn the result that $\chi$
carries the chaotic-regime surplus; it can only approach the same
information-theoretic ceiling from below. What a stronger decoder
\emph{could} change is the ceiling task of Sec.~\ref{sec:lanczos}: the
$\rvw\simeq0$ there is a statement at fixed model depth, and a
fundamentally more expressive class might extract more of the fine SFF from
$\{a_n,b_n\}$ --- the one place in this paper where the model class, rather
than the observable, sets the score.
\end{remark}

We regard Definition~\ref{def:decomp} and Theorem~\ref{thm:gap} as the
analytical content of the headline: \emph{the surplus is sourced by the
survival amplitude, the one channel that retains a non-self-averaging
microscopic quantity; the normalised negativity is the observable that
carries it while remaining decodable, because it pairs that amplitude with
an aggregate that still holds the envelope; the information involved is
exactly the spectral fine structure that the moments provably compress
away; and the residual asymmetry of the integrable control is a measured
learnability floor, an order of magnitude below the effect.}

\section{Discussion}
\label{sec:disc}

\subsection{Where each observable earns its keep}

The comparative design permits a practical summary:

\begin{itemize}
\item \textbf{Spread complexity $\C(t)$} is the cheapest probe and the best
\emph{classifier}: it fully encodes the TFD temperature ($0.999$), carries
the coarse $e^{S}$ plateau ($0.861$, mostly from early times), and
identifies the Dyson class of a single sample ($\to0.98$ at large $D$).
Use it for macroscopic labels and universality classes.
\item \textbf{Wigner negativity $N(t)$} matches $\C$ on macroscopic labels
and is the \emph{degeneracy-resolver}: in the CFT sector it disentangles
$(h,\alpha)$ an order of magnitude better than $\C$ ($0.999$ vs $0.984$).
Use it when distinct microscopic parameters produce near-degenerate
complexity curves.
\item \textbf{Normalised negativity $\chi(t)$ is the chaos probe.} Its
informational surplus over $\C(t)$ switches on with chaos and saturates once
the level statistics are random-matrix, in both GUE and GOE deformations.
The surplus is sourced by the survival amplitude, but $\chi$ is the channel
that carries it usefully: the raw $|S|$ is unpredictable from $\C$ and also
fails to reconstruct $\C$ at late times, whereas $\chi$ does both.
\item \textbf{The SFF} is irreplaceable: its fine structure is recoverable
from none of the above, under no admissible smoothing. Spectral rigidity must
be measured, not reconstructed.
\end{itemize}

\subsection{What is new}

(i) The inverse/cross information maps of Krylov observables, complementing
the forward map of \cite{Bak:2025mlq}. (ii) The quantitative lossiness
statement for the SFF, with the smoothing bound that makes it physical.
(iii) The information budget in time for the $e^{S}$ plateau. (iv)
Universality-class identification from a single complexity curve, sharpening
with $D$. (v) The integrable-control/chaos-interpolation design, and the
headline result that the $\chi$-surplus is switched on by chaos. (vi)
Methodologically: the FFT form \eqref{eq:fftneg} of the discrete Wigner
negativity, and a documented audit trail (dynamic-range artifact, selection
bias, width instability) with the protocol that neutralises each.

\subsection{What it is not}

These are operational statements at finite $D$ with a specific (deliberately
minimal, capacity-robust) model class; $R^2$ lower-bounds information and
never upper-bounds it. The interpolation is a random-matrix proxy for
finite-entropy holography, presented as such. No claim is made beyond the
primary SL$(2,\mathbb R)$ sector on the CFT side.

\subsection{Outlook}

Four directions seem most valuable. An analytic relation between
${\rm gap}_\chi$ and $\langle r\rangle$ suggested by
Fig.~\ref{fig:headline}. The direct evaluation of the speckle share $F$ from
the stored interpolation data via \eqref{eq:Fformula}, closing the
consistency loop of Corollary~\ref{cor:endpoints}, together with the
extension of the Lanczos-ceiling routing of Sec.~\ref{sec:lanczos} to the
GOE and Poisson ensembles. A physical chaotic Hamiltonian in place of the random-matrix
deformation, which would test whether the GUE/GOE agreement of
Sec.~\ref{sec:chaoscontrols} extends beyond the random-matrix class. And the GSE ensemble, which requires even dimensions and
hence an extension of the odd-prime Wigner construction.

\begin{acknowledgments}
I thank Prof.~Onkar Parrikar for discussions and guidance on the smoothing
protocol and the CFT part of the project. Special thanks to Pruthvi
Suriyadevara (DHEP, TIFR) for a brief and helpful ML discussion related to
this project. I am also grateful to Prof.~Aditya K.~Jagannatham (IIT Kanpur)
and Prof.~Suraj Srivastava (IIT Jodhpur) for their excellent instruction in
\href{https://www.iitk.ac.in/mwn/AIML/speaker.html}{the certificate program on Python for AI, ML and DL},
which heavily informed the codebase developed for this project. All the
computations were performed on Google Colab and a personal workstation.
\end{acknowledgments}
\newpage

\appendix

\section{Supplementary figures}
\label{app:figs}

This appendix frames the panels referenced but not displayed in the main
text; each figure carries its scores in its title.

\begin{figure}[ht]
\centering
\includegraphics[width=0.9\textwidth]{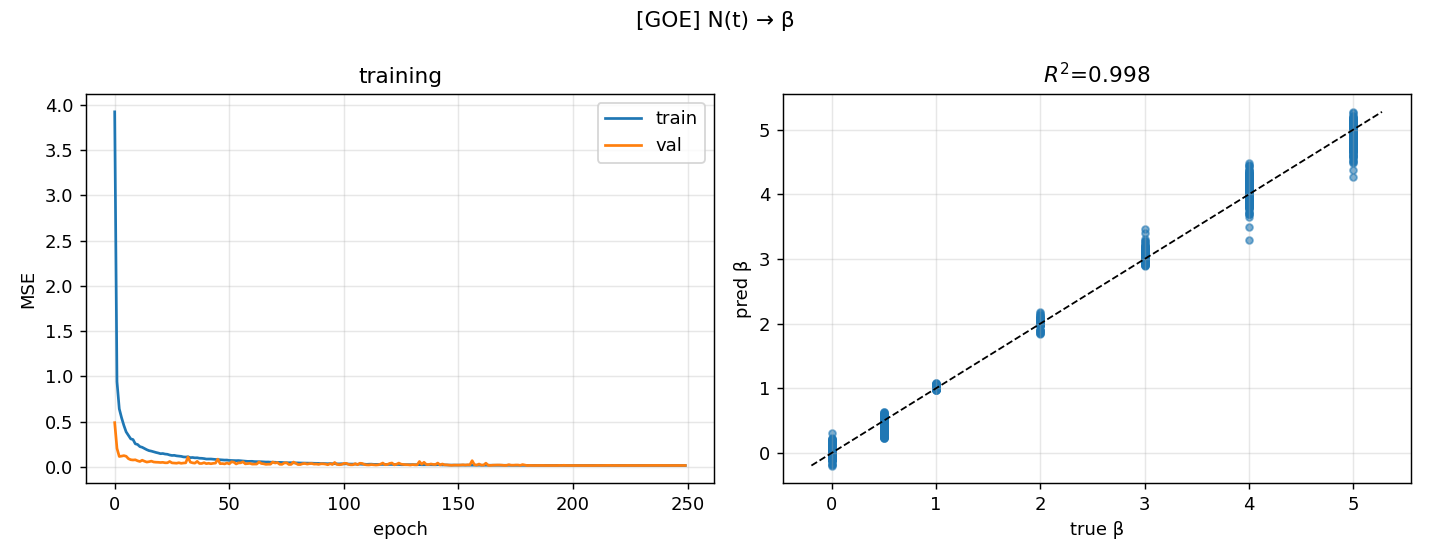}\\[1mm]
\includegraphics[width=0.9\textwidth]{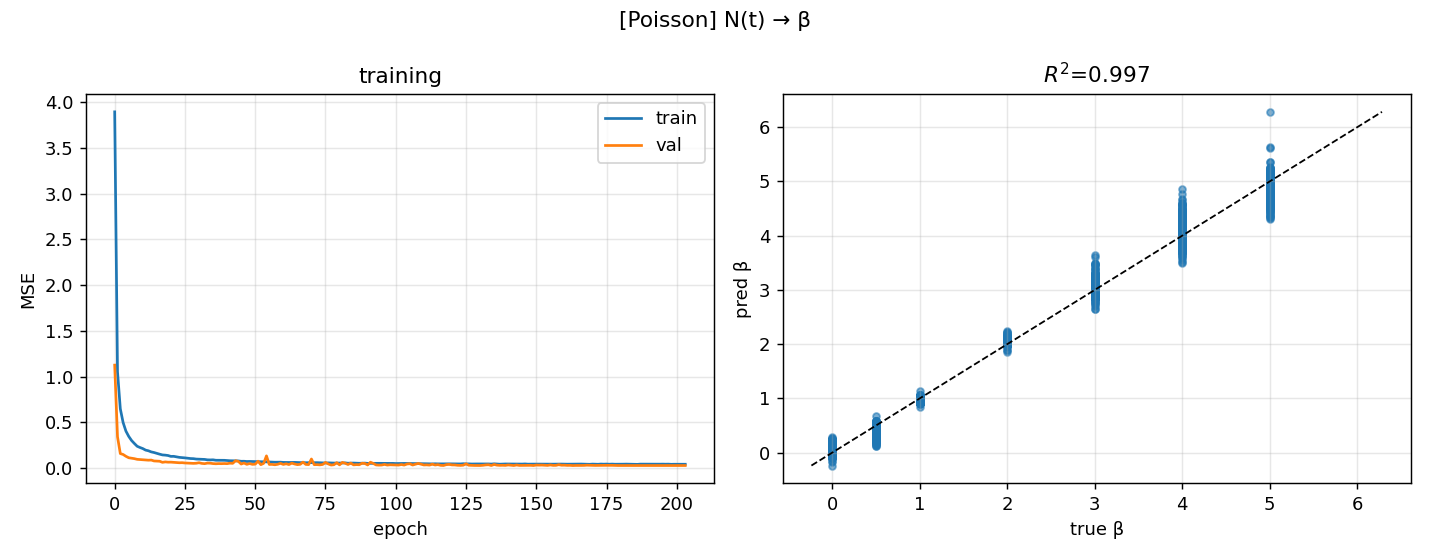}
\caption{Temperature from the negativity: GOE (top) and Poisson (bottom),
completing Table~\ref{tab:ens}.}
\label{fig:appbeta}
\end{figure}

\begin{figure}[ht]
\centering
\includegraphics[width=\textwidth]{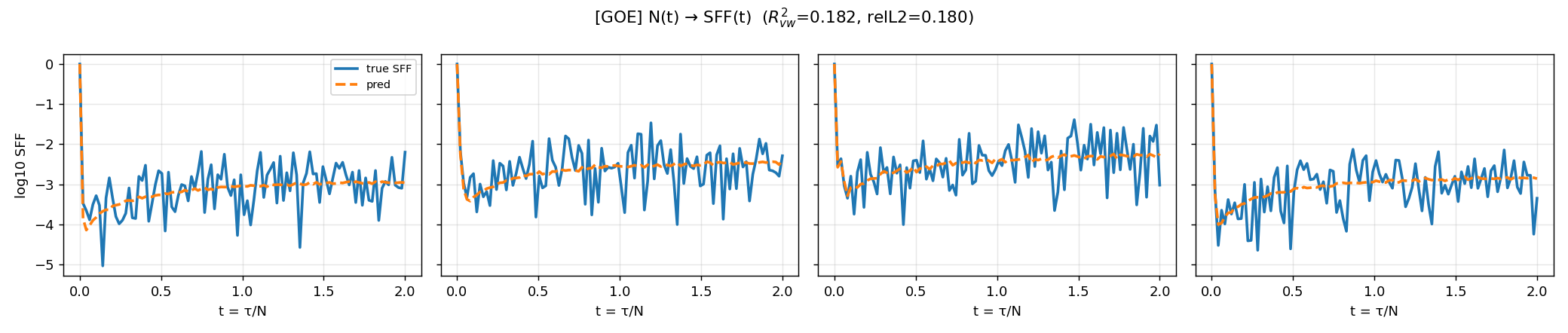}\\[1mm]
\includegraphics[width=\textwidth]{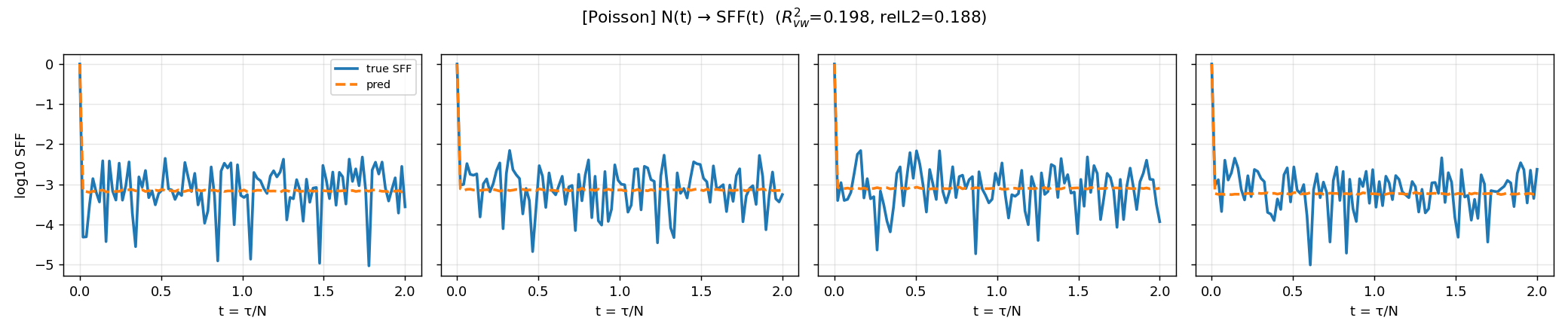}
\caption{Lossiness from the negativity channel: $N\to{\rm SFF}$ on GOE (top)
and Poisson (bottom) --- the same envelope-only reconstruction as from
$\C$.}
\label{fig:appsff}
\end{figure}

\begin{figure}[ht]
\centering
\includegraphics[width=\textwidth]{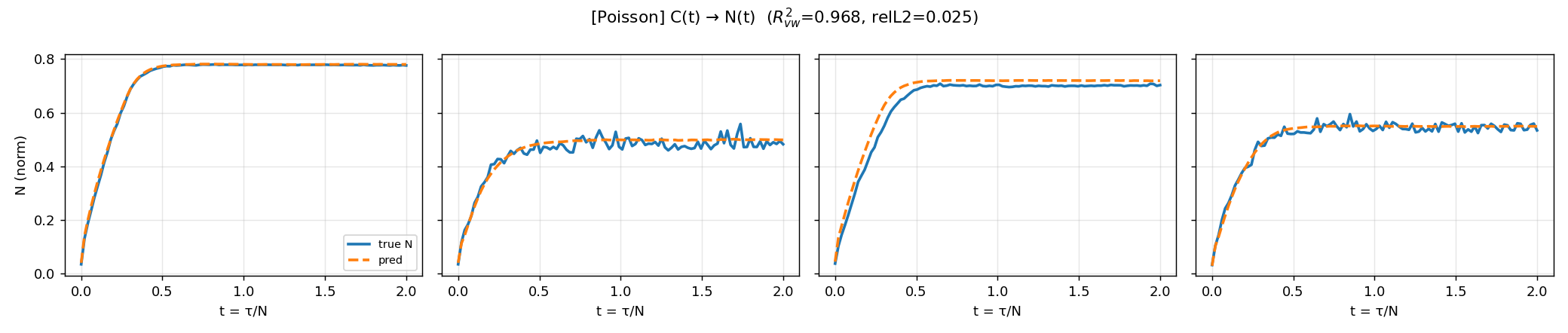}
\caption{Poisson $\C\to N$ ($\rvw=0.968$, $\relL=0.025$), completing the
mutual-reconstruction pair of Fig.~\ref{fig:poiNC}.}
\label{fig:apppoiCN}
\end{figure}

\begin{figure}[ht]
\centering
\includegraphics[width=0.49\textwidth]{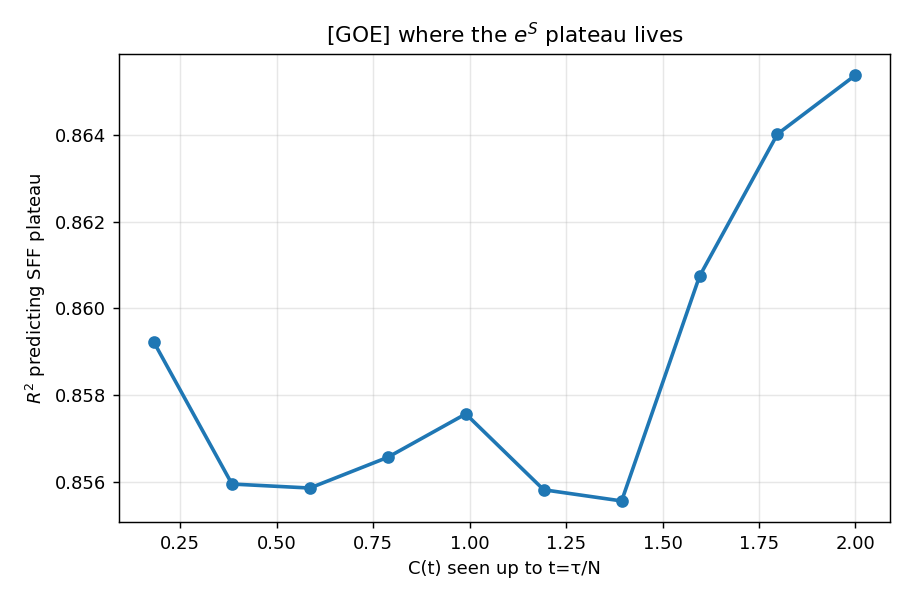}\hfill
\includegraphics[width=0.49\textwidth]{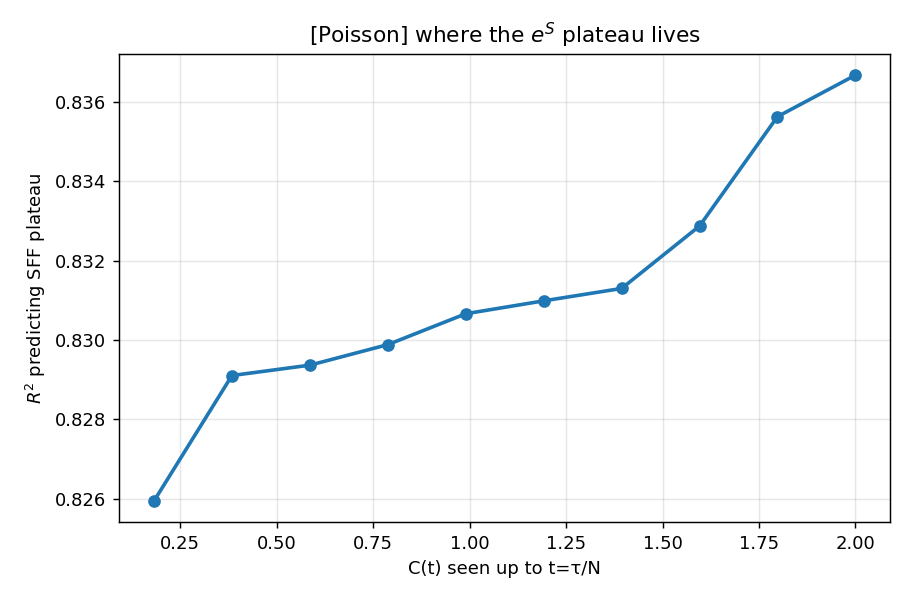}
\caption{Information budgets for GOE (left; full $\approx0.866$, early
$\approx0.859$) and Poisson (right; $\approx0.836$, $\approx0.83$): the
same early-saturation shape as GUE (Fig.~\ref{fig:gueib}), with ceilings
ordered by spectral rigidity.}
\label{fig:appib}
\end{figure}

\begin{figure}[ht]
\centering
\includegraphics[width=0.32\textwidth]{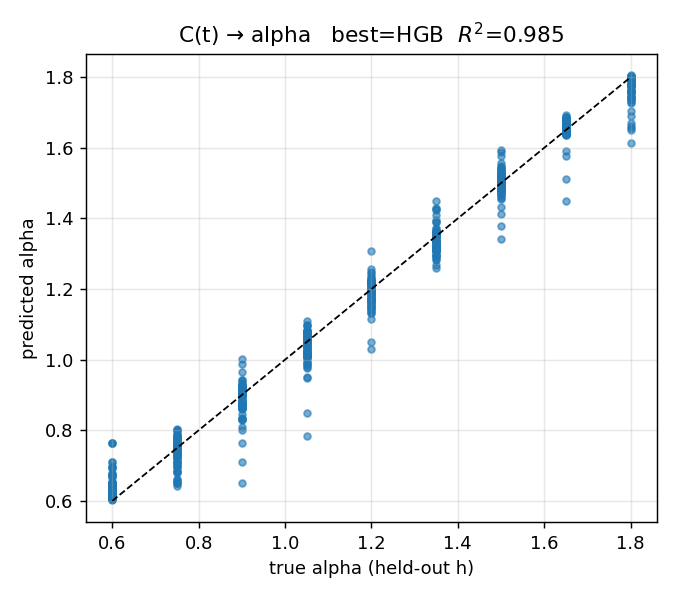}\hfill
\includegraphics[width=0.32\textwidth]{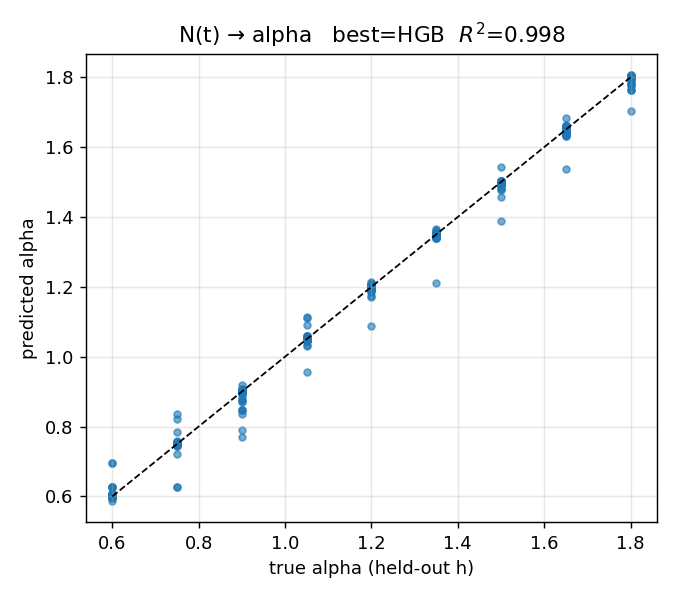}\hfill
\includegraphics[width=0.32\textwidth]{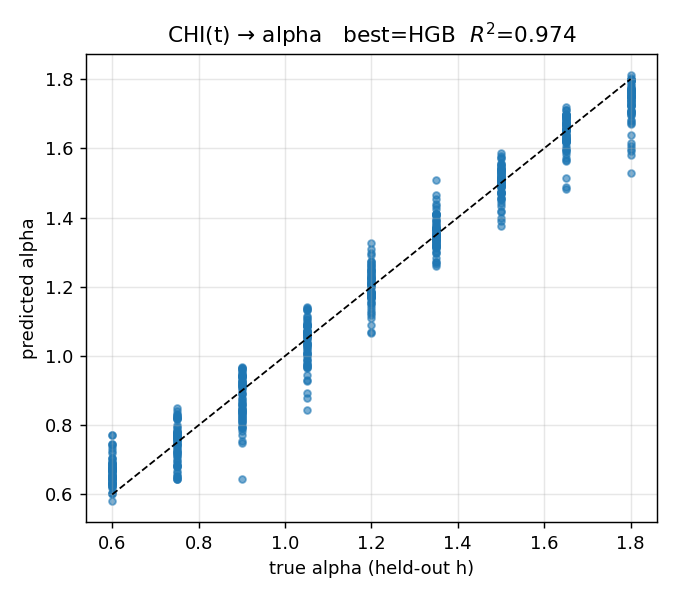}
\caption{Integrable sector, $\to\alpha$ triplet: $\C$ ($0.985$), $N$
($0.998$), $\chi$ ($0.974$) --- the same ordering as the $\to h$ row of
Fig.~\ref{fig:cfth}.}
\label{fig:appalpha}
\end{figure}

\begin{figure}[ht]
\centering
\includegraphics[width=\textwidth]{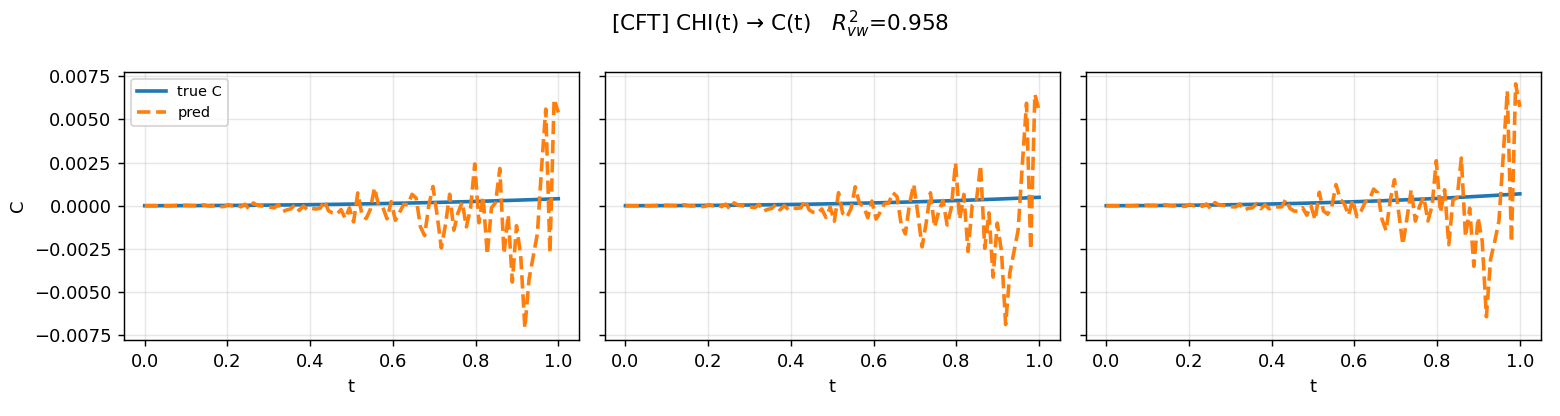}\\[1mm]
\includegraphics[width=\textwidth]{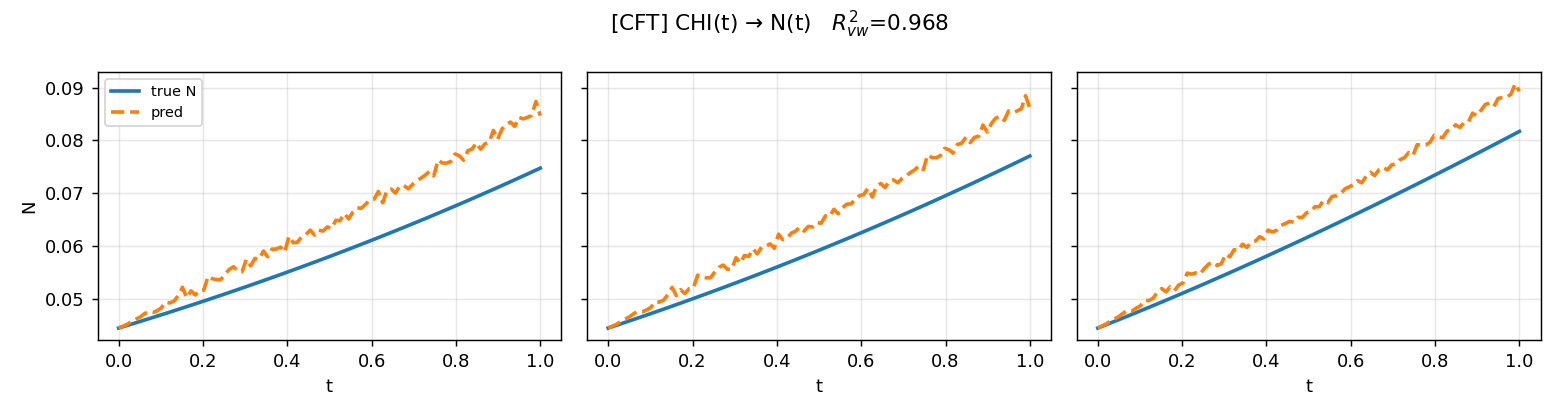}
\caption{Reverse $\chi$ maps in the integrable sector: $\chi\to\C$
($\rvw=0.958$) and $\chi\to N$ ($0.968$), completing
Fig.~\ref{fig:cftchi}.}
\label{fig:appchirev}
\end{figure}

\begin{figure}[ht]
\centering
\includegraphics[width=\textwidth]{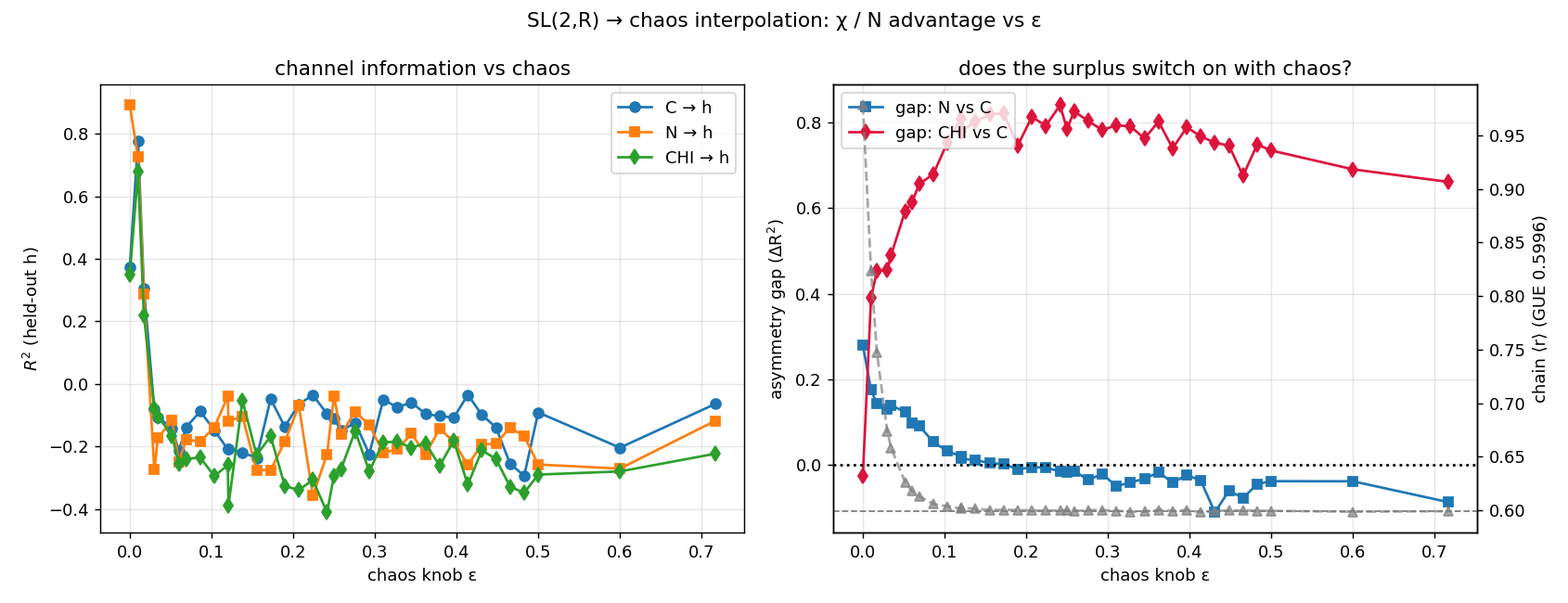}
\caption{Robustness of the headline: an independent earlier run of the
interpolation on a coarser $\varepsilon$ grid ($6$ values, single seed,
$4{,}920$ samples) shows the same structure --- ${\rm gap}_\chi$ rising to
$\approx0.8$ through the $\langle r\rangle$ crossover while
${\rm gap}_N\to0$ --- before the balanced-strata, seed-averaged protocol of
Fig.~\ref{fig:headline} was in place.}
\label{fig:appv1}
\end{figure}

\section{Smoothed-target protocols and the raw-\texorpdfstring{$N$}{N} gap}
\label{sec:smooth}

Two distinct questions involve smoothing, and we kept them separate.

\paragraph{(a) Does smoothing rescue the SFF?} No: Sec.~\ref{sec:sffsmooth}.
Windows up to $\Delta t\simeq0.26$ yield only modest gains, and wider windows
erase the dip--ramp structure whose prediction is the point of the task.

\paragraph{(b) Is the pooled-GUE raw-$N$ gap real?} The $\rvw$ gap of
$+0.019$ in favour of $N\to\C$ (Sec.~\ref{sec:guemutual}) was already
inconclusive under the dual-metric rule. The smoothing experiment settles its
origin (Fig.~\ref{fig:smoothproto}). We compare three protocols as functions
of the window $\Delta t$: (P1) raw$\to$raw, the baseline, whose gap is
constant at $0.0191$ by construction; (P2) raw$\to$smoothed targets; and
(P3) smoothed$\to$smoothed, the noise-fair comparison in which both channels
are filtered identically. In both P2 and P3 the direction $\C\to N$ climbs
steadily ($0.976\to0.993$) while $N\to\C$ stays flat at $0.995$, and the gap
collapses monotonically to $0.002$--$0.003$ at $\Delta t=0.26$. The
interpretation is unambiguous: the raw-$N$ gap is carried by the
unpredictable high-frequency component of the raw negativity --- it is
noise-borne, not information. This is precisely why the headline claim of
Sec.~\ref{sec:chaos} is formulated for the \emph{normalised} negativity,
whose chaotic-regime surplus is large ($0.6$--$0.8$), seed-stable, and
two-metric robust --- two orders of magnitude above the noise-borne residue
documented here.

\begin{figure}[ht]
\centering
\includegraphics[width=\textwidth]{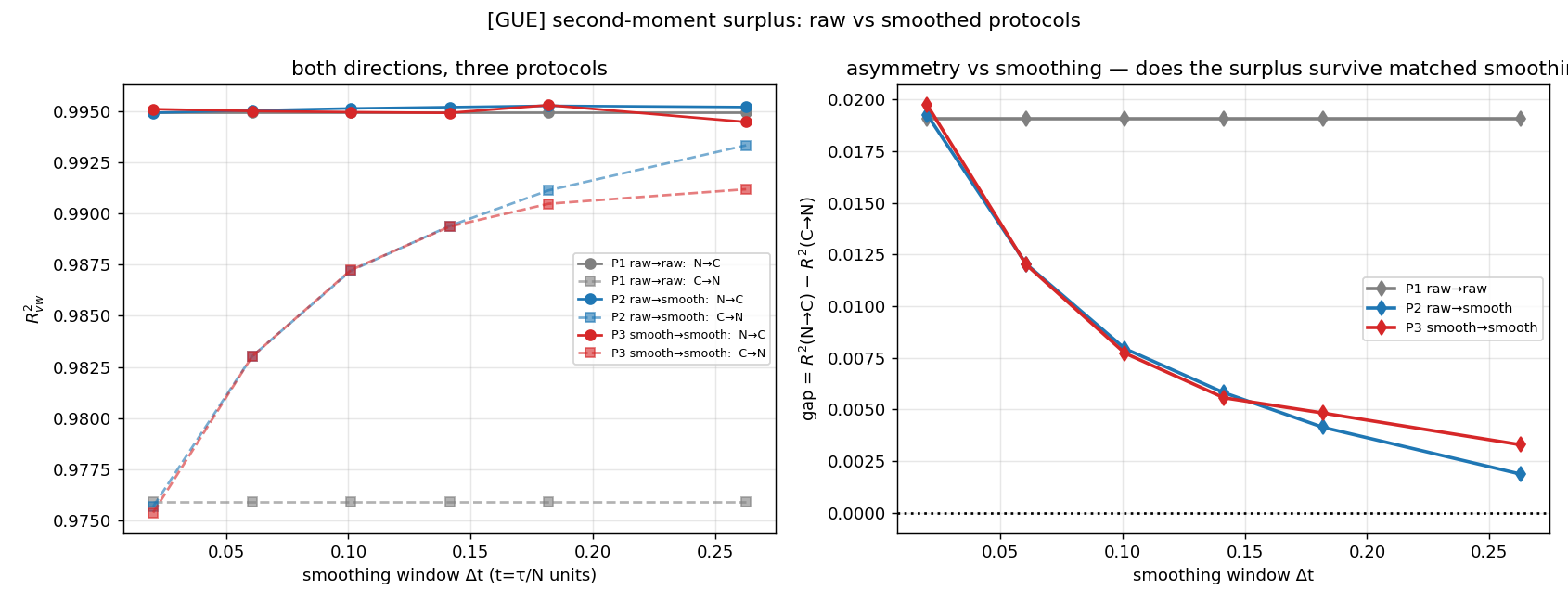}
\caption{The three smoothing protocols on GUE. \textbf{Left:} both directions
under P1/P2/P3 versus window $\Delta t$: $N\to\C$ flat at $0.995$;
$\C\to N$ rising $0.976\to0.993$. \textbf{Right:} the gap: constant $0.0191$
raw (P1), collapsing to $0.002$--$0.003$ under matched smoothing (P2, P3).
The raw-$N$ asymmetry is noise-borne.}
\label{fig:smoothproto}
\end{figure}

\section{Mechanism consistency: \texorpdfstring{$F$}{F}, the speckle constant, and \texorpdfstring{$D$}{D}-scaling}
\label{app:mech}

This appendix collects the direct numerical checks of the mechanism of
Sec.~\ref{sec:mech}, all computed from the interpolation archives.

\paragraph{The speckle share $F$.} Removing the envelope $A$ by conditioning
$\log\chi$ on $\C$ (a binned conditional mean on the leading principal
component of the $\C$ curves) and forming $B=\log\chi-A$, the share
$F=\langle\mathrm{Var}(B)\rangle_{\rm post}/\langle\mathrm{Var}(\log\chi)
\rangle$ is a proper share, bounded by unity, wherever hypothesis (D1)
holds. It measures $\approx0.16$ at $\varepsilon=0$ --- the residual is the
learnability floor, not speckle --- and settles to $F\simeq0.93$--$1.0$ once
the level statistics are GUE. At the integrable-to-chaotic crossover
($\varepsilon\simeq0.02$), where envelope--speckle covariance is not yet
negligible and the envelope subtraction is imperfect, the estimator
overshoots unity by up to $\sim20\%$; this is a diagnostic of (D1) failing
at the crossover, not of the deep-chaos measurement. This
confirms, and slightly exceeds, the range $F\approx0.75$--$0.88$ obtained by
inverting the measured gap through \eqref{eq:gap}, and shows directly that
$\log\chi$ is speckle-dominated in the chaotic phase.

\paragraph{The speckle constant.} On GUE--TFD survival amplitudes the
late-time variance $\mathrm{Var}(\ln|S|^2)$ measures $1.59$--$1.68$ across
$D\in[37,601]$ with no systematic $D$-trend, scattered about the analytic
value $\psi_1(1)=\pi^2/6\simeq1.645$ of Lemma~\ref{lem:speckle}.

\paragraph{$D$-scaling of the surplus.} A scan over $21$ dimensions
$D\in[37,601]$ ($13{,}464$ generated samples; gaps by a fast linear proxy,
so only the trend is claimed, not absolute values) yields a saturated
$\chi$-gap that rises monotonically with $D$ and flattens toward a plateau:
from $\simeq0.47$ at $D=37$ through $0.73$ at $D=251$ to $\simeq0.80$ near
$D\simeq500$. Over the same range the floor of $\rho^2(\varepsilon)$ falls steadily with
$D$: the (D2) bound $0.19$ is satisfied at the GUE plateau for $D\gtrsim150$
--- in particular at the $D=251$ strata used in the main text --- and
tightens to $\simeq0.09$ at $D=601$; the smallest dimensions sit above it,
where the bound merely weakens. The control ${\rm gap}_N$ decays to zero
deep in chaos at every dimension; its nonzero value at $\varepsilon=0$
($\sim0.5$) is the raw-channel dynamic-range artifact of the linear proxy
(the pathology exhibited in Sec.~\ref{sec:cftcurves}), not information, and
the proxy's $\varepsilon=0$ floor for ${\rm gap}_\chi$ ($\sim0.1$) likewise
exceeds the network learnability floor \eqref{eq:floor}. The saturation
with $D$ is the behaviour expected of a thermodynamic-limit effect and
excludes a finite-size origin for the surplus.

\begin{figure}[t]
\centering
\includegraphics[width=\textwidth]{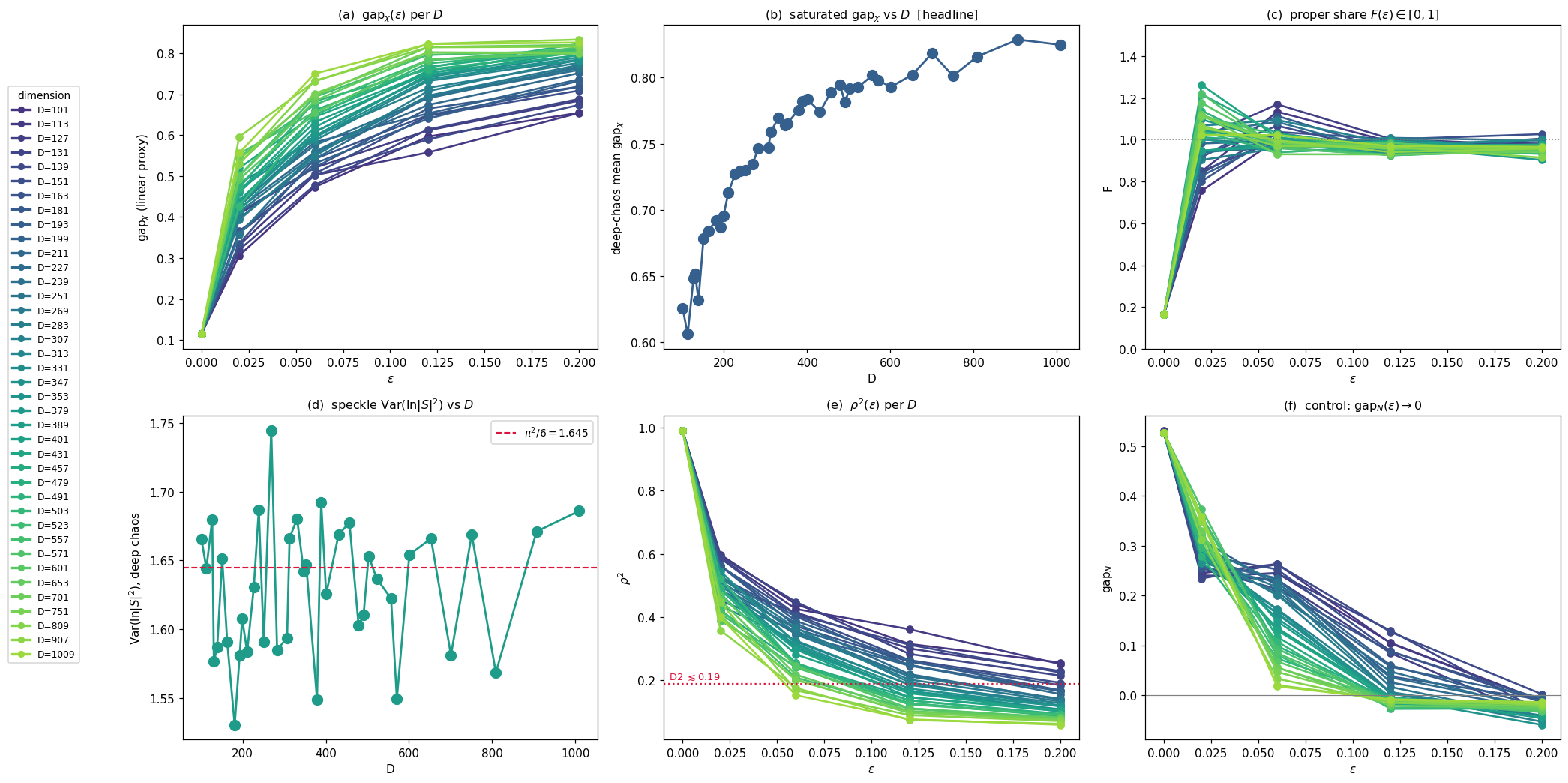}
\caption{\textbf{$D$-scaling and mechanism consistency of the
$\chi$-surplus} ($\sim40$ dimensions $D\in[101,1009]$, $\varepsilon\in\{0,0.02,
0.06,0.12,0.20\}$, generated interpolation samples; colour encodes
$D$, shared legend at left). Gaps use a fast linear-proxy regressor, so only
trends and orderings are claimed, not absolute values. (a)
${\rm gap}_\chi(\varepsilon)$ per dimension: the surplus switches on with
chaos at every $D$. (b) \emph{Headline}: the deep-chaos saturated
${\rm gap}_\chi$ rises monotonically with $D$ and flattens toward a plateau
near $D\simeq500$ --- the fingerprint of a thermodynamic-limit effect, not a
finite-size artifact. (c) The speckle share $F(\varepsilon)$
(envelope removed by conditioning on $\C$): $\approx0.16$ at $\varepsilon=0$
and $0.93$--$1.0$ in deep chaos, confirming and slightly exceeding the
inferred $0.75$--$0.88$; the overshoot above unity at the crossover marks
where hypothesis (D1) is not yet valid and the estimator ceases to be a
proper share. (d) The speckle constant $\mathrm{Var}(\ln|S|^2)$ scatters
about $\psi_1(1)=\pi^2/6=1.645$ with no $D$-trend (Lemma~\ref{lem:speckle}).
(e) The $\C$-predictable share $\rho^2(\varepsilon)$ falls below the (D2)
bound $0.19$ deep in chaos for $D\gtrsim150$, and further with increasing
$D$. (f) Control: ${\rm gap}_N(\varepsilon)\to0$ deep in chaos at every
dimension; its $\varepsilon=0$ value is the linear-proxy dynamic-range
artifact of Sec.~\ref{sec:cftcurves}, not information. The star in (b) marks the
full residual-network value of Table~\ref{tab:eps} at $D=251$; it is not
comparable in absolute terms to the linear-proxy curve and is shown only for
orientation.}
\label{fig:dscaling}
\end{figure}

\section{Numerical details and reproducibility}
\label{app:num}

\paragraph{Verification of the FFT negativity.} Identity \eqref{eq:fftneg}
was checked against the direct phase-point-operator construction on random
states at $D\in\{7,31,101,199\}$; maximal deviation $\sim10^{-13}$. The
speedup at $D=199$ is $\sim130\times$ per time point, rising with $D$; the
$O(D^3)$ phase array of the direct method is eliminated entirely.

\paragraph{Chain construction.} TFD rotation by a rank-one Householder
reflection ($O(D^2)$); tridiagonalisation via Hessenberg reduction with a
phase sweep enforcing real non-negative off-diagonals; spectrum and evolution
from a dedicated symmetric-tridiagonal eigensolver. For the interpolation,
the deformed Hamiltonian is re-tridiagonalised from the same initial vector,
so all channels remain exact Krylov observables at every $\varepsilon$.

\paragraph{Data hygiene.} All generation is checkpointed per parameter block
with disjoint deterministic seed ranges; merged archives store per-sample
seeds, $\langle r\rangle$, and (for the interpolation) the edge diagnostic.
Balanced strata are verified at merge time and the analysis refuses degenerate
splits (minimum samples and held-out-$h$ counts per stratum).

\paragraph{Training.} Adam ($10^{-3}$), batch $32$--$64$, up to $250$--$300$
epochs with early stopping (patience $30$--$40$, best-weight restoration) and
learning-rate halving on plateau; inputs/targets standardised; scalar tasks
additionally fit a histogram gradient-boosted regressor
(300--500 trees, learning rate $0.05$) with validation-based selection among
network/trees/blend. All figures in this paper are produced directly by the
generation and analysis scripts, with scores embedded in the titles at run
time.

\bibliographystyle{apsrev4-2}
\bibliography{Reference_ml_krylov}

\end{document}